\documentclass[aps,twocolumn,superscriptaddress,floatfix,nofootinbib]{revtex4-1} %linenumbers

\usepackage[export]{adjustbox}
\usepackage{amsmath,amssymb}
\usepackage{braket}
\usepackage{mathtools}
\usepackage{hyperref}
\usepackage{mathrsfs}
\usepackage{dsfont,isomath}
\usepackage{verbatim}
\usepackage{setspace}
\usepackage{ulem}
\usepackage{xcolor}

\usepackage{algorithm}
\usepackage{float}
\usepackage[noend]{algpseudocode}
\usepackage{setspace}

\usepackage{dsfont}
\newcommand{\A}{\mathcal{A}}
\newcommand{\1}{{\widehat{\mathds{1}}}}
\newcommand{\W}{\widehat{W}}
\newcommand{\w}{\widehat{V}}
\newcommand{\sw}{\widehat{w}}
\newcommand{\V}{\widehat{V}}
\newcommand{\Q}{\widehat{Q}}
\renewcommand{\P}{\widehat{P}}
\renewcommand{\AA}{\widehat{A}}
\newcommand{\BB}{\widehat{\v{b}}}%{\widehat{B}}
\newcommand{\CC}{\widehat{\v{c}}}
\newcommand{\DD}{\widehat{d}}
\newcommand{\PP}{\mathbb{P}}

\newcommand{\Z}{\mathbb{Z}}

\newcommand{\R}{\mathbb{R}}

\newcommand{\n}[1]{\left| #1 \right|}%%adjustable-height norm shortcut
\newcommand{\dn}[1]{|| #1 ||}%%adjustable-height norm shortcut
\newcommand{\st}[1]{\left\{#1\right\}}%%adjustable-height set notation
\newcommand{\setc}[2]{\{#1\; :  \; #2 \}}

\renewcommand{\v}[1]{\boldsymbol{#1}}%%shortcut to make a vector (overwrites the default command)
\DeclareMathOperator{\Tr}{Tr}
\DeclareMathOperator{\tr}{tr}
\DeclareMathOperator{\Id}{Id}
\DeclareMathOperator{\diag}{diag}
\DeclareMathOperator{\spec}{spec}
\DeclarePairedDelimiterX{\IP}[2]{\langle}{\rangle}{#1, #2}

\renewcommand{\O}{\widehat{\mathcal{O}}}

\usepackage{pifont}% http://ctan.org/pkg/pifont
\newcommand{\cmark}{\ding{51}}%
\newcommand{\xmark}{\ding{55}}%

\usepackage{amsthm}
\newtheorem{thm}{Theorem}

\newtheorem{prop}[thm]{Proposition}

\newtheorem{lemma}[thm]{Lemma}

\theoremstyle{definition}

\newtheorem{defn}[thm]{Definition}

\theoremstyle{remark}
\newtheorem{eg}[thm]{Example}

\begin{document}
	\title{Local Matrix Product Operators: Canonical Form, Compression, \& Control Theory}

	\author{Daniel E. Parker}
	\email[]{daniel\_parker@berkeley.edu}
	\affiliation{Department of Physics, University of California, Berkeley, CA 94720, USA}
	
	\author{Xiangyu Cao}
	\email[]{xiangyu.cao@berkeley.edu}
	\affiliation{Department of Physics, University of California, Berkeley, CA 94720, USA}
	
	\author{Michael P. Zaletel}
	\email[]{mikezaletel@berkeley.edu}
	\affiliation{Department of Physics, University of California, Berkeley, CA 94720, USA}

	\date{\today}
	
	\begin{abstract}
       
		We present a new method for compressing matrix product operators (MPOs) which represent sums of local terms, such as Hamiltonians. Just as with area law states, such local operators may be fully specified with a small amount of information per site. Standard matrix product state (MPS) tools are ill-suited to this case, due to extensive Schmidt values that coexist with intensive ones, and Jordan blocks in the transfer matrix. We ameliorate these issues by introducing an ``almost Schmidt decomposition" that respects locality.  Our method is ``$\varepsilon$-close" to the accuracy of MPS-based methods for finite MPOs, and extends seamlessly to the thermodynamic limit, where MPS techniques are inapplicable. In the framework of control theory, our method generalizes Kung's algorithm for model order reduction. Several examples are provided, including an all-MPO version of the operator recursion method (Lanczos algorithm) directly in the thermodynamic limit. All results are accompanied by practical algorithms, well-suited for the large MPOs that arise in DMRG for long-range or quasi-2D models.
	\end{abstract}
	\maketitle

	\section{Introduction}

While it is now well understood how matrix product states (MPS) can approximate 1d ground states~\cite{Hastings_2007,schollwock2011density,mcculloch2007density,hauschild2018efficient,laurens1,schuch2008entropy,verstraete2006matrix}, matrix-product representations of operators (MPOs) remain less understood. MPOs feature prominently in modern implementations of the density matrix renormalization group (DMRG)~\cite{schollwock2011density}, yet we lack a complete understanding of the resources required for an MPO approximation of a complex (but local) operator, an important ingredient for several problems of current interest.
 For instance, DMRG calculations of 1d systems with long-ranged interactions or 2d cylinder geometries are hampered by the large bond dimension of MPO representations of the Hamiltonian. Complex operators also arise during the Heisenberg  evolution of simpler ones, so efficient numerical representations would have wide ranging applications in the study of quantum thermalization and the emergence of hydrodynamics.

While a MPO can formally be treated as a MPS in a doubled Hilbert space, this neglects the special structure of operators like Hamiltonians: they  are a sum of local terms, $\widehat{H} = \sum_j \widehat{H}_j$, where $\widehat{H}_j$ is localized around site $j$.
If the standard MPS compression algorithm via Schmidt decomposition (i.e., singular value decomposition) is directly applied to operators, this  structure leads to an ill-conditioned thermodynamic limit, in which some of the Schmidt values become infinite.
In 1d, locality gives rise to the following simple property that is the basis for our results. When a 1d system is partitioned into left and right halves, any local operator can be written as:
\begin{equation}
	\widehat{H} = \widehat{H}_L \otimes \1_R + \1_L \otimes \widehat{H}_R + \sum_a h_{ab} \widehat{h}_L^a \otimes \widehat{h}_R^b \,.
	\label{eq:regular_form}
\end{equation}
where $\widehat{h}_{L/R}^a$ run over traceless operators localized on the left/right halves respectively, with coefficients $h_{ab}$.
The first two terms contain the part of the operator supported on strictly one or the other side of the partition, whose magnitude grows linearly with system size, while the third term contains the terms in the operator straddling the partition.
This immediately suggests a compression scheme: approximate the intensive part $h_{ab}$ using a singular value decomposition (SVD), whose rank will determine the bond dimension of the MPO, while leaving the extensive terms untouched.
Doing so manifestly preserves locality, which will allow us to take the limit of infinite system size,  addressing the long-standing problem of efficiently representing operators in the thermodynamic limit~\cite{michel2010schur,chan2016matrix,hubig2017generic,pirvu2010matrix,zaletel2015time}.
This idea was discussed in Ref.~\cite{chan2016matrix}.
However, the coefficients $h_{ab}$, and the resulting singular value spectrum, depend on the choice of operators $\widehat{h}_{L/R}^a$, and \textit{a priori} there is no reason SVD truncation should be optimal. 
In this work we provide the simple `fix' which makes the procedure optimal: the compression is performed
only after the MPO is brought to a \textit{canonical form} in which $\Tr[\widehat{h}_{L/R}^a \widehat{h}_{L/R}^b] \propto \delta_{ab}$. The main result of this work is an compression algorithm for both finite and infinite MPOs (iMPOs) which works for physical Hamiltonians with virtually any type of interaction.

Canonical forms play a crucial r\^ole in MPS compression and many other algorithms, but the naive generalization of the MPS definition to MPOs fails to capture the locality structure of Eq.~\eqref{eq:regular_form} (for this reason, naive SVD truncation of an MPO in the same manner as MPS generically destroys locality.)
We therefore adapt the MPS technology of ``canonicalization'' and compression algorithms to the class of ``first degree" MPOs, which includes short and long ranged Hamiltonians. As a byproduct, we provide a rigorous analysis of the convergence of well-known iterative ``canonicalization'' algorithms for infinite MPSes. We also present a non-iterative compression algorithm specific to the type of iMPOs that occur in DMRG calculations, which exploits their upper-triangular structure to efficiently handle MPOs with bond dimensions on the order of several thousands.
Finally, we detail an intriguing connection to notions from control theory: our compression scheme is a generalization of Kung's method for model-order reduction via balanced truncation\cite{Kung1978}. Whenever possible, we provide rigorous proofs of our statements. Our results apply to both finite MPOs and infinite matrix product operators, although we put more emphasis on the infinite case.

This work is organized into two parts: the first three sections are a ``practical handbook'' for compressing finite MPOs, followed by a more sophisticated treatment of infinite MPOs. The practical handbook starts with an overview of the key ideas of MPO compression in Section \ref{sec:MPO_compression_idea} and Section \ref{sec:MPO_review} reviews standard facts about MPOs to set notation. We then provide all the concepts and algorithms needed for finite MPO compression in Section \ref{sec:finite_MPOs}, along with a quick numerical example. We then transition to infinite MPOs, which require a somewhat more detailed and mathematical treatment. Section \ref{sec:local_iMPOs} specifies the class of ``first degree'' MPOs our method applies to, and shows their Jordan block structure is completely fixed by locality. Sections \ref{sec:iMPO_canonical_form} is devoted to canonical forms and algorithms to compute them. We give the algorithm for compressing infinite MPOs in Section \ref{sec:iMPO_compression}. Section~\ref{sec:error_bounds} reveals the peculiar structure of the operator entanglement of local MPOs, which we use to show the error from our compression scheme is $\varepsilon$-close to optimal. We also show that the change in the sup norm is small under compression. Section~\ref{sec:relation_to_control_theory} goes on to reinterpret our compression algorithm within control theory. We provide a few examples of iMPO compression in Section~\ref{sec:examples}: compressing operators with long-ranged interactions and computing Lanczos coefficients for operator dynamics. We conclude in Section~\ref{sec:conclusions}. The Appendices prove statements from the main text and describe how all elementary algebra operations can be performed on MPOs.

\tableofcontents

\section{The Idea of Compression}
\label{sec:MPO_compression_idea}

To introduce the key ideas, we first present them on the level of operators, then later translate them into the language of MPOs. Consider a local operator $\widehat{H}$ on $N$ sites. As mentioned in the introduction, we can split the system into left and right halves at some bond, which gives the \textbf{regular form} of an operator
\begin{align}
	\label{eq:H_split}
	\widehat{H} &= \widehat{H}_L \otimes \1_R + \1_L \otimes\widehat{H}_R + \sum_{a,b=1}^\chi  \mathsf{M}_{ab} \widehat{h}_L^a \otimes \widehat{h}^a_R\\
	\nonumber
&=  \begin{pmatrix}
		\1_L & \widehat{\v{h}}_L & \widehat{H}_L
	\end{pmatrix}
	\begin{pmatrix}
		1 & & \\
		& \mathsf{M} & \\
		& & 1
	\end{pmatrix}
	\begin{pmatrix}
		\widehat{H}_R &
		\widehat{\v{h}}_R &
		\1_R
	\end{pmatrix}^T, 
\end{align}
where we have introduced vectors of operators $\widehat{\v{h}}_{L/R}$ on the left and right, and the matrix $\mathsf{M}$ keeps track of the coefficients which straddle the cut. This decomposition is not unique --- we can insert basis transformations to the left / right. So, roughly speaking, we will require \eqref{eq:H_split} be a Schmidt decomposition by ensuring that $\mathsf{M}$ is diagonal and that the components of the vectors are mutually orthogonal. One can then compress $\widehat{H}$ by truncating the Schmidt spectrum  --- but there is a slight wrinkle due to locality. 

To understand the extra structure present in a local operator, let's consider an example. Let 
\begin{equation}
	\widehat{H}_\text{e.g} = \sum_{n=1}^N J \widehat{X}_{n} \widehat{X}_{n+1} + K \widehat{X}_n \widehat{Z}_{n+1} \widehat{X}_{n+2} + h \widehat{Z}_n,
	\label{eq:motivational_example}
\end{equation}
where $\widehat{X}_{n}$ and $\widehat{Z}_{n} $ are operators acting on lattice site $n$. $H_\text{e.g.}$ is a linear combination of strings, such as $\cdots  \otimes \1_{1} \otimes \1_{2} \otimes X_3 \otimes X_{4} \otimes \1_{5} \otimes \1_6 \otimes \cdots$. If we split $\widehat{H}_{e.g.}$ across a bond $n$ in the middle, we can write it in regular form (non-uniquely) as
\begin{equation}
\begin{aligned}
	\widehat{\v{h}}_L &= (
	\widehat{X}_n, \widehat{X}_n,  \widehat{X}_{n-1}\widehat{Z}_n
	)\\
	\widehat{\v{h}}_R &= 
	(
		\widehat{X}_{n+1}, \widehat{Z}_{n+1}\widehat{X}_{n+2}, \widehat{X}_{n+1}
	)\\
	\mathsf{M} &= \diag(J, K, K)\\
	\widehat{H}_L &= \sum_{k=1}^n J \widehat{X}_{k-1} \widehat{X}_k + K \widehat{X}_{k-2} \widehat{Z}_{k-1} \widehat{X}_{k} + h \widehat{Z}_k,
\end{aligned}
\label{eq:example_regular_form_operator}
\end{equation}
and with $\widehat{H}_R$ similar to $\widehat{H}_L$.
We see $H_{L/R}$ differs from the $\widehat{h}_{L/R}$ in two respects: first, it's norm diverges linearly with system size (it is \textit{extensive}) and second, it contains terms arbitrarily far from the partition. So in order for the Schmidt compression to be well defined in the thermodynamic limit and preserve locality, it is eminently reasonable to single out $\widehat{H}_{L/R}$ and treat them separately in a Schmidt decomposition.

	This motivates the generalization and modification of canonical forms and Schmidt decompositions for the case of local operators.
\begin{defn}
	A local operator in regular form Eq. \eqref{eq:H_split}, is in \textbf{left canonical form} if
	\begin{equation}
		\IP{\widehat{h}_L^a}{\widehat{h}_L^b} = \delta^{ab}, \quad 0 \le a,b \le \chi,
		\label{eq:left_canonical_operator}
	\end{equation}
	where $\IP{\widehat{A}}{\widehat{B}} := \Tr[\widehat{A}^\dagger \widehat{B}]/\Tr[\1]$ is the inner-product for operators and $\widehat{h}_L^0 := \1_L$. \textbf{Right canonical form} is the same with $L \leftrightarrow R$.
	\label{defn:operator_canonical_form}
\end{defn}
Notice that we have excluded $\widehat{H}_{L/R}$ from the definition. If an operator is both left canonical and right canonical on a bond, then we can  form the ``almost'' Schmidt decomposition by an SVD decomposition $\mathsf{M} = \mathsf{U} \mathsf{S} \mathsf{V}^\dagger$. 
\begin{defn}
	Suppose $\widehat{H}$ is a local operator and suppose it is both left and right canonical at a bond. Then the \textbf{almost-Schmidt decomposition} of $\widehat{H}$ is 
	\begin{equation}
		\widehat{H} = \widehat{H}_L \otimes \1_R + \1_L \otimes\widehat{H}_R + \sum_{a=1}^\chi s_a \widehat{h}_L^a \otimes \widehat{h}^a_R,
	\label{eq:almost_schmidt_decomposition}
\end{equation}
for some real numbers $s_1 \ge s_2 \ge \cdots \ge s_\chi$.
\label{def:almost_Schmidt}
\end{defn}
This is not a true Schmidt decomposition because we have excluded $\widehat{H}_{L/R}$; $\braket{h_{L/R}^a, H_{L/R}}$ is generically non-zero. This seeming imperfection will actually prove to be a feature, leading to concise algorithms and an truncation error $\varepsilon$-close to optimal with respect to \textit{both} the Frobenius and operator (induced) norms (see Sec. \ref{sec:error_bounds}.) Once we know the almost-Schmidt decomposition of an operator, compressing it to a bond dimension $\chi' < \chi$ is easy: simply restrict the sum in Eq. \eqref{eq:almost_schmidt_decomposition} to run from $1$ to $\chi'$ instead of $\chi$. 
Our task is now to translate this idea from the level of operators to concrete computations and algorithms in the language of MPOs.

\section{Review of MPOs}
\label{sec:MPO_review}

Matrix product operators (MPOs) arise in DMRG as a pithy representation of 1d Hamiltonians. This section will review a few essential facts about finite and infinite MPOs for the reader's convenience and to set notation. The well-known construction  of MPOs comes from viewing a Hamiltonian as a finite-state machine~\cite{crosswhite,schollwock2011density}, which we illustrate with an example.

Consider $\widehat{H}_{e.g.}$ from Eq. \eqref{eq:motivational_example} again. All of the Pauli strings needed to generate $\widehat{H}_{e.g.}$ can be described by a finite state machine, shown in Fig. \ref{fig:MPO_automata}. (We will see below this machine can be improved.) The MPO itself is the adjacency matrix of the finite state machine: 
\begin{equation}
	\W_\text{e.g} = \left(\begin{array}{@{}c|ccc|c@{}}
		\1 & \widehat{X} &\widehat{X} & 0  & h \widehat{Z}\\ \hline
		   &  0          &  0      & 0  & J\widehat{X}\\
		   &  0          &  0      & \widehat{Z}  & 0 \\ 
		   &  0    & 0  & 0 & K \widehat{X} \\ \hline
		  &   &   & & \1\\
	\end{array}
\right) \,,\,
\label{eq:motivational_MPO}
\end{equation}
where the hat on the matrix $\W_\text{e.g}$ indicates that its components are operator-valued. The Hamiltonian on the open chain $[1,N]$ then has the compact representation
\begin{equation}
	\widehat{H}_\text{e.g.} = \v{\ell} \underbrace{\W_\text{e.g.} \W_\text{e.g.} \cdots \W_\text{e.g.}}_{N \text{ matrices}} \v{r},
	\label{eq:motivational_example_MPO_form}
\end{equation}
where $\v{\ell} := (1 \,\; \v{0}_3 \; 0)$ and $\v{r}^\dagger := (0 \,\; \v{0}_3 \; 1)$ are c-number vectors, also called ``boundary conditions''. They encode the instructions ``start at node $i$'' and ``end at node $f$''. The multiplication of MPOs in \eqref{eq:motivational_example_MPO_form} is a matrix product in the \textit{auxiliary space} and a tensor product in the \textit{physical space}, such that physical indices of the $n$th matrix in \eqref{eq:motivational_example_MPO_form} acts on lattice site $n$.

\begin{figure}
	\includegraphics{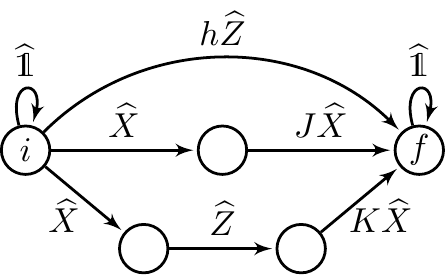}
	\caption{A finite-state machine that generates Eq. \ref{eq:motivational_example}. } 
\label{fig:MPO_automata}
\end{figure}

The example above is a so-called \textbf{infinite MPO (iMPO)}: the whole operator only depends on one matrix $\W$, regardless of the system size. A regular \textbf{MPO} is made of inhomogenous matrices
\begin{equation}
\widehat{H} =  \v{\ell} \W^{(1)}\W^{(2)} \cdots \W^{(N)} \v{r}  \,.
\end{equation}
where $\W^{(1)},  \dots, \W^{(N)}$ are distinct matrices and need not be square, with $\W^{(n)}$ of size $\chi^{(n-1)} \times \chi^{(n)}$ so that matrix multiplication makes sense. 

In a local Hamiltonian, each term begins and ends with strings of identities, which gives rise to the first two terms in the regular form of an operator, Eq. \eqref{eq:H_split} above. This property is encoded by the distingished nodes $i$ and $f$ in the finite state machine Fig. \ref{fig:MPO_automata}, and is reflected by the block structure of the MPO \eqref{eq:motivational_MPO}. We therefore restrict ourselves to a special class of (i)MPOs which manifestly maintain this local structure. 
\begin{defn}
	An (i)MPO is in \textbf{regular form} if each matrix has the block upper triangular structure
	\begin{equation}
		\W
		= \begin{pmatrix}
			\1 & \CC & \DD\\
			0 & \AA & \BB\\
			0 & 0 & \1
		\end{pmatrix} \,,\,
		\label{eq:regular_form_MPO}
	\end{equation}
	where the first and last blocks have dimension $1$ for both rows and columns.\footnote{Structurally, $\DD$ is a single operator, and $\CC$ and $\BB$ are operator-valued vectors.} Furthermore, we require that the boundary conditions are of the form
	\begin{equation}
		\v{\ell} = \begin{pmatrix}
			1 & * & *
		\end{pmatrix},
		\; \v{r}^\dagger = \begin{pmatrix}
			* & * & 1
		\end{pmatrix}
		\label{eq:standard_boundary_conditions}
	\end{equation}
	where $*$ denotes an arbitrary block. 
	\label{defn:regular_form}
\end{defn}
The shape of $\W$ in \eqref{eq:regular_form} is thus entirely determined by the shape of $\AA$. For iMPOs, $\AA$ is a square matrix of size $\chi \times \chi$ where $\chi$ is called the \textit{bond dimension}. (Some authors instead define the bond dimension as the size of $\W$, $\chi+2$.) Operators in regular form are represented by (i)MPOs in regular form, and all (i)MPOs in this work will be in regular form.

The usual diagram notation for tensor networks cannot capture the block structure of \eqref{eq:regular_form_MPO}, so we simply work with equations, making them index-free whenever possible. In the rare exceptions, the auxiliary space is indexed by Latin letters starting from zero to highlight the block structure: $a,b,c \dots = 0 ; 1, 2, \dots \chi; \chi+1$. 

The class of (i)MPOs in regular form is closed under addition, scalar multiplication, and operator multiplication. These constructions are computationally straightforward and more-or-less well-known. They are collected in Appendix \ref{app:elementary_operations} for the reader's convenience. 

Physical operators admit many distinct MPO representations; MPOs have a large \textit{gauge freedom}. An operator $\widehat{H} = \v{\ell} \W^{(1)} \cdots \W^{(N)}\v{r}$ can also be represented by $\widehat{H} = \v{\ell}' \W^{(1)'} \cdots \W^{(N)'} \v{r}'$ whenever there are matrices $L^{(0)},\dots, L^{(N)}$ that satisfy the interlacing conditions
\begin{equation}
%\begin{align}
	\W^{(n)'}L^{(n)} = L^{(n-1)} \W^{(n)}, 
	\v{\ell}' L^{(0)} = \v{\ell},\\
	\v{r}' = L^{(N)} \v{r}.
%\end{align}
\label{eq:gauge_MPO}
\end{equation}
In the infinite case, all the $L^{(n)}$'s are equal to some $L$, so the gauge transformation resembles a similarity transform:
\begin{equation}
	\W'L = L \W.
	\label{eq:gauge_iMPO}
\end{equation}
To preserve the regular form \eqref{eq:regular_form_MPO}, all gauge matrices must be block triangular,
\begin{equation}
    L =  \left( \begin{array}{@{}c|c|c@{}}
				1 & \v{t} & r  \\ \hline
				0 & \mathsf{L} & \v{s} \\ \hline
				0 & 0 & 1
			\end{array} \right). \label{eq:gauge_form}
\end{equation}
Note that $L$ need not be square, but only shaped to be compatible with \eqref{eq:gauge_MPO} or \eqref{eq:gauge_iMPO}\footnote{Some authors define a less general class of invertible gauge transformation $\W' = L \W L^{-1}$, which precludes $L$ from changing the bond dimension.}. In particular, $\W'$ and $\W$ may have different bond dimensions. 

For instance, we can gauge transform $\W_\text{e.g.}$ to
\begin{equation}
	\W'_\text{e.g} = \left(\begin{array}{@{}c|cc|c@{}}
		\1 & \widehat{X} & 0  & h \widehat{Z}\\ \hline
		   &  0               & Z  & J\widehat{X}\\
		   &  0    & 0  & K \widehat{X} \\ \hline
		  &   &   & \1\\
	\end{array}
	\right)
	\label{eq:efficient_motivational_MPO}
\end{equation}
which encodes $\widehat{H}_\text{e.g.}$ more simply than $\W_\text{e.g}$. This previews our end goal: given a MPO (and an error tolerance), how do we compute the smallest MPO that encodes the same operator?

\section{Finite MPO Compression}
\label{sec:finite_MPOs}

Now that we have reviewed MPOs, we give a ``practical handbook'' for compressing finite matrix product operators. We proceed expeditiously: first upgrading canonical forms and ``sweeps'' to MPOs, then giving the compression algorithm, and lastly a brief numerical example. Readers familiar with matrix product states will find that our compression method amount to a small --- yet conceptually significant --- modification of standard MPS algorithms. As the subsequent treatment of iMPOs will revisit all the concepts here in greater detail, many technical details are postponed for later sections.

\subsection{MPO Canonical Forms}
\label{subsec:canonical_form_finite_case}
Just as with matrix product states, the main tool for manipulating matrix product operators is the idea of \textit{canonical forms}. They are choices of gauge that make the rows or columns of the matrix $\W$ orthogonal, an essential step for controlling the errors from compression or carrying out the DMRG algorithm.

We define canonical forms in terms of a condition on the matrix itself, then show that canonical MPOs represent canonical operators. 
\begin{defn}
	An MPO $\widehat{H} = \v{\ell}\W^{(1)}\cdots \W^{(N)}\v{r}$ is in \textbf{left canonical form} if, for each $n >1$, the upper left block of $\W^{(n)}$, 
		\begin{equation}
			\V^{(n)} := 	\begin{pmatrix}
				\1 & \CC^{(n)}\\
				0 & \AA^{(n)}
			\end{pmatrix} \,,\,
		\end{equation}
	    has orthonormal columns: 
		\begin{equation}
			\forall b,c \leq \chi^{(n)},	\sum_{a=0}^{\chi} \IP{\W_{ab}^{(n)}}{\W_{ac}^{(n)}} = \delta_{bc}.
		\end{equation}
		For $n=1$ we instead require $\IP{[\v{\ell} \W^{(1)}]_{b}}{[\v{\ell} \W^{(1)}]_{c}} = \delta_{bc}$ for all $b,c \le \chi^{(1)}$. 
	\label{defn:MPO_left_canonical_form}
	\end{defn}
An MPO is in \textbf{right canonical form} if, and only if, its \textit{mirror}\footnote{A MPO is mirrored by (I) transposing each matrix $\W^{(n)}$, (II) exchanging $\v{\ell}^\dagger \leftrightarrow \v{r}$, (III) reversing all auxiliary indices ($0 \leftrightarrow \chi + 1$, $1,\dots ,\chi \leftrightarrow \chi, \dots, 1$), and (IV) reversing the physical positions.} is in left canonical form.
Right canonical forms are always directly analagous, so we focus on the left-handed case.

Let us now see why left canonical MPOs describe left canonical operators, in the sense of Defn. \ref{defn:operator_canonical_form}.\footnote{Actually the two definitions are entirely equivalent, but we show only one implication for concision.} If we split an MPO in left canonical form at a bond $n$, then we can multiply the matrices together to put the operator into regular form:
\begin{align}
	\widehat{H}_W 
	\ &=\ \left(\v{\ell} \W^{(1)}\cdots \W^{(n)}\right)
	\left( \W^{(n+1)} \cdots \W^{(N)} \v{r} \right) \nonumber \\
	\ &=\ \begin{pmatrix}
		\1_L^{(n)} & \widehat{\v{h}}_L^{(n)} & \widehat{H}_L^{(n)}
	\end{pmatrix}
	\begin{pmatrix}
		\widehat{H}_R &
		\widehat{\v{h}}_R &
		\1_R
	\end{pmatrix}^T. 
\end{align}
Standard form for MPOs implies that the vectors of operators are related by the recursion relation
\begin{equation}
	\begin{pmatrix}
		\1_L^{(n-1)} & \widehat{\v{h}}_L^{(n-1)} 
	\end{pmatrix} \V^{(n)} 
    = \begin{pmatrix}
		\1_L^{(n)} & \widehat{\v{h}}_L^{(n)}. 
	\end{pmatrix}
\end{equation}
If the MPO's are in regular form, then $\V^{(1)},\dots \V^{(n)}$ have orthonormal columns, so by induction,
\begin{align}
	\IP{\widehat{h}_{L,a}}{\widehat{h}_{L,b}} = \delta_{ab},\,  0\le a,b \le \chi^{(n)} \,,
	%\label{eq:left_can_form_physical_op}
\end{align}
where $\widehat{h}_{L,0} := \1_L$. Left canonical form for MPOs therefore ensures that all components but the last of the vector $(\1_L , \widehat{\v{h}}_L , \widehat{H}_L)$ are orthonormal --- and imposes no constraint whatsoever on $\widehat{H}_L$. So MPO canonical form implies operator canonical form, Defn. \ref{defn:operator_canonical_form}.

Now that we have defined canonical forms for MPOs, our next task is compute them. One can always find a gauge transform, Eq. \eqref{eq:gauge_MPO}, to bring a finite MPO to left canonical form and, just as in the MPS situation, we can compute the change of gauge via a QR decomposition. Suppose $\W$ is an MPO in regular form of dimensions $(1+\chi+1)$ by $(1+\chi'+1)$ with $\V$ given by \eqref{eq:regular_form_MPO}. If we group indices as $V_{(\alpha a) b}$, where $0 \le \alpha < d^2$ indexes the standard orthonormal basis of $\A$, then $\V$ can be interpreted as a matrix with shape $d^2(1+\chi) \times (1+\chi')$.
Performing a (thin) QR decomposition gives	
\begin{equation}
		\w = 
		\begin{pmatrix}
		\1 & \CC\\
		0 & \AA
		\end{pmatrix}
		\overset{QR}{=} 
		 {\begin{pmatrix}
			\1 &\CC'\\
			0 & \AA'
		\end{pmatrix} }
		\begin{pmatrix}
			1 & \v{t}\\
			0 & \mathsf{R}
		\end{pmatrix} \,,
	\label{eq:QRdef}
	\end{equation}
	where $\mathsf{R}$ is upper-triangular. 
\begin{defn}
	Define the \textbf{block-respecting $\widehat{QR}$ decomposition} of $\W$ as 
	\begin{equation}
	\widehat{QR}[\W] = \Q R
\end{equation}
with
\begin{equation}
	\Q := \begin{pmatrix}
			\1 & \CC' & \DD\\
			0 & \AA' & \BB\\
			0 & 0 & \1
		\end{pmatrix}, \,
		R := 
		\begin{pmatrix}
			1 & \v{t} & 0\\
			0 & \mathsf{R} & 0\\
			0 & 0 & 1
		\end{pmatrix} 
		\label{eq:block_QR_decomposition}
	\end{equation}
	where the upper-left block comes from \eqref{eq:QRdef}. Therefore, $\Q$ is in left canonical form, and $R$ is upper-triangular. \label{def:QR}
\end{defn}

With this, we can define a sweeping procedure to put a finite MPO into left canonical form.
\begin{align}
	&\v{\ell} \W^{(1)} \W^{(2)}\W^{(3)} \cdots\\
	\stackrel{QR}{=}\ &\v{\ell} \left[\Q^{(1)} R^{(1)}\right] \W^{(2)}\W^{(3)} \cdots\\
	\stackrel{\phantom{QR}}{=}\ &\v{\ell} \Q^{(1)} \left[R^{(1)} \W^{(2)}\right]\W^{(3)} \cdots\\
	\stackrel{QR}{=}\ &\v{\ell} \Q^{(1)} \left[\Q^{(2)}R^{(2)} \right] \W^{(3)} \cdots\\
	\stackrel{\phantom{QR}}{=}\ &\v{\ell} \Q^{(1)} \Q^{(2)}\left[R^{(2)}\W^{(3)}\right] \cdots\\
	\label{eq:explicit_QR_sweep}
\end{align}
By the definitition of the block QR decomposition, the first $1+\chi^{(n)}$ columns of each $\Q^{(n)}$ are indeed orthonormal. Moreover, $\{R^{(1)}, \dots, R^{(N)}\}$ specifies a gauge transform from $\{\v{\ell},W^{(n)},\v{r}\}$ to $\{\v{\ell},Q^{(n)},R^{(N)}\v{r}\}$. We summarize the procedure as Algorithm \ref{alg:left_can_alg_finite}.

\begin{figure}
\begin{algorithm}[H]
	\caption{Left Canonical Form for finite MPOs}
	\label{alg:left_can_alg_finite}
	\begin{algorithmic}[1]
		\Procedure{MPOLeftCan}{$\{\v{\ell},\{\W^{(n)}\}_{n=1}^N ,\v{r}\}$}
		\State $R^{(0)} \gets \v{\ell}$
		\For{$n \in [1,N]$}
			\State $(\Q^{(n)}, R^{(n)}) \gets \widehat{QR}[R^{(n-1)} \W^{(n)}]$
			\Comment{Eq.~\eqref{eq:QRdef}}
		\EndFor
		\State	\textbf{return} $\{\v{\ell}, \{\Q^{(n)}\}_{n=1}^N,  R^{(N)} \v{r}\}, \{R^{(n)}\}$
		\EndProcedure		
	\end{algorithmic}
\end{algorithm}
\end{figure}

Note that Algorithm \ref{alg:left_can_alg_finite} is almost identical to a standard ``right-sweep'' that brings an MPS to its left-canonical form, except that the block-respecting $\widehat{QR}$ decomposition is used \textit{in lieu} of normal QR. 

\subsection{Finite MPO Compression}
\label{subsec:finite_MPO_compression}

We can now give the compression procedure for finite MPOs. Suppose we have a finite MPO on sites $[1,N]$. We first bring the whole chain to right canonical form
\[
	\widehat{H}_W = \v{\ell} \, \W_R^{(1)} \, \W_R^{(2)} \dots \W_R^{(N)} \, \v{r} \,, 
\]
by the mirror of Algorithm~\ref{alg:left_can_alg_finite}. To truncate at bond $(n,n+1)$, we first bring sites $[1,n]$ to left canonical form
\begin{equation}
\begin{aligned}
      &\v{\ell} \, \W_R \, \W_R \, \cdots\,  \W_R \, \W_R  \, \cdots \W_R  \,\v{r} \nonumber \\
    = \,& \v{\ell} \, \W_L  \, {R}  \, \W_R \, \cdots \, \W_R \, \W_R  \cdots  \W_R \, \v{r} \nonumber \\
     \vdots\,\,& \nonumber \\
	 = \,& \, \v{\ell}\, \underbrace{\W_L   \,\W_L \, \cdots  \, \W_L}_{\text{sites } [1,n]}\, R^{(n)} \, \underbrace{{\W_R} \, \cdots \,  \W_R}_{\text{sites } [n+1,N]} \, \v{r}.
\end{aligned}
	 \label{eq:mixed0}
 \end{equation}
(Superscripts have been suppressed for clarity.) The block structure of $R^{(n)}$ is fixed by block QR decomposition, Eq. \eqref{eq:block_QR_decomposition}, and we can always decompose it as\footnote{Here and below, we use the short hand $\mathrm{diag}(1, \mathsf{A}, 1) = A$ for block diagonal matrices, with sans-serif letters for the middle block.}
\begin{equation}
   	R^{(n)}
		= M R' \,,\, 
		M =
		\begin{pmatrix}
			1 & 0 & 0\\
			0 & \mathsf{M} & 0\\
			0 & 0 & 1
		\end{pmatrix}
		R' = \begin{pmatrix}
			1 & \v{t} & 0\\
			0 & \Id_{\chi} & 0  \\
			0 & 0 & 1
		\end{pmatrix}.
		\label{eq:decompose_Ln}
\end{equation}
We then perform an singular value decomposition of $M$ and write
\begin{equation}
   M  = U S V^{\dagger} \,,\, 
   \mathsf{S} = \mathrm{diag}(s_1\ge s_2 \ge \dots \ge s_{\chi}) \,, \label{eq:SVD_C}
\end{equation} 
where the middle blocks are unitary: $\mathsf{U}^\dagger \mathsf{U} = \mathsf{V}^\dagger \mathsf{V} = \Id_{\chi}$. Combining \eqref{eq:mixed0} through \eqref{eq:SVD_C}, we obtain 
\begin{equation}
	\widehat{H}_W =  \cdots \W_{L}^{(n-1)}   \Q^{(n)} S \P^{(n+1)}  \W_R^{(n+1)} \cdots
	\label{eq:mixed_diag}
\end{equation}
where 
\begin{equation}
	\Q^{(n)}  := \W_{L}^{(n)}  U, \quad \P^{(n+1)}  :=   V^\dagger R' \W_{R}^{(n+1)}
\end{equation} 
are still left and right canonical, respectively.\footnote{Right-canonical form is preserved because $R'$ only affects the top row while leaving the bottom $\chi+1$ rows orthonormal, as required for right-canonical form.} 
Therefore Eq.~\eqref{eq:mixed_diag} is left canonical on the left, right canonical on the right, and the central matrix $S$ is diagonal --- so it is an almost-Schmidt decomposition, Eq.~\eqref{eq:almost_schmidt_decomposition}, as desired.

We can now reduce the bond dimension by dropping the smallest singular values, as well as the corresponding columns of $\Q$ and rows of $\P$. The compression scheme is summarized in Algorithm~\ref{alg:MPO_compression}. The truncation is combined with a left-sweep, so the returned MPO is left canonical.

	Due the presence of ``sweeps'' in the algorithms, it is not immediately clear how to generalize them to the infinite case, nor is the precise relation to truncations by ``true'' Schmidt decompositions clear. We will address these points in Sections \ref{sec:iMPO_compression} and \ref{sec:error_bounds} below. We note that our compression scheme is $\varepsilon$-close to optimal, in a sense we make clear below.

\begin{figure}
\begin{algorithm}[H]
	\caption{MPO Compression}
	\label{alg:MPO_compression}
	\setstretch{1.35}
	\begin{algorithmic}[1]
		\Procedure{Compress}{$\{\v{\ell}, \W^{(n)}, \v{r}\},\eta$} \Comment{Cutoff $\eta$}
		\State $\v{\ell}, \{\W_R^{(n)}\} , \v{r} \gets \textsc{RightCan}[\v{\ell},\{  \W^{(n)} \},  \v{r} ]$ 
		\State $R \gets \v{\ell}$
		\For{$n = 1, \dots, N-1$}
		    \State $(\W_L^{(n)}, R) \gets \widehat{QR}[R \W_R^{(n)}]$\Comment{Eq. \eqref{eq:block_QR_decomposition}}
			\State $(M, R') \gets R$ \Comment{ Eq. \eqref{eq:decompose_Ln}}
			\State $(U, S, V^\dagger) \gets \textsc{SVD}[M]$
			\State $\chi' \gets  \max \{a: s_{a} > \eta \}$; $I \gets \st{0,1,\dots,\chi',\chi+1}$.
			%\State $\PP \gets$ Eq. \eqref{eq:projector_matrix}
			\State $\Q^{(n)} \gets [\W_{L}^{(n)}  U]_{0:\chi+1,I}$
			\State $R \gets [V^\dagger R']_{I,0:\chi+1}$
			%\State $\Q^{(n-1)} \gets \Q^{(n-1)} R_\perp U S \PP$
		\EndFor
		\State  $(\Q^{(N)}, R) \gets \widehat{QR}[R \W_R^{(N)}]$
		\State \textbf{return} $ \v{\ell}, \{ \Q^{(n)} \}, 
		 R\v{r}$
		\EndProcedure	
	\end{algorithmic}
\end{algorithm}
\end{figure}

\subsection{An Example}

To demonstrate the utility of our compression scheme, we give a brief numerical example. Specifically, we compress a Hamiltonian with long-ranged interactions and show our method is quite comperable to the standard ``MPS'' compression technique, i.e. treating the operator like an MPS in a doubled Hilbert space. We note, however, that our ``MPO'' compression technique outscales the naive ``MPS'' technique because it contains only intensive values in the entanglement spectrum.

\begin{figure}
	    \centering
	    \includegraphics{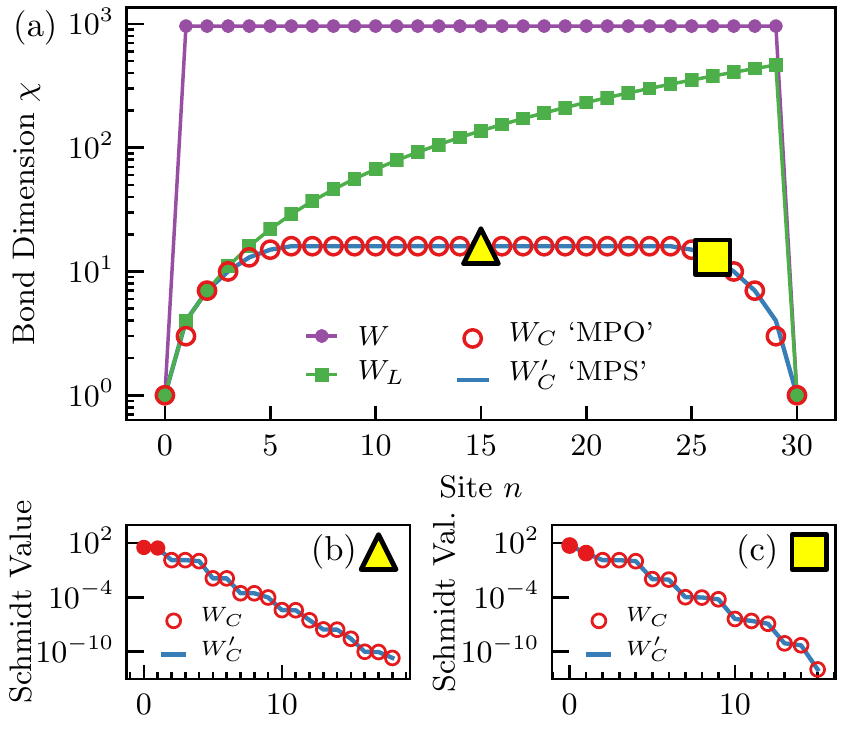}
		\caption{Compression of a finite MPO representing the Hamiltonian~\eqref{eq:H_1_example}. (a) The bond dimensions for: $\W$, the naive MPO representation of $H_1$; $\W_L$, the left-canonical representation by Alg. \ref{alg:left_can_alg_finite}; $\W_C$, the compressed MPO by Alg. \ref{alg:MPO_compression}, and $\W_C'$, the result of the standard MPS compression. (b,c) The Schmidt spectra of $\W_C'$ and almost-Schmidt spectra of $\W_C$ at the sites denoted by the triangle and square, respectively. The numerical precision was taken to be $\varepsilon_{\text{can}} = 10^{-12}$ for canonicalization and $\varepsilon_{C} = 10^{-4}$ for compression.} 
		\label{fig:mpo}
	\end{figure}

	It is well known that a two body interaction $V\left( i-j \right) \O_i \O_j$, where $V(r) = \sum_{j=1}^{\chi} a_j \lambda_j^r$ is a sum of $\chi$ exponentials has an exact MPO representation with bond dimension $\chi$.\footnote{See Eq. \ref{eq:first_degree_example} for an example.} Our algorithm will \textit{automatically} discover this structure even if the MPO is initially presented in a non-optimal form.
	
	We therefore select a more challenging example with power-law interactions:
\begin{equation}
	H_1 = \sum_{k,n,m=1}^N J_{kn} J_{nm} \widehat{Z}_k \widehat{Z}_n \widehat{Z}_m + J'_{nm} \widehat{Z}_n \widehat{Z}_m 
	\label{eq:H_1_example}
\end{equation}
where $J_{nm} = \n{n-m}^{-2}$ and $J'_{nm} = \n{n-m}^{-4}$.  In \eqref{eq:H_1_example} and below, we include a three-body term to test our algorithms beyond the domain of two-body Hamiltonians, which was addressed in previous work~\cite{zaletel2015time}. The results are shown in Fig. \ref{fig:mpo}.

The compression in Fig. \ref{fig:mpo}  follows Algorithm \ref{alg:MPO_compression}, and takes place in two stages. First, a right-sweep with block-QR decomposition (Algorithm \ref{alg:left_can_alg_finite}) performs a preliminary bond reduction: it only reduces bond dimensions if columns are linearly dependent.
%up to a small numerical tolerance. 
%This ``trivial'' compression only removes truly redundant information, such as the matrix elements at site two that encode non-existent terms such as $Z_{-5} Z_{0} Z_2$. 
Then a left-sweep of almost Schmidt value truncation results in a more significant compression. We compare the resulting bond dimensions with those obtained from a standard MPS compression (which does not preserve the block structure) and find them essentially identical. In fact, the whole entanglement spectrum from the almost-Schmidt decomposition closely matches the one from the true Schmidt decomposition. The only difference is the first two Schmidt values are extensive and not present in the almost-Schmidt spectrum.\footnote{Such an precise match of the spectra holds only for simple Hamiltonians; in general, however, we have the interlacing relations~\eqref{eq:interlacing}.}. We return to this point in Section \ref{sec:error_bounds} below, when we discuss operator entanglement.

This concludes our discussion of compressing finite MPOs. We now move on to infinite matrix product operators.

\section{Local Infinite  Matrix Product Operators}
\label{sec:local_iMPOs}

We now transition to infinite matrix product operators. The discussion proceeds analogously to the finite case above. However, working with infinite operators requires additional care, and our discussion will become corresponding more precise and detailed. Indeed, before we can define and compute canonical forms, we must examine exactly what it means for an infinite MPO to be local. We will precisely define and characterize a good class of operators --- operators of ``first degree'' --- which (1) includes local physical Hamiltonians and (2) are described by ``local" iMPOs.

Locality is a non-trivial requirement for a physical operator. It is accompanied by a host of properties, such as an extensive norm, and that spatially-separated terms should commute. For Hamiltonians, perhaps the most important consequence of locality, however, is the existence of thermodynamic limits: the ground state energy and other thermodynamic observables grow as first order polynomials in the size of the system, i.e. extensively. We would like to be able to work with and compress all such local Hamiltonians. As characterizing the class of iMPOs with extensive ground states is quite difficult, we will instead work with a class of operators characterized by an extensive norm, which includes virtually all local physical Hamiltonians. As an analogy, just as local Hamiltonians of interest contribute a constant amount of energy per site, we work with operators that are described be a constant amount of ``information per site". We will often call such operators ``local as iMPOs" or simply ``local".

\subsection{Norm and Transfer Matrices}

The norm of an operator is a starkly different object than that of a state. States, of course, are normalized, so the norm of a generic iMPS should be $1$ in the limit $N \to \infty$. This is rooted in the iMPS transfer matrix, where a standard result~\cite{mcculloch2007density} shows that the largest eigenvalue is non-degenerate with eigenvalue $\lambda=1$, after normalization. In contrast, the space of operators admits many different norms, and this choice must often be resolved by physical considerations. When one is interested in ground state energies and static expectation values, the sup norm is usually the correct choice. However, for questions of quantum dynamics in the common setting of infinite temperature, the Frobenius (aka Hilbert-Schmidt) norm is the natural one, which is relatively easy to compute.

In this work, our ``default" norm will be a Frobenius norm per unit length. For a translation-invariant operator $\widehat{H}$, call its restriction to $N$ sites $\widehat{H}_N$ and define 
\begin{equation}
	\dn{\widehat{H}}_F^2 := \lim_{N\to\infty} \IP{\widehat{H}_N}{\widehat{H}_N} = \lim_{N\to \infty} \frac{\Tr[\widehat{H}_N^\dagger\widehat{H}_N]}{\Tr[\1^N]}.
	\label{eq:extensive_norm}
\end{equation}
where the subscript ``$F$'' is a reminder that this is essentially the Frobenius norm.\footnote{We note that this norm is not submultiplicative: see Appendix D.} The norm is normalized so that $\dn{I}_F = 1$, unlike the usual Frobenius norm where the norm of the identity is the dimension of the space. We will be interested in iMPOs where this norm is extensive. Despite this choice of norm, we prove in Section \ref{sec:error_bounds} that our compression algorithm behaves well with respect to the sup norm as well --- so our choice of norm is suitable for both dynamics and statics applications. We will therefore refer to \eqref{eq:extensive_norm} as \textit{the} norm of an operator in this work. 

To compute the norm of an operator expressed as an iMPO, we must recall the definition of the \textbf{transfer matrix}. The space of single site operators forms an algebra $\A$ with an inner product $\IP{\cdot}{\cdot}$ such that $\IP\1\1 = 1$. We fix an orthonormal basis $\A = \operatorname{span}\setc{\widehat{O}_\alpha}{0 \le \alpha < d}$ (indexed by Greek letters $\alpha, \beta,\dots$) starting with $\widehat{O}_0 = \1$. For example, one might take the algebra of spin-$\tfrac{1}{2}$ operators with the basis of Pauli operators $\{\1, \widehat{X}, \widehat{Y}, \widehat{Z}\}$. Then the real algebra over this basis gives Hermitian operators and the complex algebra gives all operators. 
For Fermions, $\Tr[\widehat{c}^\dagger \widehat{c}] = \Tr[\widehat{n}^\dagger \widehat{n}] = 1$, so one orthonormal basis is $\{\1, \sqrt{2}\widehat{c}^\dagger, \sqrt{2} \widehat{c}, \widehat{Z} = \1 - 2 \widehat{n}\}$ with complex coefficients. In such a single site basis, any operator-valued matrix $\W$ becomes equivalent to an vector of c-number matrices $\{W_\alpha \}$ defined via
\begin{equation} 
\W = \sum_{\alpha} \widehat{O}_\alpha W_\alpha  \,,\quad\, (W_\alpha)_{ab} := \IP{\widehat{O}_\alpha}{\W_{ab}}.
\label{eq:Walpha} \end{equation}

\begin{defn}
	Suppose $\W$ is an operator-valued square matrix that acts on the auxiliary vector space $\mathcal{V}$ of dimension $\chi$. Then the $\W$-\textbf{transfer matrix} is a linear operator on $\mathcal{V} \otimes \mathcal{V}$, defined as
	\begin{equation}
		T_W :=  \sum_\alpha \overline{W}_\alpha \otimes W_\alpha,
		\label{eq:defn_transfer_matrix} \,
	\end{equation}
where the bar denotes complex conjugation. 
\end{defn}It is sometimes convenient to identify $\mathcal{V} \otimes \mathcal{V}$ with the space of square matrices. Then $T_W$ acts on matrices $X \in \mathcal{V} \otimes \mathcal{V}$ on the left by
	\begin{equation} 
      X T_W = \sum_\alpha W_\alpha^{\dagger} X W_\alpha, \label{eq:TW_matrix_notation}
	\end{equation}
	where $W^\dagger_\alpha$ is the Hermitian conjugate as usual. By Choi's Theorem~\cite{choi1975completely}, transfer matrix are always postive operators: whenever $X$ is positive semi-definite, so is $X T_W$. 

The transfer matrix gives a simple formula for the norm of an operator in terms of its MPO representation. On a lattice of $N$ sites, the norm squared is
\begin{equation}
	\dn{\widehat{H}_N}_F^2  =  ({\v{\ell}} \v{\ell}) \left(T_W\right)^N  (\v{r} \v{r}) \,.
	\label{eq:norm_transfer_matrix}
\end{equation} 
where $\v{\ell} \v{\ell} := \overline{\v{\ell}} \otimes \v{\ell}$ and $\v{r} \v{r} := \overline{\v{r}} \otimes \v{r}$.

The only way that \eqref{eq:norm_transfer_matrix} can give rise to an extensive norm, \eqref{eq:extensive_norm}, is if the iMPO transfer matrix $T_W$ \eqref{eq:defn_transfer_matrix} is dominated by some nontrivial Jordan block with eigenvalue $1$.

To build intuition, we first consider the simple example
\begin{equation}
    \widehat{H} = \sum_i \DD_i 
	\text{ with } \W = \begin{pmatrix} \1 & \DD \\ 0 & \1 \end{pmatrix} \,,\, 
\end{equation}
such that $\IP{\1}{\DD}=0$ and $\IP{\DD}{\DD}= \rho$. Of course, $\dn{H_N}_F^2 = N \rho$. Then the transfer matrix $T_W$ is a $4 \times 4$ matrix
\begin{equation}
    T_{W} =
    \begin{pmatrix} 
    1 & 0 & 0 & \rho \\
    0 & 1 & 0 & 0 \\
    0 & 0 & 1 & 0\\
    0 & 0 & 0 & 1
    \end{pmatrix} \,
    \sim 
    \left(
    \begin{array}{cc|c|c}
        1 & \rho & 0 & 0 \\
    0 & 1 & 0 & 0 \\ \hline
    0 & 0 & 1 & 0\\ \hline
    0 & 0 & 0 & 1
    \end{array}
    \right),
    \label{eq:T_W_example}
\end{equation}
where $\sim$ denotes a similarity transform (but not a gauge transform). Taking powers $T_{W}^N$, \eqref{eq:norm_transfer_matrix} shows that the Jordan block is clearly responsible for the extensive norm~\footnote{The other two blocks do not contribute to the extensive norm, but can become relevant when the operator has an extensive trace; see Appendix \ref{app:proofs} for details.}. This behavior should be generic; all local operators should have an extensive norm. However, not all iMPOs in regular form satisfy \eqref{eq:extensive_norm} because, even though such a Jordan block always exists, it may not dominate the norm \eqref{eq:norm_transfer_matrix} as $N \to \infty$. The remedy is to precisely define the what it means for iMPOs to be ``local as an iMPO''.

\subsection{First Degree Operators}
\label{subsec:first_degree_operators}

This section will carefully define the class of first degree operators. Before giving the mathematical definition, let us provide some motiviation.

A natural class of iMPOS which are local by any reasonable criterion are those whose finite state machines do not involve any loops, such as Fig. \ref{fig:MPO_automata}. Such iMPOs represent operators where each term has identities on all sites except on a contiguous block of at most $\chi$ sites. This structure implies that the ground state must be extensive. These operators can be readily characterized as follows:
\begin{defn}
   An iMPO $\W$ is \textbf{strictly local} if its $\AA$ block is strictly upper-triangular. 
\end{defn}
However, this definition has important drawbacks: the property of $\AA$ being strictly upper-triangular is neither  gauge invariant, nor robust under small perturbations  --- which inevitably arise as numerical errors from compression. This definition is therefore an inadequate starting point to define a good class of local operators.

As mentioned earlier, the cure is actually to consider a \textit{larger} class of operators. We will define this class first in terms of the transfer matrix and we will show by the end of the section that these are the operators with extensive norms \eqref{eq:extensive_norm}. Specifically, we make a condition on the spectral properties of the $\AA$ block of their iMPO representation. 
\begin{defn}
	Suppose $\W$ is an iMPO in regular form \eqref{eq:regular_form_MPO}, and $T_A$ is the transfer matrix corresponding to its $\AA$ block. 
	$\W$ is called \textbf{first degree} if $\n{\lambda} < 1$ for all eigenvalues $\lambda$ of $T_A$.\footnote{We note our definition is closely akin to the idea of an ``interaction'' in the mathematical physics literature. See e.g. Chapter 6 of \cite{operator_algebras}.} 
	\label{def:first_degree_operator}
\end{defn}
The name ``first degree'' anticipates Prop. \ref{prop:dominant_jordan_block}, which states that first degree operators have extensive norm $\dn{\widehat{O}_N}_F^2 = O(N)$. Physically, this definition amounts to the requirement that there is a decomposition \eqref{eq:regular_form} where the operators $h^a_{L/R}$ fall off with exponentially-localized tails. 

\begin{table}
	\centering
	\begin{ruledtabular}
		\begin{tabular}{lccc}
			Property	& \textbf{SL} & \textbf{FD} & \textbf{Gen}\\ \colrule
	Leading eig.val. of $T_A$ & $\lambda = 0$ & $\lambda < 1$ & $\lambda < \infty$\\		
	Norm $\dn{H_N}_F^2$ & $\sim N$ & $\sim N$  & $\sim \lambda^{N}$\\
	Open Set & \xmark & \cmark  & \cmark\\	
	Closed under commutation. & \cmark & \xmark & \cmark\\
	Canonical form (see Sec. \ref{sec:iMPO_canonical_form}) & \cmark & \cmark  & \xmark\\	
\end{tabular}
\end{ruledtabular}
	\caption{Properties of different set of iMPOs: strictly local (SL), first degree (FD), and the set of general (Gen) iMPOs without restriction.  }
	\label{tab:locality_hierachy_table}
\end{table}

By Definition~\ref{def:first_degree_operator}, the set of first degree iMPOs is a topologically open set, and is therefore numerically robust, but also a superset of strictly local iMPOs. Indeed, strict locality implies that the $T_A$ matrix is also strictly upper-triangular and thus nilpotent (all $\lambda = 0$). To give an example of an first degree iMPO which is not strictly local, consider
\begin{equation}
	\widehat{H}_{\text{FD}} = \sum_{i} \sum_{k = 0}^{\infty} 
\widehat{X}_i  \left[ \prod_{j=i+1}^{i+k} \alpha \widehat{Z}_{j} \right]  \widehat{Y}_{i+k+1} \,.
\label{eq:first_degree_example}
\end{equation}
whose iMPO representation is
\begin{equation}
	\W_{\text{FD}}
	= \begin{pmatrix}
		\1 & \widehat{X} & 0\\
		0 &	\alpha \widehat{Z} & \widehat{Y}\\
		0 & 0 & \1
	\end{pmatrix} .
\end{equation}
The only eigenvalue of $T_A$ is $\n{\alpha}^2$, so
\begin{equation}
	\dn{\widehat{H}_{FD,N}}_F^2 \sim
	\begin{cases}
		N 					& \n{\alpha} < 1\\
		N^2 					& \n{\alpha} =1\\
		\n{\alpha}^{2N} 	& \n{\alpha} > 1\\
	\end{cases}
\end{equation}
so $\W_{\text{FD}}$ is first degree if and only if $\n{\alpha} < 1$. In this sense, the definition of first degree operators is tight. (Note that $\widehat{H}_\text{FB}$ only has extensive ground state energy for ${\alpha} < 1$.) The spectral nature of the definition also makes the class of first degree iMPOs invariant under iMPO gauge transforms (see Lemma~\ref{lem:quasilocal-gauge} in App. \ref{app:proofs}).

We caution that the class of first degree Hamiltonians is quite vast. It includes all operators that are usually classified as ``local Hamiltonians". For instance, it include all $k$-local Hamiltonians, but also Hamiltonians with long ranged interactions with exponential falloff. 
In general, we expect all physical Hamiltonians are first degree operators, but not all first degree operators are physical Hamiltonians. For example, projectors are first degree operators which do \textit{not} make sense as Hamiltonians. See Appendix \ref{app:elementary_operations} for an example of another non-Hamiltonian first degree operators.

A slight drawback of our broad definition first degree operators is that --- unlike strictly local operators --- they are not closed under commutation (the commutator of two first degree operators can be ``second degree''). Nevertheless, one can show (see Appendix \ref{app:elementary_operations}) that if $\W$ is first degree and $\W'$ strictly local, the commutator $[\W, \W']$ is still first degree. This is sufficient for our applications, including operator dynamics (see Section~\ref{sec:examples} below).

\subsection{The dominant Jordan block of $T_W$}
\label{subsec:dominant_jordan_blocks}

We now show that the transfer matrix of first degree iMPOs have the dominant Jordan block structure required for an extensive norm \eqref{eq:extensive_norm}. From the finite state machine picture, we know that the iMPO always maps the initial state to the initial state, and the final state to the final state. Intuitively, the dominant Jordan block encodes the fact that these are the ``most important processes" in the state machine, rather than running around loops in intermediate states. 

We begin with an intermediate result which will be crucial to establish canonical forms in Section~\ref{sec:iMPO_canonical_form} below. 
\begin{prop}
	Suppose that $\W$ is a first degree iMPO and consider its upper-left block
	\begin{equation}
		\V := \begin{pmatrix}
			\1 & \CC\\
			0 & \AA
		\end{pmatrix}
		\label{eq:V_matrix_defn}
	\end{equation}
	Then the transfer matrix $T_V$ has a unique dominant left eigenvalue of unity with an eigenvector $X$ of the form	
	\begin{equation}
		X T_V = X, \quad 
	    X = \begin{pmatrix}
			1 & \v{x} \\ \v{x}^\dagger & \mathsf{X}
	    \end{pmatrix}.
		\label{eq:T_V_leading_eigenvector}
	\end{equation}
	All other eigenvalues $\lambda$ satisfy $\n{\lambda} < 1$. 
	\label{prop:quasi_local_Tv}
\end{prop}
\begin{proof}
	Since $\V$ has block sizes $(1,\chi)$, the transfer matrix $T_V$ has block sizes  $(1,\chi,\chi, \chi^2)$ in the natural basis.\footnote{Schematically, $(\overline{1} \oplus \overline{\chi}) \otimes (1 \oplus \chi) \cong (1 \oplus \chi \oplus \chi \oplus \chi^2)$.} Moreover, it is block upper-triangular in that basis:
\begin{equation}
T_V =
	 \begin{pmatrix}
	1 & * & *  &  *\\
	0 &  \overline{A}_0 & 0 &  * \\
	0 & 0 & A_0 & * \\
	0 & 0 & 0 &  T_A 
	\end{pmatrix},\; A_0 = \IP{\1}{\AA},
	\label{eq:Tv_block}
\end{equation} so the eigenvalues of $T_V$ are those of the diagonal blocks.

By first degreeness, all eigenvalues $\lambda$ of the $T_A$ block have $\n{\lambda} < 1$. A technical linear algebra fact, Lemma \ref{lemma:A_0_spectrum_constraint} from App. \ref{app:proofs}, shows the same is true for the $A_0$ and $\overline{A}_0$ blocks. The dominant eigenvalue of $T_V$ is therefore $\lambda=1$ from the trivial upper-left block of $T_V$.

To find the eigenvector, we compute $X T_V$, which yields
\begin{equation}
	\begin{pmatrix}
		1 & \v{c}_0 + \v{x} A_0\\
		\v{c}_0^\dagger + A_0^\dagger \v{x}^\dagger
		&
		\sum_{\alpha} \v{c}_\alpha^\dagger \v{c} + \v{c}_\alpha^\dagger \v{x} A_\alpha + A_\alpha^\dagger \v{x}^\dagger \v{c}_\alpha + A_\alpha^\dagger \mathsf{X} A_\alpha
	\end{pmatrix}.
\end{equation}
So $\v{x}$ and $\mathsf{X}$ are determined by
\begin{subequations}
\label{eq:lyapunov}
\begin{align}
	\v{x}[I- A_0] &= \v{c}_0 \\
	\mathsf{X}[\Id - T_A] &= Q\\ 
	Q &:= \sum_\alpha \v{c}_\alpha^\dagger \v{c}_\alpha + \v{c}_\alpha^\dagger \v{x} A_\alpha + A^\dagger_\alpha \v{x}^\dagger \v{c}_\alpha. 
\end{align}
\end{subequations}
As the eigenvalues $\lambda$ of $A_0$ and $T_A$ satisfy $\n{\lambda} < 1$, the operators on the left-hand sides of \eqref{eq:lyapunov} are invertible and solutions $\v{x}$ and $\mathsf{X}$ exist. The dominant eigenvalue therefore has the form \eqref{eq:T_V_leading_eigenvector}.
\end{proof}
Intuitively, in terms of the state machine, the leading eigenvector of $T_V$ is dominated by the ``initial to initial" process. It is worth noting that \eqref{eq:lyapunov} can be written as $Y - \sum_\alpha A_\alpha^\dagger Y A_\alpha = Q$, which is reminiscent of the discrete Lyapunov equation $Y - A^\dagger Y A = Q$ which occurs in control theory. This is a first indication of a nice connection, which we shall detail in Section~\ref{sec:relation_to_control_theory} below.

We now ``enlarge'' the leading eigenvector of $T_V$ to form the dominant Jordan block of $T_W$, which is responsible for the extensive norm, Eq. \eqref{eq:extensive_norm}. 
\begin{prop}
	Suppose $\W$ is an first degree iMPO for $\widehat{H}$ with order-unity trace: $\tr[\widehat{H}] = O(1)$. Then there is a vector $\v{z}$ such that the matrices
\begin{equation}
Z = \begin{pmatrix}
X & \v{z}\\
\v{z}^\dagger & 0
\end{pmatrix},
\text{ and }
Z' = \begin{pmatrix}
0 & 0\\
0 & 1
\end{pmatrix}, \label{eq:Jordan_z}
\end{equation}
[with the same $X$ from Eq. \eqref{eq:T_V_leading_eigenvector}] span the dominant Jordan block of $T_W$:
\begin{equation}
    \begin{pmatrix}
			Z \, T_{W} &
			Z' T_W \\
		\end{pmatrix}
		= \begin{pmatrix}
			Z &
			Z' \\
		\end{pmatrix} \begin{pmatrix}
			1 & \rho \\
			0 & 1
		\end{pmatrix},
		\label{eq:dominant_Jordan_block}
\end{equation}
for some real number $\rho \ge 0$. The norm, Eq. \eqref{eq:extensive_norm}, is then extensive with  $\dn{\widehat{H}_N}_F^2 \to \rho N$ as $N \to \infty$. 
\label{prop:dominant_jordan_block}
\end{prop}
This proposition is easily generalized to traceful operators at the cost of a more complex Jordan block structure. The proof, given in Appendix \ref{app:proofs}, is similar to the one for Prop \ref{prop:quasi_local_Tv}, but somewhat more technical.

We note that $X$, $\v{z}$ and $\rho$ can be calculated from $\W$, but computational tractable formulas use canonical forms, and await us in Sec. \ref{sec:iMPO_canonical_form}. Intuitively, the reason for the extensive norm is that the overlaps of $\v{\ell\ell}$ with $Z$ and $\v{r}\v{r}$ with $Z'$ are both $1$, so if $\W$ is first degree then 
\begin{equation}
    \v{\ell\ell} T_W^N \v{r}\v{r} \sim 
    \begin{pmatrix}
    1 & 0\\
    \end{pmatrix}
\begin{pmatrix}
1 & \rho\\
0 & 1\\
\end{pmatrix}^N
\begin{pmatrix}
0\\
1
\end{pmatrix} = N\rho.
\end{equation}
Therefore first degree operators, as anticipated by their name, have (Frobenius) norm which is a first degree polynomial in $N$.

In summary, we have identified a well-behaved class of local iMPOs --- first degree operators --- that are general enough to contain most operators of interest, and satisfy the physical requirements of an extensive norm. Crucially, first degree iMPOs are qualitatively distinct from generic infinite MPSes: their transfer matrix do not have a unique dominant eigenvalue, but rather a dominant Jordan block (whose eigenvalue is fixed to unity without normalization). Table~\ref{tab:locality_hierachy_table} recapitulates these results. The distinction between a unique dominant eigenvalue versus a Jordan block is of paramount importance as we upgrade canonical forms from states to operators.

\section{Canonical forms for Infinite MPOs}
\label{sec:iMPO_canonical_form}

This section discusses canonical forms for infinite matrix product operators. We first show that canonical forms \textit{exist}: any first degree iMPO admits a choice of gauge that brings it to left canonical form. Actually \textit{computing} such a gauge transform is rather subtle. We first give a general-purpose algorithm, based on QR iteration, with fast convergence for generic iMPOs. Most iMPOs constructed to represent an analytical formula have a special property: they are upper triangular. In this case, canonicalization can be done by an more efficient, iteration free method. We also show that once an operator is in canonical form, it is easy to read off its norm. To our knowledge, canonical forms for operators have not been defined before, perhaps because of the non-trivial first degree requirement.

\subsection{Existence of iMPO Canonical Forms}
\label{subsec:iMPO_can_form_existence}

The definition of canonical form is much the same as in the finite case.
\begin{defn}
		An iMPO $\W$ is in \textbf{left-canonical form}  if its upper-left block $\V$
		has orthonormal columns: $\forall b,c \leq \chi'$, 
		\begin{equation}
			\sum_{a=0}^{\chi} \IP{\W_{ab}}{\W_{ac}} = \delta_{bc}.
			\label{eq:canonical_orthonormal}
		\end{equation}
		\label{defn:left_canonical_form}
	\end{defn}
	An iMPO is in \textbf{right canonical form} if its mirror is left canonical.

Defn. \ref{defn:left_canonical_form}, the definition of iMPO canonical form, is closely related to the MPS case. Precisely, $\W$ is left canonical as an iMPO if, and only if, $\V$ is left canonical \textit{as an MPS}. We can thus import many properties from the case of states. For example, $\eqref{eq:canonical_orthonormal}$ can be written in terms of the transfer matrix (defined in \eqref{eq:TW_matrix_notation}) as 
	\begin{equation}
      \Id_{[0, \chi]} T_V = \sum_\alpha V_\alpha^{\dagger} V_\alpha =\Id_{[0, \chi]}.
	  \label{eq:canonical_transfer}
  \end{equation}
  So $\W$ is left-canonical whenever $\Id_{[0,\chi]}$ is a left eigenvector of $T_V$ with eigenvalue $1$.  This fact is exactly what allows us to prove that canonical forms exist.

  \begin{prop}
	  Let $\W$ be a first degree iMPO. Then there exists a matrix $L$ that which specifies a gauge transform
	  \begin{equation}
		  \W_L L = L \W
	  \end{equation}
	  so that $\W_L$ is left canonical. \label{prop:leftcan_exist}
  \end{prop}

  The proof itself is given in Appendix \ref{app:convergence}, but we briefly outline the idea. Prop. \ref{prop:quasi_local_Tv} tells us that, for any first degree $\W$, the dominant eigenvector of $T_V$ is $X T_V = X$. Suppose that we could take the ``square root decomposition'' $X = K^\dagger K$ with some invertible matrix $K$. Then we could enlarge $K$ to $L = \diag(K \;\, 1)$ and use it as a gauge transform $\W_L = L \W L^{-1}$. Such a $\W_L$ is left-canonical:
  \begin{align*}
	  \Id T_{V_L} &= \sum_{\alpha} (K^{-1})^\dagger \V_{\alpha}^{\dagger} K^\dagger K \V_\alpha K^{-1} \\
	  &= (K^{-1})^\dagger  X K^{-1} = \Id
  \end{align*}
  where $\V_L = K \V K^{-1}$ is the upper-left part of $\W_L$. To turn this into a genuine proof, one must deal carefully with the case when $L$ is \textit{not} invertible --- and this is precisely what we do in Appendix \ref{app:convergence}.

  To demonstrate the utility of canonical forms, we now give a simple formula for the norm of an (i)MPO. For any traceless operator, we can easily ``improve'' the canonical form via the gauge transform
  \begin{equation}
	  L_{lc} := \begin{pmatrix}
		  1 & &\\
		  & I & \v{s}\\
		  & & 1
	  \end{pmatrix},
	   \v{s}  := [A_0 - I]^{-1} \v{b}_0.
	  \label{eq:last_col_id_free_gauge}
  \end{equation}
  (Also see Lemma \ref{lem:traceform}.)  This will gauge away the identity components of the last so that:
\begin{equation}
	\IP{\1}{\DD} = \IP{\1}{\BB_a} = 0,\quad \forall 1\le a \le \chi \,,
	\label{eq:last_col_identity_free}
\end{equation} 
Doing this makes the dominant Jordan block particularly simple.

\begin{prop}
	Suppose $\W$ is an iMPO for $\widehat{H}$ in left-canonical form where \eqref{eq:last_col_identity_free} holds. Then the dominant Jordan block of $\W$ is given by \eqref{eq:Jordan_z} and \eqref{eq:dominant_Jordan_block} with $X = \Id_{[0,\chi]}$, $\v{z}_a = \IP{\AA_{ab}}{\BB_b} + \IP{\CC_a}{\DD}$, and
\begin{equation}
    \lim_{N\to \infty} \dn{H_N}_F^2/N = \rho = \IP{\DD}{\DD} + \sum_{a=1}^\chi \IP{\widehat{b}_a}{\widehat{b}_a}.
    \label{eq:norm_left_can}
\end{equation}
\end{prop}
The proof is immediate from matrix multiplication. In practice, then, one should compute the intensive norm of an iMPO by bringing it to left canonical form, gauging away identities in $\BB$ by \eqref{eq:last_col_Id_free}, and applying \eqref{eq:norm_left_can}. The intuitive reason this works is that, in left canonical form, orthonormality pushes all the weight in each term to the last site (e.g. $0.3 X_1 Y_2 Z_3 \to X_1 Y_2 [0.3 Z_3]$.) The norm is then simply the sums of the squares of the weights of the ending sites. The condition \eqref{eq:last_col_identity_free} ensures that all the edges incident to ``$f$'' in the automata are identity-free, i.e. no terms can ``end prematurely''. 

The finite case is directly analogous. A finite operator $H$ whose MPO is left-canonical with each $\W^{(n)}$ also identity-free in the last column has norm
\begin{equation}
	\dn{H_N}_F^2 = \sum_{n=1}^N \left[\IP{\DD^{(n)}}{\DD^{(n)}} + \sum_{a=1}^{\chi^{(n)}} \IP{\widehat{b}_a^{(n)}}{\widehat{b}_a^{(n)}} \right].
	\label{eq:norm_left_can_finite}
\end{equation}

\subsection{QR Iteration}
\label{sec:QR_iteration}

We now present a general-purpose algorithm to gauge an iMPO $\W$ into left canonical form. Recall that if we can decompose the dominant eigenvector $X T_V = X$ as $X = R^\dagger R$, then $R$ is exactly the gauge transform we need. Any algorithm along these lines must follow the strategy: (I) find $X$, (II) decompose it to find $R$, and (III) deal with the case where $R$ is not invertible. We will see that (I) and (II) are straightforward, but (III) requires considerable care.

Because $X$ is the dominant eigenvector, it is simple to compute using the power method. If $X_{n+1} := X_n T_V$, then $X_n \to X$ as $n\to \infty$. The speed of convergence is controlled by the gap to the second-largest eigenvalue. Unlike in the MPS case, the second-largest eigenvalue is typically far less than $1$, so $X_n$ converges quite fast. We have therefore achieved (I).

To decompose $X$, we need to take the square-root. Simply taking the matrix square-root of $X$ via eigendecomposition or Cholesky decomposition will severely reduce the precision (from $10^{-16}$ to $10^{-8}$ with the standard floating point), which is undesirable. To sidestep this, we use the technique of QR iteration, wherein each application of $T_V$ is performed by taking a QR decomposition. Precisely, let $\W_0 := \W$ and for $n \ge 1$ inductively define
\begin{equation}
		\Q_n R_n := \widehat{QR}[\W_{n-1}],  \quad	\W_n := R_n \W.
\end{equation}
Let $\widetilde{R}_n$ denote the restriction of $R_n =\diag(\widetilde{R}_n \; 1)$ to the upper left blocks (and similarly for $\widetilde{Q}_n$). We have
\begin{equation}
	\sum_{\alpha} \V^\dagger_\alpha \widetilde{R}_{n-1}^\dagger \widetilde{R}_{n-1} \V_{\alpha} 
	= \sum_{\alpha} \widetilde{R}_n^\dagger \left( \widetilde{Q}_{n} \right)^\dagger_\alpha \left( \widetilde{Q}_{n} \right)_\alpha \widetilde{R}_{n},
\end{equation}
so $\left( \widetilde{R}^\dagger_{n-1} \widetilde{R}_{n-1} \right) T_V = \widetilde{R}_n^\dagger \widetilde{R}_n$. This computes the application of the transfer matrix while maintaining the factorized form, giving the limit:
\begin{equation}
	\widetilde{R}_n^\dagger \widetilde{R}_n = X_n \xrightarrow{n\to\infty} X = \widetilde{R}^\dagger \widetilde{R}. 
	\label{eq:iterated_QR_convergence}
\end{equation}
One could then gauge-transform by $R = \diag(\widetilde{R} \; 1)$ as $\W_L R = R W$ to find a left canonical $\W_L$. We have now achieved (II). 

The above procedure is no more than a simple adaption of a well-known standard method in the iMPS context~\cite{laurens1}, and suffices to compute canonical forms for generic iMPOs. However there are many reasonable iMPOs for which it fails badly (we will encounter them in the application discussed in Section~\ref{sec:examples}, Fig.~\ref{fig:lanczos} below). The essential problem is that convergence $X_n \to X$ does \textit{not} guarentee $\widetilde{R}_n \to \widetilde{R}$, especially when $X$ is a singular matrix. 
This is the main obstruction to achiving (III). 

\begin{figure}
	\begin{algorithm}[H]
		\caption{iMPO Left Can. Form: Iterated QR}
		\label{alg:left_can_alg_infinite_iterated_QR}
		\setstretch{1.2}
		\begin{algorithmic}[1]
			\Procedure{LeftCanQRIter}{$\W, \eta$} \Comment{$\eta$: desired precision}
			\State $L \gets \Id_{[0,\chi+1]}$
			\State $\varepsilon \gets \infty$\Comment{Current error}
			\While{$\varepsilon > \eta$} \Comment{Repeat until convergence}
				\State $(\Q, R) \gets \widehat{QR}(\W)$ \Comment{Eq.~\eqref{eq:QRdef}}
				\State $\W \gets R \Q$
				\State $L \gets R \, L$ 
				\State $\varepsilon \gets \dn{R - \Id}$ \textbf{if} $R$ is square \textbf{else }$\infty$
			\EndWhile
			\State \textbf{return} $\Q, L$
			\EndProcedure		
		\end{algorithmic}
	\end{algorithm}
\end{figure}
Algorithm \ref{alg:left_can_alg_infinite_iterated_QR} presents the ``practical solution'' to this conundrum. The idea is to apply a gauge transformation after every QR step, i.e.:
$$ \W_0 = \Q_1 R_1 \,,\, \W_1 = R_1 \Q_1 \,,\, 
\W_1 = \Q_2 R_2 \,,\, \W_2 = R_2 \Q_2 \dots $$
Then $\W_1$ is related to $\W_0$ by a gauge transform $R_1 \W_0 = \W_1 R_1$, and $\W_2$ to $\W_0$ by $ R_2 R_1 \W_0 = \W_1 R_2 R_1$, etc. The desired gauge transform to a canonical form will be approached by the product $L_n =  R_n R_{n-1} \dots R_1$. An important advantage of this method comes from bond dimension reduction: to see this, suppose that $\W_0$ has bond dimension $\chi_0$ but linearly dependent columns, so that $\Q_1, R_1$ can have shape $(\chi_0+2) \times (\chi_1+2), (\chi_1+2) \times (\chi_0+2)$ respectively, with $\chi_1 < \chi_0$.\footnote{This is known as ``rank-revealing'' QR, and can be done by removing vanishing rows of $R$ and the corresponding columns of $\Q$ after running some standard QR routine, for example.} As a result, $\W_1$ will have a smaller bond dimension $\chi_1$. Thus, the first few iterations will reduce the bond dimension of $\W$. Eventually, the bond dimension will stabilize, and $R_n$ will become a square matrix, and invertible in most situations, thereby ameliorating the problem (III).

Unfortunately,  there are still pathological cases where this algorithm will fail as well, but it gives a good balance between speed, applicability, and ease-of-implementation. Appendix \ref{app:convergence} proves the conditions under which Alg. \ref{alg:left_can_alg_infinite_iterated_QR} converges, supplies non-converging counterexamples, and a more complex algorithm which we prove \textit{always} converges (Algorithm~\ref{alg:general_left_canonical_form}). We reiterate that Algorithm \ref{alg:left_can_alg_infinite_iterated_QR} will work almost always in practice, and the fool-proof algorithm is only used to handle rare exceptions.

We remark that the above discussion on iMPO canonical forms (including Appendix~\ref{app:convergence}) can also be regarded as a careful treatment of iMPS canonical forms. To our knowledge, the subtlety involved in the convergence of QR iteration has not been thoroughly discussed previously, since it appears that the matrices encountered in iMPS calculations are always in a generic class for which any QR iteration scheme converges.

\subsection{Upper Triangular Algorithm}
\label{subsec:UT_algorithm}

\label{sec:iMPOcan_utri_algorithm}

When an iMPO is an upper-triangular operator-valued matrix --- as is often the case when MPOs are constructed to represent an analytical Hamiltonian --- it is possible to put it into canonical form with a non-iterative algorithm. In some sense, algorithms for canonical forms are a generalization of the Gram-Schmidt algorithm, where elementary row- and column-operations are replaced by gauge transforms. In the upper-triangular case, however, gauge transformations are so close to elementary row/column operations that we can adapt Gram-Schmidt directly. The result is a non-iterative algorithm that uses an upper-triangular solver to compute the gauge transform one column at a time. 

Suppose we have an upper-triangular MPO
\begin{equation}
	\W_{M-1} = \begin{pmatrix}
		\1 & \vrule & \vrule & \vrule & \cdots\\[0.2em]
		& \v{\sw}_1 & \vrule & \vrule & \cdots\\[0.2em]
		& & \v{\sw}_2 & \vrule & \cdots\\[0.2em]
		& & & \v{\sw}_3 & \cdots\\[0.2em]
		& & & & \ddots
	\end{pmatrix}.
\end{equation}
and assume, for induction, that the first $M$ column vectors $\v{\sw}_0,\cdots\v{\sw}_{M-1}$ are already orthonormal. We want to modify $\v{\sw}_M \to \v{\sw}_{M}'$ to be orthogonal to all previous columns. To do this, we apply a gauge transformation which is the identity except for the $M$th column:
\begin{equation}
	R_{M} = \begin{pmatrix}
		1 & 0 &  &r_0 & &\\
		& \ddots & & \vdots & &\\
		& & 1 & r_{M-1} & &\\
		& & & s_M & &\\
		& & & & \ddots &\\
		& & & & & 1
	\end{pmatrix}.
	\label{eq:gauge_transformation_column_M}
\end{equation}
The transformation $\W_M = R_{M} \W_{M-1} R_{M}^{-1}$ is then easily computed\footnote{The inverse $R_M^{-1}$ has the same form as $R_M$ but with $r_a \to -r_a$ and $s_M \to 1/s_M$.} and maintains the upper-triangular form, while only affecting columns $M$ and beyond. In particular, setting $s_M =1$ temporarily,
\begin{equation}
	\v{\sw}_M' = \v{\sw}_M - \sum_{a=0}^{M-1} r_a \v{\sw}_a + \sum_{a=0}^{M-1} r_a \widehat{d}_M \v{e}_a
	\label{eq:new_column_M}
\end{equation}
where
$\v{e}_a$ is the standard basis vector $\left( \v{e}_a \right)_b = \delta_{ba}$ and $\widehat{d}_M := \left( \v{\sw}_M \right)_M = \W_{MM}$ is the diagonal component of the $M$th column. In Gram-Schmidt, the last term is absent, and one would simply set $r_b = \IP{\v{\sw}_b}{\v{\sw}_M}$ to orthogonalize the columns. We need only make a slight modification to account for the last term.

Orthogonality against column $b < M$ is the condition
\begin{equation}
	0
\equiv \braket{\v{\sw}_b, \v{\sw}_M}
+ \sum_{a=0}^{M-1} \left( -\braket{\v{\sw}_b,\v{\sw}_a} + \braket{ \v{\sw}_b, \widehat{d}_M \v{e}_a} \right)r_a.
	\label{eq:upper_triangular_orthogonality_column_transform}
\end{equation}
This is just a linear equation $K \v{r}  = \v{c}$ where
\begin{subequations}
\begin{align}
	K_{ba} \ &=\ \delta_{ba} - \braket{\W_{ba}, \W_{MM}}\\
	c_b \ &=\ \braket{\v{\sw}_b, \v{\sw}_M},
\end{align}
\label{eq:orthogonality_linear_eqn}%
\end{subequations}
the Kronecker-$\delta$ comes from the induction hypothesis $\braket{\v{\sw}_b, \v{\sw}_a} = \delta_{ba}$, and $K$ is lower-triangular. Therefore we can easily solve for $\v{r} = K^{-1} \v{c}$ by back-substitution to find the $r_b$'s, giving an $\v{\sw}_M'$ orthogonal to previous columns. We can use the final free parameter, $s_M$, to normalize. The effect of $s_M$ on column $M$ is 
\begin{equation}
	\v{\sw}_M' \to \v{\sw}_M^{''} = \frac{1}{s_M} (\v{\sw}_M' - \widehat{d}_M \v{e}_M ) +  \widehat{d}_M \v{e}_M ,
\end{equation}
The normalization condition $1 \equiv \braket{\v{\sw}_N^{''}, \v{\sw}_N^{''}}$ implies
\begin{equation}
	s_M =\sqrt{\frac{
	\braket{\v{\sw}_M, \v{\sw}_M}
}{
	1 - \braket{\widehat{d}_M, \widehat{d}_M}
}}.
\label{eq:column_normalization}
\end{equation}
The first order condition ensures the denominator is non-zero.

We have thus solved for the gauge transformation $R_M$ to orthonormalize column $M$ against the previous columns. Of course, this gauge will modify the columns beyond $M$, but those are treated in subsequent steps. The procedure is summarized in Algorithm \ref{alg:left_can_form_triangular}  and has a total cost of $O(\chi^3)$ operations. In each loop, we perform a triangular solve and a matrix multiplication. The triangular solve costs $O(\chi^2)$ and, since $R$ is almost the identity matrix, we can apply it in time $O(\chi^2)$ as well. With the outer loop of size $\chi$, we have a total cost of $O(\chi^3)$. 

\begin{algorithm}[H]
	\caption{iMPO Left Can. Form: Triangular}
	\label{alg:left_can_form_triangular}
	\setstretch{1.35}
	\begin{algorithmic}[1]
		\Require
			$\W$ upper-triangular
		\Procedure{LeftCanTriangular}{$\W$}
			\State $R_T \gets I_{1+\chi+1}$
			\For{$M \in [1,\chi]$}
			\State $K_{ba} = \delta_{ba} - \braket{\W_{ab}^\dagger, \W_{MM}},\quad m,k \in [0,M-1]$	
			\State $c_b = \sum_{a=0}^{M-1} \braket{\W_{bM}^\dagger,\W_{aM}}, \quad m \in [0,M-1]$
			\State $\v{r} \gets K^{-1} \v{c}$ \Comment{$O(\chi^2)$ triangular solve}
			\State $R \gets \Id_{1+\chi+1}$, $R_{bM} \gets r_b,\quad m \in [0,M-1]$
			\State $\W \gets R \W R^{-1}$, $R_T \gets R R_T$ \Comment{only $O(\chi^2)$}
			\State $s \gets $ Eq. \eqref{eq:column_normalization}
			\State $R \gets \Id_{1+\chi+1}$, $R_{MM} \gets s$
	\State $\W \gets R \W R^{-1}$, $R_T \gets L R_T$ \Comment{only $O(\chi)$}
			\EndFor
		\State \textbf{return} $\W, R_T$
		\EndProcedure	
	\end{algorithmic}
\end{algorithm}

Several remarks are in order. First, this algorithm has an easily-curable instability, which arises when $s_M$ in \eqref{eq:column_normalization} is vanishingly small. This means $ \sw_{M}' - \widehat{d}_M \v{e}_M$ is also vanishing. Consequently, in terms of the state machine, the $M$th state cannot be reached from the initial state, so one should simply discard the $M$th row and column of $\W$ (as well as the $M$th row of the gauge matrix), and carry on.

Second, most upper-triangular MPOs encountered in practice in DMRG have have no diagonal components. In this case, $\W$ is often \textit{strictly} upper triangular, whereupon $K=I$, the linear system becomes trivial, and the algorithm essentially reduces to normal Gram-Schmidt.\footnote{We caution that it is still necessary to row-transform the $j>M$ columns at each step, so the algorithm is distinct from a simple QR factorization.} 

Third, this algorithm is easily generalized to the case of extended unit cells, which occur frequently in applications to 2d DMRG. Let us sketch how this extension works. An iMPO with an extended unit cell with $N$ sites is composed of repeating blocks $[\W^{(1)}\W^{(2)}\cdots \W^{(N)}]$. Gauge transforms are now collections of matrices $R_1,\dots, R_N$ which satisfy intertwining relations $R_{n-1} \W^{(n)'} = \W^{(n)} R_n$ for $n \in \Z/N\Z$. So we must now carry out the algorithm where the matrices $R_n$ and $R_{n-1}$ on each side are not the same. After a gauge transformation for column $M$ using $R_n$'s with the same form as \eqref{eq:gauge_transformation_column_M},
\begin{equation}
	\v{\sw}_M^{(n)'}
	= \v{\sw}_M^{(n)} - \sum_{a=0}^{M-1} r_a^{(n)} \v{\sw}_a + r_a^{(n-1)} \widehat{d}_M^{(n)} \v{e}_c.
\end{equation}

The new orthogonality condition is the linear equation $\mathcal{M} \mathcal{R} = \mathcal{C}$ where $\mathcal{M}$ is a $\chi\times \chi$ lower-triangular block matrix where each block is $N\times N$:

\begin{subequations}
    \begin{align}
    \mathcal{M} &:=
    	\begin{pmatrix}
		M_{11} & & & \\
		M_{21} & M_{22} & &\\	
		\vdots & \ddots & \ddots &\\
		M_{\chi 1} & \cdots & \cdots & M_{\chi\chi} 
	\end{pmatrix}\\ 
	M_{ab} &:=
	\begin{pmatrix}
		1 & & & M^{(1)}_{ab}\\
		M_{ab}^{(2)} & \ddots & &\\	
		& \ddots & 1 &\\
		& & M_{ab}^{(N-1)} & 1
	\end{pmatrix}\\
	\mathcal{C} &:= \begin{pmatrix}
		\v{c}_1 & \v{c}_2 & \cdots & \v{c}_\chi
	\end{pmatrix}^T\\
	\v{c}_a &:= \begin{pmatrix}
		\v{c}_a^{(1)} & \v{c}_a^{(2)} & \cdots & \v{c}_a^{(N)}
	\end{pmatrix}^T\\
	M^{(n)}_{ba} &:= -\braket{\W_{ab}^{(n)}, \W^{(n)}_{MM}}\\
	\v{c}^{(n)}_a &:= \braket{\sw_a^{(n)}, \sw_M^{(n)}}.
    \end{align}
\end{subequations}
Again, when $\W^{(n)}$ are all strictly upper triangular, the system is trivial and $\mathcal{R} = \mathcal{C}$. In general, however, this linear system is solvable in $O(N\chi^2)$ operations by exploiting the special structure of $\mathcal{M}$. Specifically, as each $M_{ab}$ is almost tridiagonal, one may solve $M_{ab} \v{x} = \v{c}$ in $O(N)$ with a combination of forwards and backwards substitution.\footnote{In particular, let $\alpha_k$ and $\beta_k$ be such that $x_k  = \alpha_k x_1  + \beta_k$. Put $(\alpha_1, \beta_1) := (1,0)$ and recursively compute  $\alpha_{k+1} = - (M_{ab}^{(k)}/1) \alpha_k$, $\beta_{k+1} = \beta_{k} + (c_k/1)$. Then $x_1 = \beta_{N+1}/(1-\alpha_{N+1})$ and the other $x_k$'s follow from $x_k = \alpha_k x_1 + \beta_k$.} This allows $\mathcal{M}$ to be solved by forwards substitution as
\begin{subequations}
\begin{align}
    	\mathcal{R} &:= \begin{pmatrix}
		\v{r}_1 & \v{r}_2 & \cdots & \v{r}_\chi
	\end{pmatrix}^T\\
	\v{r}_a &= M_{aa}^{-1}\Big[ \v{c}_a - \sum_{b=1}^{a-1} M_{ab} \v{r}_b \Big].
\end{align}
One can thus solve for $\mathcal{R}$ in $O(N\chi^2)$ operations.

\end{subequations}

The other parts of the algorithm are simple to adapt, and the total cost to find the left-canonical form is $\mathrm{O}(N \chi^3)$, linear in the unit cell size $N$. This is a highly practical algorithm for compressing the iMPOs that appear in 2d DMRG.

In this section we have shown that first degree iMPOs can always be brought to canonical forms. We then gave two algorithms for computing them, one which converges well for almost all local iMPOs, and one which is specialized to upper-triangular iMPOs. Appendix \ref{app:convergence} gives a yet-more-general algorithm, which is guarenteed to converge for \textit{all} first degree iMPOs. We now proceed to compression of infinite MPOs which, unlike canonicalization, hews closely to the finite case.

\section{Compression of iMPOs}
\label{sec:iMPO_compression}

We now explain how to compress infinite MPOs. The algorithm is directly analagous to the finite case: use canonical forms to make an almost-Schmidt decomposition of the operator, then truncate the almost-Schmidt values. Subsequently, Section \ref{sec:error_bounds} will show it is virtually optimal by bounding its error and Section \ref{sec:relation_to_control_theory} will link operator compression to problems in control theory. 

Suppose $\W_R$ is an iMPO in right canonical form. 
Using the gauge from Lemma \ref{lem:traceform} we may impose $c_0 = \braket{\1,\CC} = 0$ without loss of generality.\footnote{Actually we only need the $\v{t}$ part and set $\v{s} = 0$. in Eq. \eqref{eq:last_col_Id_free}.} There is then a gauge transform between right and left canonical form,
\begin{equation}
    C \W_R = \W_L C,   \label{eq:iMPOgauge}
\end{equation}
and $c_0 = 0$ implies $C = \diag(1 \; \mathsf{C} \; 1)$ is block-diagonal. (To ease bookkeeping, we treat $\W_R$ and $\W_L$ as square matrices of the same dimension, though the algorithm works equally well for non-square iMPOs.) The SVD of $C = USV^\dagger$, now implies 
	 \begin{equation}
    U S V^{\dagger} \W_R = \W_L 
    U S V^{\dagger} \,, \label{eq:gauge_svd}
\end{equation}  
where $U$ and $V$ are unitary. Therefore, we can use them to gauge transform $\W_{L,R}$ into
\begin{equation}
\widehat{Q} := U^{\dagger} \W_L U \text{ and }
\widehat{P} :=V^{\dagger} \W_R  V \,,
\end{equation} 
which are left and right canonical, respectively. Furthermore, \eqref{eq:gauge_svd} implies that they are related by the gauge transform
\begin{equation}
    \widehat{Q} S = S \widehat{P}  \,. \label{eq:QSisSP}
\end{equation}
Consequently, we obtain a \textbf{mixed canonical form} for the iMPO:
\begin{align}
    \widehat{H}_W = &\cdots \W_R  \W_R  \W_R \W_R \cdots \nonumber \\
      = &\cdots \W_L  \W_L C  \W_R  \W_R \cdots \nonumber \\
      = & \cdots \W_L  \W_L U S V^{\dagger}  \W_R  \W_R \cdots \nonumber \\
      =& \cdots  \widehat{Q}  \widehat{Q}  S  \widehat{P}  \widehat{P} \cdots  \,. \label{eq:QSP}
\end{align}
In the second line above, we inserted an $L$ matrix at $-\infty$ and moved to the center using \eqref{eq:iMPOgauge}; in the fourth line, $U$ and $V$ are moved to $-\infty$ and $+\infty$ respectively\footnote{These operations incur $\mathrm{O}(1)$ errors near the boundary, which are negligible for an iMPO.}.

Compression of iMPOs must be done on all bonds simultaneously and self-consistently, otherwise errors are incurred even when the compression is exact. To ensure this self-consistency, suppose for now that only $\chi' < \chi$ singular values are non-vanishing.\footnote{In this case, the optimal compression error is zero, but the procedure itself is identical to the case where the singular values are numerically small.} Then
\begin{equation}
	S = \PP \PP^\dagger S = S \PP \PP^\dagger = \PP S' \P^\dagger
	\label{eq:projector_conjuration}
\end{equation}
where $\PP$ is the projection matrix to the first $\chi'$ indices in the middle block
\begin{equation}
	\PP_{ab} = \begin{cases}
		\delta_{ab} & a \in \st{0,1,\dots,\chi',\chi+1}\\
		0 & \text{ otherwise,}
	\end{cases}
	\label{eq:projector_matrix}
\end{equation}
and $\mathsf{S}' = \diag(s_1,\dots,s_{\chi'})$. We can then us the fact that, in mixed canonical form, the position of $S$ can be freely translated to any site using \eqref{eq:QSisSP}. We can then use \eqref{eq:projector_conjuration} to ``conjure'' up projectors at every bond of \eqref{eq:QSP}:
\begin{align}
	H_W &= \cdots \Q \Q S \P \P \cdots\\ 
	&= 
	\cdots  \PP \PP^\dagger \widehat{Q} \PP  \PP^\dagger \widehat{Q}   \PP S' \PP^\dagger \widehat{P} \PP \PP^\dagger \widehat{P} \PP \PP^\dagger \cdots \nonumber \\
%    &= \left.\cdots  \PP \right) \left( \PP^\dagger \widehat{Q} \PP)  (\PP^\dagger \widehat{Q}   \PP \right) \widetilde{S'} \left(\PP^\dagger \widehat{P} \PP\right) \left(\PP^\dagger \widehat{P} \PP \right) \left( \PP^\dagger \cdots \nonumber \right.\\
    &=  \cdots \widehat{Q}' \widehat{Q}'  S'  \widehat{P}'  \widehat{P}' \cdots 
\end{align}
where $\widehat{Q}' =  \PP^\dagger \widehat{Q} \PP$ and $\widehat{P}' =  
\PP^\dagger \widehat{P} \PP$ now have bond dimension $\chi'$. Either $\widehat{Q}'$ or $\widehat{P}'$ can be returned as a compression of the original iMPO; one may make a choice keeping in mind that $ \widehat{Q}'$ and $ \widehat{P}'$ are approximately left and right canonical, respectively. Since we have assumed that the singular values beyond $\chi'$ vanish exactly, this is an exact compression. When this is not true, there will be some finite error (see Sec. \ref{sec:error_bounds})  but the procedure is unchanged. Algorithm \ref{alg:iMPO_compression} gives an implementation, which we reiterate works also for non-square matrices.

\begin{figure}
\begin{algorithm}[H]
	\caption{iMPO Compression}
	\label{alg:iMPO_compression}
	\setstretch{1.35}
	\begin{algorithmic}[1]
		\Procedure{iCompress}{$\W,\eta$} \Comment{Cutoff $\eta$}
		\State $\W_R \gets \textsc{RightCan}[\W]$ 
		\State $\W_R \gets R \W_R R^{-1}$ so that $\CC_0 = 0$  \Comment{Use $\v{t}$ from Lem. ~\eqref{lem:traceform}}
		%\State $\left( \P, R^{-1} \right) \gets \textsc{RightCan}[\Q]$
		%\State $(R_\perp, M) \gets R^{-1}$ \Comment{by Eq. %\eqref{eq:R_matrix_decomposition}}
	    \State $\W_L, C \gets
	    \textsc{LeftCan}[\W_R] $
		\State $(U, S, V^\dagger) \gets \textsc{SVD}[C]$
		\State $\Q, \P \gets U^{\dagger} \W_L U \,,\, V^{\dagger} \W_R  V $
		\State $\chi' \gets  \max \{a \in [1,\chi]: s_{a} > \eta \}$ \Comment{Defines $\PP$ \eqref{eq:projector_matrix}}
		\State $\Q, S, \P \gets \PP^\dagger \Q  \PP, , \PP^\dagger S \PP , \PP^\dagger \P  \PP$
		\State \textbf{return} $\P$ \Comment{One could also return $\Q$.}
		\EndProcedure	
	\end{algorithmic}
\end{algorithm}
\end{figure}

\section{Operator Entanglement and Error Bounds}
\label{sec:error_bounds}

In this section we discuss the error resulting from compressing an operator. The first stage in our analysis will be to show that, just as the singular values of an MPS are closely related to the entanglement, the almost-Schmidt values of an MPO are closely related to the \textit{operator} entanglement entropy. We will immediately apply this relation to answer a practical question: how accurate is our compression algorithm? We will derive a quantitative bound on the error and show the algorithm is $\epsilon$-close to optimal. Finally, we will show that the change in the sup norm is small under compression and hence our compression algorithm is suitable to use when finding ground states.

\subsection{Relation to Operator Entanglement}

To assess the accuracy of our MPO compression scheme, we require a point of comparison. For this, we recall that all MPO's can be thought of as (non-injective) MPSes, and can be compressed via the true Schmidt decomposition. We will refer to this as the ``MPS'' compression method. For iMPOs, the iMPS method will simply fail, due to the Jordan block structure and the reasons detailed in Section \ref{sec:iMPO_canonical_form}, as well as below, so our compression scheme has no obvious competitor in the infinite case. On a finite chain, however, both methods are valid, and it is meaningful to compare the MPO and ``MPS'' methods.

It is well-known that the matrix product compression of a state is intimately related to its bipartite entanglement spectrum. The same notion can be defined for an operator $\widehat{H}$ viewed as a state. If we consider a finite chain $[1, N]$ and make an entanglement cut on bond $(n,n+1)$, then the (true) operator \textbf{Schmidt decomposition} is
\begin{equation}
	\widehat{H} = \sum_{a=-1}^\chi \lambda_a \O_L^a \otimes \O_R^a \,,\, \Tr[\O_L^{a\dagger} \O_L^b] = \delta^{ab} 
	\label{eq:operator_schmidt_decomposition}
\end{equation}
(and the same for $R$), where the $\O_L$'s and $\O_R$'s act only on the left or right of the cut respectively. The Schmidt values $\lambda_{-1} \ge \lambda_0 \ge \cdots \lambda_{\chi} > 0$ are unique and positive.\footnote{The irregular index convention for the $\lambda_a$'s will prove convenient below.} Note that we do \textit{not} normalize $\sum_a \lambda_a^2$ to unity. 

The reason the MPO compression scheme works is the close, quantitative, resemblance between the almost-Schmidt decomposition, Eq. \eqref{eq:almost_schmidt_decomposition}, and the true Schmidt decomposition, Eq. \eqref{eq:operator_schmidt_decomposition}. To see this, we start with the almost-Schmidt decomposition and convert it to the true one. Suppose we have an almost-Schmidt decomposition (Definition~\ref{def:almost_Schmidt}):
	\begin{align}
	\widehat{H}
	&= \widehat{H}_L \otimes \1_R + \1_L \otimes \widehat{H}_R + \sum_a s_a \widehat{h}_L^a \otimes \widehat{h}^a_R \nonumber \\
 &=	\begin{pmatrix}
		\1_L & \widehat{\v{h}}_L & \widehat{H}_L
	\end{pmatrix}
	\begin{pmatrix}
	        1 & & \\  & \mathsf{S} & \\ & & 1
	\end{pmatrix}\begin{pmatrix}
		\widehat{H}_R & \widehat{\v{h}}_R & \1_R
	\end{pmatrix}^T. \label{eq:almost-Schmidt_entanglement}
\end{align}
where $\mathsf{S} = \diag(s_1 \ge \dots \ge s_\chi)$ is a diagonal matrix built from the almost-Schmidt values and $\{\1_{L/R},\widehat{h}_{L/R}^1,\dots,\widehat{h}_{L/R}^\chi\}$ are already orthonormal. All we need to do to get to the true Schmidt decomposition is to add $\widehat{H}_{L/R}$ to the list and orthonormalize. Explicitly, we apply a Gram-Schmidt update:
\begin{equation}
	\begin{pmatrix}
		\1_L & \widehat{\v{h}}_L & \widehat{H}_L
	\end{pmatrix}
	= 
	\begin{pmatrix}
		\1_L & \widehat{\v{h}}_L & \widehat{H}_L'
	\end{pmatrix}
	\begin{pmatrix}
		1 & 0 & 0\\
		0 & \Id & \v{p}_L\\
		0 & 0 & \mathcal{N}_L
	\end{pmatrix},
	\label{eq:pLNL}
\end{equation}
where $p_L^a := \braket{\widehat{h}_L^a, \widehat{H}_L}$ ensures orthogonality and $\mathcal{N}_L := \dn{\widehat{H}_L}_F^2 - \dn{\v{p}_L}_F^2$ enforces normalization, so that $\st{\1,\widehat{h}_L^1,\dots,\widehat{h}_L^\chi,\widehat{H}_L'}$ are now orthonormal. Doing the same on the right side, the operator now becomes
\begin{equation}
	\widehat{H}= \begin{pmatrix}
		\1_L & \widehat{\v{h}}_L & \widehat{H}_L'
	\end{pmatrix}
	\underbrace{
	\begin{pmatrix}
		\mathcal{N}_R & \v{p}_R & 0\\
		0 & \mathsf{S} & \v{p}_L\\
		0 & 0 & \mathcal{N}_L
	\end{pmatrix}}_{M :=}
	\begin{pmatrix}
		\widehat{H}_R' & \widehat{\v{h}}_R & \1_R
	\end{pmatrix}^T. \label{eq:Schmidt_for_MPO}
\end{equation}
It follows that the true Schmidt values, i.e. the entanglement spectrum, is given by the singular values of the matrix $M$. The essential point is that $M$ and $S$ are almost the same matrix --- and so their spectra are as well. We compute the precise relation between the singular values of $M$ and its matrix elements in Appendix \ref{app:schmidt} with rank-one updates, and import those results to here for show the optimality of our method.

The dominant feature of the entanglement spectrum is the separation of scales between extensive and intensive values. Suppose $\widehat{H}$ comes from a translation-invariant MPO on $N\gg 1$ sites, and our entanglement cut is at some bond $(n,n+1)$ near the middle. Then the matrix elements of $M$ have a separation of scales:
\begin{equation}
	\mathcal{N}_L, \mathcal{N}_R \in \Theta(N),\quad  s_{a}, \v{p}_L, \v{p}_R \in O(1). \label{eq:scales}
\end{equation}
Without the $\v{p}$'s $M$ would be diagonal. There would then be two extensive singular values, namely $\mathcal{N}_L^2$ and $\mathcal{N}_R^2$, and $\chi$ intensive ones, $s_1^2,\dots, s_\chi^2$. Appendix \ref{app:schmidt} shows that the extensive/intensive separation remains after the $\v{p}$'s have been taken into account:
\begin{equation}
    \lambda_{-1}^2,\lambda_{0}^2 \in \mathrm{\Theta}(N) \,,\,
    \lambda_a \in \mathrm{O}(1) \,,\,  a = 1, \dots \chi \,. \label{eq:schmidt_structure}
\end{equation}
This result illustrates again why the MPS compression scheme must fail with iMPOs: the extensive Schmidt values diverge in the thermodynamic limt. Normalizing the Schmidt values, that is, considering $\sigma_a := \lambda_a / \sqrt{\sum_b \lambda_b^2}$, would not be helpful: for any $a > 0$, $\sigma_a \in \mathrm{O}(1/N)$ vanishes in the thermodynamic limit, so that the normalized spectrum retains no nontrivial information about the operator.

Intuitively, the separation of scales is a consequence of locality. Indeed, the two extensive Schmidt vectors are very close to $\widehat{H}_L \otimes \1_R$ and $\1_L \otimes \widehat{H}_R$ --- exactly the operators that the block structure of our MPOs keeps track of ``for free''. In other words, the local MPO construction automatically keeps track of the extensive part of the spectrum (to a good approximation), and we need only deal with the \textit{intensive} part. This is precisely the role of the almost-Schmidt decomposition. 

\subsection{Comparison of MPO and ``MPS-style" Compression}

We now make to a quantitative comparison between MPO and ``MPS-style'' methods on a finite chain. If we compress an operator from bond dimension $\chi$ down to $\chi'$ with either scheme, the new operators are
\begin{align*}
	\widehat{H}_\text{MPS} &= \sum_{a=-1}^{\chi'} \lambda_a \O_L^a \otimes \O_R^a \,, \\
	\widehat{H}_\text{MPO} &= \widehat{H}_L \otimes \1_R + \1_L \otimes \widehat{H}_R + \sum_{a=1}^{\chi'} s_a \widehat{h}_L^a \otimes \widehat{h}_R^a  \,,
\end{align*}
respectively. The orthogonality properties of the decompositions tell us
\begin{align}
\dn{\widehat{H} -  \widehat{H}_{\text{MPS}}  }_F^2 =  &\sum_{a=\chi'+1}^\chi \lambda_a^2 := \varepsilon_{\text{MPS}}(\chi'),\\
\dn{\widehat{H} -  \widehat{H}_{\text{MPO}}  }_F^2 =  &\sum_{a=\chi'+1}^\chi s_a^2 := \varepsilon_{\text{MPO}}(\chi').
\end{align}
To compare these, we use the eigenvalue interlacing relation (derived in Appendix~\ref{app:schmidt})
\begin{align}
  s_a \ge \lambda_a \ge s_{a+2}   \,,\, \forall \, a \in [1, \chi - 2] \,.
  \label{eq:interlacing}
\end{align}
We can therefore conclude
\begin{equation}
\varepsilon_{\text{MPS}} (\chi') \le \varepsilon_{\text{MPO}} (\chi') \le 
\varepsilon_{\text{MPS}}(\chi'-2) \,.
\label{eq:single_bond_compression_error}
\end{equation}
This means the difference between our scheme and the MPS scheme is within \textit{two} Schmidt values, which is negligible, since in practice one always truncates sufficiently deep into the spectrum that $s_{\chi'}$ is small. 

Since the MPS truncation scheme is known to be optimal\cite{verstraete2006matrix}, we can make the error from our MPO scheme $\varepsilon$-close to optimal, by truncating at $\chi'$ large enough that $\n{s_{\chi'} - s_{\chi'-2}} < \varepsilon$. There is no strict guarantee that this is possible, but for physical operators the entanglement spectrum usually becomes a continuum with increasingly small separation. It is in this sense that our truncation scheme is $\varepsilon$-close to optimal. We remark that the error analysis above applies to the truncation of a finite MPO on an individual bond. It would be interesting to analyze the global error of an iMPO compression, but we expect it to be almost exactly the same as the iMPS case.  

In summary, the MPO compression scheme only captures the intensive Schmidt values, avoiding the pathological, extensive parts of the entanglement spectrum. As a result, we obtain an excellent approximation to the optimal ``MPS'' compression while preserving the locality structure.

\subsection{Ground State Error Bound}
We conclude this section by discussing the special case of Hamiltonians. We will show that the change in the ground state of a Hamiltonian under compression is small and, just as one would expect, the error is proportional to the weight of the truncated singular values.

To frame the question, let us back up for a second. We envisage two common applications for our compression algorithm: compressing operators for use in infinite-temperature dynamics, and compressing Hamiltonians whose naive MPO bond dimensions are too large for DMRG. For the first, the figure of merit for the compression error is the change in the Frobenius norm of the operator --- which we have already shown is small and proportional to the sum of the truncated singular values. For the second, however, the figure of merit is the change in the \textit{sup} norm
	\begin{equation}
		\dn{\widehat{O}}_s^2 := \sup_{\ket{\Psi}} \frac{\braket{\Psi|\widehat{O}\widehat{O}|\Psi}}{\braket{\Psi|\Psi}}, 
		\label{eq:sup_norm}
	\end{equation}
	 or, perhaps more physically, the ground state energy. Our task for this section is to show the change in the ground state energy is also small under compression\footnote{We note that small changes to the Hamiltonian can cause dramatic changes to the ground state wave\textit{function}. For example, if the Hamiltonian is $\epsilon$-close to a first order phase transition, like $H = \sum (1-\frac{\epsilon}{2}) \widehat{Z} + \widehat{X} \widehat{X}$, then an $\epsilon$ change (such as $\epsilon \widehat{Z}$) in the Hamiltonian will completely alter the ground state, even though the change in the ground state energy will still be $\epsilon$-small. Far from phase boundaries, the ground state wavefunction and its expectation values should change continuously with the Hamiltonian.}.

As mentioned above, the class of first degree operators to which our algorithms apply is broader than the class of physically reasonable Hamiltonians. For instance, there are projectors which can be represented with small bond dimension MPOs which are first degree, but whose ground states energies are \textit{not} extensive. If one feeds in an first degree operator which is a ``non-Hamiltonian" with a superextensive ground state energy, then the error in the ground state energy may be very large. But, while this is mathematically true, such operators do not make sense as physical Hamiltonians. We therefore exclude them for consideration and restrict ourselves to operators which are sums of terms with support on at most $k$ sites.\footnote{Similar bounds apply to broader classes of Hamiltonians, but require greater technical complexity.} This allows us to give the following bound.
	\begin{prop}
		\label{prop:GS_error_bound}
		Suppose $\widehat{H}$ is an operator on $N$ sites with on-site dimension $d$ and at most $k$-body interactions. Suppose $\widehat{H}$ can be written in the form
		\begin{equation}
			\widehat{H} = \widehat{H}_L \1_R + \1_L \widehat{H}_R + \sum_{a,b=1}^\chi \widehat{h}_L^a M_{ab} \widehat{h}_R^b 
		\end{equation}
		where each $\widehat{h}_S^a$ is a unique tensor product of on-site operators (such as a Pauli string $XYXZ$ or $\widehat{c}^\dagger \widehat{c} \widehat{c}^\dagger \widehat{c} $ for fermions). 

		If we take the singular value decomposition $M = USV^\dagger$ and define $O_L := \widehat{h}_L U_L$, $O_R := V_R^\dagger \widehat{h}_R$, then for $\chi' < \chi$, we can define the compressed Hamiltonian
	\begin{equation}
		\widehat{H}' := \widehat{H}_L \1_R + \1_L \widehat{H}_R + \sum_{a=1}^{\chi'} \widehat{O}_L^a s_a \widehat{O}_R^a 
	\end{equation}
	where $\st{s_1 \ge s_2 \ge \cdots \ge s_\chi}$ are the singular values. Then the change in the ground state energy $\delta E$ satisfies
	\begin{equation}
		\delta E \le \dn{\widehat{H}-\widehat{H}'}_s \le \sqrt{d^k \sum_{a=\chi'}^\chi s_a^2} \le d^{\frac{k}{2}} \dn{\widehat{H} - \widehat{H'}}_F.
		\label{eq:ground_state_error_bound}
	\end{equation}

	\end{prop}
	The proof is given in Appendix \ref{app:ham_error_bound}.

In other words, the change in the ground state energy from truncation is proportional to the truncated singular values. It is crucial that this error does not involve $N$, the number of sites, so one can easily take the thermodynamic limit to find that, in an infinite system, the change in the ground state energy from truncating on every bond is also small. We also note that, although we have expressed this bound in terms of operators for convenience, this bound also applies to our MPO compression algorithm. Thus one may take a Hamiltonian, write it in a suboptimal MPO representation with a large bond dimension, then compress it to a small bond dimension and run DMRG or other algorithms to find its ground state energy with only a small error. This is particularly useful in the case of long-ranged interactions or two-dimensional problems, where the MPO dimensions for the naive MPOs are can be impractically large.

\section{Relation to Control Theory}
\label{sec:relation_to_control_theory}

Remarkably, our MPO canonicalization procedure is a generalization of an extremely well-studied problem in the field of control theory known as ``model order reduction.'' With this connection in mind, one can use highly optimized libraries from that community to compute MPOs compressions for general two-body Hamiltonians. The relation to control theory was noted previously in Refs~\cite{pirvu2010matrix,zaletel2015time}. Morally, one can think of writing the interaction potential as a sum of decaying exponentials. The MPO, in turn, can then be written as the sum of the small bond dimension MPOs for each operator. Our compression procedure is a strict generalization of this technique: if the input to our algorithm is a two-body interaction, then it \textit{automatically} reproduces the sum of exponentials technique. On the other hand, higher-body Hamiltonians do not obviously map to the problem solved in control theory, so it would be interesting to pursue whether our procedure has useful implications for control theory.

The control systems setting is a ``state-space'' system: a dynamical system whose state is parameterized by a $\chi$-dimensional vector $\v{x}(t)$ with linear dynamics in discrete time. The dynamics are defined by the update rule 
\begin{align}
	\v{x}(t) &= A \v{x}(t-1) + B \v{u}(t) \nonumber \\    
	\v{y}(t) &= C \v{x}(t) + D \v{u}(t) \label{eq:ABCD_control}
\end{align}
where $\v{u}(t)$ is an $n_i$-dimensional vector of possible ``input'' perturbations, $\v{y}(t)$ a $n_o$ dimensional vector of ``outputs,'' and $A$ is a matrix of size $\chi \times \chi$, $B$ is $\chi \times n_i$, $C$ is $n_o \times \chi$, and $D$ is $n_o \times n_i$. The data can thus be bundled into a $(n_o + \chi) \times (n_i + \chi)$ matrix
$\begin{pmatrix}
C & D\\
A & B
\end{pmatrix}$, which was the motivation for our MPO block conventions. One also defines \textbf{transfer function} of the system, $G(t) := C A^t B$, an $n_o \times n_i$ matrix which describes the linear input-output response at time $t$.

Two fundamental questions arise in the control theory setting. (I) Given a set of observations $G(t)$, what state-space system $(A,B,C,D)$ can reproduce the observations? (II) Given a state-space system of dimension $\chi$, can we produce a state-space system of lower order $\chi' < \chi$ which approximates $G(t)$? This problem could arise, for example, when modelling  a complex electrical circuit, where $\v{x}(t)$ parameterizes the voltages on wire segments, which we wish to approximate by a simpler ``lumped element" circuit with fewer components.

It is easy to see that a state-space system is equivalent to a MPO in the particular case of a \textit{two-body} Hamiltonian.
A two-body interaction takes the general form 
\begin{equation}
\widehat{H} = \sum_{x > y} \sum_{\alpha, \beta = 1}^{n_o, n_i} \widehat{O}^\alpha_x V^{\alpha\beta}(x - y) \widehat{P}^\beta_y \label{eq:H_twobody}
\end{equation}
where $\{\1\} \cup \{\widehat{O}_x^{\alpha}\}_{\alpha=1}^{n_0}$ and $\{\1\} \cup \{\widehat{P}_y^{\beta}\}_{\beta=1}^{n_i}$ are orthonormal sets of operators on sites $x$ and $y$ respectively. On the other hand, each set of matrices $A, B, C$ as in~\eqref{eq:ABCD_control} define an MPO in regular form via
$$ \CC = \widehat{\v{O}} C \,,\, \BB = B \widehat{\v{P}} \,,\, \AA = A \1 \,,\, \DD = 0  \,,$$
where $\widehat{\v{O}} = ( \widehat{O}^{\alpha} )_{\alpha=1}^{n_0}$ and $\widehat{\v{P}} = ( \widehat{P}^{\beta} )_{\beta=1}^{n_0} $. It is not hard to check this MPO represents the Hamiltonian~\eqref{eq:H_twobody} if and only if 
$$ \left[ C A^r B \right]^{ab} = V^{ab}(r) \,. $$
This data is in precise agreement with that of state-space system, with the transfer matrix $G(t)$ of the state-space encoding the two-body interaction $V(r)$. One could easily include on-site terms as well, in which case $\DD$ would be non-zero. 

With this mapping, we see that problems (I) and (II) are equivalent to finding an MPO which reproduces a desired two-body interaction, and approximating an MPO by one of lower bond dimension. In the control theory literature, (I) has been solved by an algorithm of Kung \cite{Kung1978}, and (II) by ``balanced truncation''\cite{silverman1980optimal}, which we focus on here.

The starting point of the balanced truncation algorithm is the
``controllability'' Gramian $X$ and the ``observability'' Gramian $Y$, 
\begin{subequations}
\begin{align}
X &\equiv \sum_{k=0}^\infty A^k B B^\dagger \left({A^\dagger}\right)^k \\
Y & \equiv \sum_{k=0}^\infty \left({A^\dagger}\right)^k C^\dagger C  A^k
\end{align}
\end{subequations}
They are determined by the discrete Lyapunov equations
\begin{align}
 A X A^\dagger  &= X - B B^\dagger \\
 A^\dagger Y A  &= Y - C^\dagger C 
\end{align}
We can identify these as the fixed point condition for the left/right eigenvectors of the right/left transfer matrix $T_{R/L}$ of $\widehat{W}$  (c.f. $T_V$ above) in the particular case that $\widehat{A} = A \mathds{1}$.
The controllability Gramian $X$ is nothing other than the  relevant block of the dominant eigenvector  of the transfer matrix, and similarly for $T_L$ and $Y$.

The idea of balanced truncation is to use the gauge freedom $A \to g A g^{-1}, C \to C g^{-1}, B \to g B$, under which the Gramians transform as $X \to g X g^\dagger, Y \to {g^\dagger}^{-1}Y g^{-1}$, to demand that the Gramians be equal and diagonal: $X = Y = \textrm{diag}(\Sigma)$. This is called the \textit{balanced} condition. The $\Sigma$ are called the ``Hankel singular values'' for reasons we will explain shortly.
In operator language, this is nothing other than the almost-Schmidt decomposition Eq. \eqref{eq:almost_schmidt_decomposition} with values $s_a = \Sigma_a$.
In balanced truncation, the model is then  reduced by keeping the largest $\Sigma_a$, which is known to be optimal with respect to a particular norm, the ``Hankel norm''~\cite{al1987error}.

Indeed, with this mapping in mind, the balanced truncation algorithms found in the literature are equivalent to the canonicalization procedure discussed here: solve the Lyapunov equations for the Gramians $X, Y$ (equivalent to finding the dominant eigenvector of transfer matrix), compute the Cholesky decompositions $X = R R^\dagger$ and $Y = L L^\dagger$, and then SVD $U \Sigma V = L^\dagger R$, and let $g = \Sigma^{-1/2} V R^{-1}$.

Why are they called Hankel singular values? This brings us to Kung's algorithm, which obtains an approximate state-space representation given the desired output $G(t) \sim V(r)$.
For simplicity, let's consider the simplest $n_i = n_o = 1$ case, arising for instance from a density-density interaction $\widehat{H} = \sum_{i, r > 0} \widehat{n}_{i+r} V(r) \widehat{n}_{i}$.
It is easy to see that in the mixed-canonical form at bond $(0,1)$, the left / right operators can be chosen to be
$\widehat{h}^{i}_L = \widehat{n}_{-i}, \widehat{h}^{i}_R = \widehat{n}_{i+1}$ for $i\geq 0$ so that $H = \sum_{i, j} \widehat{h}^{i}_L V(i + j + 1) \widehat{h}^{j}_R$.
The middle tensor then takes the form
\begin{align}
M = 
\begin{pmatrix}
 & V(3) & V(2) & V(1)\\
\cdots & V(4) & V(3) & V(2)\\
& V(5) & V(4) & V(3)\\
\cdot^{\cdot^{\cdot}} & & \vdots & \\%there should be a better way to do this
\end{pmatrix},
\end{align}
 which is by definition a ``Hankel matrix,'' with singular values $M = U \Sigma V$ consequently referred to as the Hankel singular values. 

The connection results in highly optimized routines to compute the optimal $A, B, C$ from the desired $V$ using the Hankel structure. These are provided, for example, in the MATLAB Control Systems Toolbox as \texttt{balred, imp2ss} and in the SLICOT library~\footnote{See \hyperlink{http://slicot.org}{\texttt{http://slicot.org}}.} as \texttt{AB09AD}. The latter has a convenient Python API provided in the ``\texttt{control}'' library~\footnote{See \hyperlink{http://python-control.org}{\texttt{http://python-control.org}}.}, which we have used with great success for quantum Hall DMRG \cite{zaletel15}.

While the equivalence is clear in the two-body case, what is the control theory interpretation of canonicalizing and truncating a more general MPO? This seems like an interesting question.

\section{iMPO Examples}
\label{sec:examples}

This section provides two numerical examples of iMPO compression. This is where our almost-Schmidt compression scheme truly shines, as the standard ``MPS''-type truncation schemes do not work at all in this regime. Indeed, to our knowledge, our algorithm is the only one known to work for general iMPOs. We first give a ``proof-of-concept'' example for long-ranged Hamiltonians and then give an iMPO implementation of the Lanczos algorithm.

\begin{figure}
	    \centering
	    \includegraphics{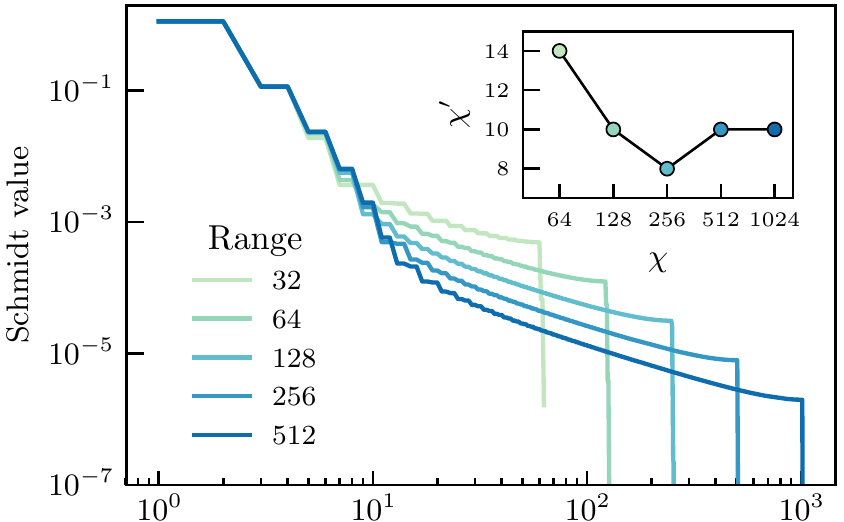}
		\caption{Compression of the iMPO representing~\eqref{eq:H_2_example}. Main: The almost-Schmidt spectra of iMPOs representing $\widehat{H}_2$ with spatial cutoff $R$ ranging from $32$ to $512$. As $R\to \infty$, the largest $s_a$  converge to a point-wise limit, while the long tails rapidly decays (so the latter are finite-$R$ artifacts). Inset: the bond dimensions of the iMPO before and afterwards with a cutoff of $\varepsilon = 10^{-4}$.  Other numerical thresholds are the same as Fig. \ref{fig:mpo}.}
    \label{fig:impo}
\end{figure}

We consider the three-body Hamiltonian
\begin{equation}
	\widehat{H}_2 = \sum_{n\in \Z} \sum_{x, y > 0} 
	\widehat{Z}_{n-x}  \widehat{X}_{n} \widehat{Z}_{n+y} J_{x} J_y, \quad 
	J_r = r^{-2} \,,
	\label{eq:H_2_example}
\end{equation}
with power-law interaction. To encode the Hamiltonian (which has a formally infinite bond-dimension), we give the power-law interaction a large spatial cutoff $R$: $J_r := 0$ for $r > R$ , which we vary, so that the pre-compression bond dimension is $\chi = 2 R$. 
The pre-compressed iMPOs have a block structure specific to three-body interaction; for example, when $R = 3$, we have
   \begin{equation}
    \left(\begin{array}{@{}c|ccc|ccc|c@{}}
		\1&  \widehat{Z} & 0  & 0  & &  & \\ \hline
		   &       0   &  \1 & 0     & J_1 \widehat{X}  &0  &0 & \\
		   &         0 &   0 & \1 &J_2 \widehat{X}  &0 &  0& \\ 
		   &      0    & 0   & 0  &J_3 \widehat{X}  &0 & 0 & \\  \hline
		   &   &   & &0 &  \1 &0 &  J_1\widehat{Z}  \\ 
		   &   &   & & 0 &0  &  \1 &  J_2 \widehat{Z}   \\
		      &   &   & & 0 & 0 &  0 &   J_3\widehat{Z}   \\ \hline
		  &   &   & &  & &  &\1\\
	\end{array}\right) \,. \label{eq:W_example0}
\end{equation} 
 We then compress them the iMPO compression routine~(Algorithm \ref{alg:iMPO_compression}) which calls the upper-triangular canonical form subroutine (Algorithm~\ref{alg:left_can_form_triangular}). The results are given in Figure \ref{fig:impo}.  
For any reasonable tolerance, as $R\to\infty$, the compressed bond dimension stabilizes to a tiny value, thanks to the rapid decay of the almost Schmidt values. It is also interesting to examine a compressed MPO (from $\chi = 2R =  256$ to $\chi' = 4$): 
% \begin{equation}
%     \left(\begin{array}{@{}c|cc|cc|c@{}}
% 		\1& .95\widehat{X} & .27\widehat{X} &   & &  \\ \hline
% 		   &  .29 \1         &  -.72\1      & .95\widehat{Z}  & .27\widehat{Z} & \\
% 		   &  -.07  \1        &  .63 \1     & .03\widehat{Z} & .01 \widehat{Z}  & \\  \hline
% 		   &   &   & .29\1 &  -.72\1 & 1.08  \widehat{X}  \\ 
% 		   &   &   & -.07\1 & .63\1 & 0.03  \widehat{X}   \\ \hline
% 		  &   &   & & & \1\\
% 	\end{array}\right) \,. \label{eq:W_example}
% \end{equation}
\begin{equation}
    \left(\begin{array}{@{}c|cc|cc|c@{}}
		\1& \widehat{Z} & \widehat{Z} &   & &  \\ \hline
		   &  0.178 \1  &           & .749\widehat{X}  & .114\widehat{X} & \\
		   &            &  .742 \1  & .11\widehat{X} & .0117 \widehat{X}  & \\  \hline
		   &            &           & 0.178 \1 & &  \widehat{Z}  \\ 
		   &            &           &        &  .742 \1  & \widehat{Z}   \\ \hline
		  &   &   & & & \1\\
	\end{array}\right) \,. \label{eq:W_example}
\end{equation}
Remarkably, while the strict locality of the uncompressed MPO is compromised, the {block} triangular structure of \eqref{eq:W_example0} is intact. We can clearly see that each power-law is approximated by a sum of exponential decays governed by the $2\times 2$ matrix on the diagonal block of \eqref{eq:W_example}. Here we have applied a gauge transform after compression to make the MPO upper-triangular, but this is not possible in general as the SVD step will destroy ``triangularizability".

Our final example is somewhat more involved: an iMPO implementation of the Lanczos algorithm. The Lanczos algorithm is originally from numerical linear algebra, where it is used to tri-diagonalize a matrix. However, it was recognized in the 1980s~\cite{mattis1981reduce} that it provides an exact mapping from many-body dynamics problems to 1d quantum mechanics problems on a semi-infinite tight-binding model. (This is known as the ``recursion method'', see \cite{viswanath2008recursion} for a review.) Recent work by some of us~\cite{parker2018universal} has found there are deep connections between the Lanczos algorithm, thermalization, operator complexity, and quantum chaos.

The Lanczos algorithm is a simple iteration. Suppose $\widehat{H}$ is a Hamiltonian and $\widehat{O}$ is a Hermitian operator. Conceptually, the Lanczos algorithm constructs the Krylov subspace $\operatorname{span}\{\widehat{O}, [\widehat{H},\widehat{O}], [\widehat{H}, [\widehat{H}, \widehat{O}]],\dots\}$ and iteratively orthonormalizes it. More precisely, we start from $\widehat{O}_{-1} = 0 $, $\widehat{O}_{0} := \widehat{O}$, $b_0 := 0$, and for $n > 0$, we define recursively
\begin{subequations}
\begin{align}
	\widehat{A}_n &:= [\widehat{H}, \widehat{O}_n] - b_{n-1} \widehat{O}_{n-2} \nonumber \\
	\widehat{O}_n &:= b_n^{-1} \widehat{A}_n \text{ where } b_n := \dn{\widehat{A}_n}^{1/2} \,.
\end{align}
\end{subequations}
The $b_n$'s are known as the \textit{Lanczos coefficients}, and it is well-known that $\{\widehat{O}_0, \dots,\widehat{O}_n\}$ is an orthonormal basis of the $n$-dimensional Krylov subspace. These objects are highly relevant for the operator dynamics $\widehat{O}(t) = e^{i\widehat{H}t} \widehat{O } e^{-i\widehat{H}t} $, and it is desirable to compute as many of them as possible.

For generic many-body problems, exactly computing $n$ Lanczos coefficients requires $\mathrm{O}(e^{C n})$ resources. Now, whenever $\widehat{H}$ and $\widehat{O}$ are representable as iMPOs, the whole Lanczos algorithm can be implemented using iMPOs using elementary operations from Appendix~\ref{app:elementary_operations} and the intensive norm formula~\eqref{eq:norm_left_can}. If $\widehat{O}_0$ is first degree and $\widehat{H}$ is strictly local, all iMPOs generated in the process will be first degree, so our compression scheme can  potentially reduce the computation cost of the Lanczos algorithm.

We benchmarked our iMPO implementation of the Lanczos algorithm, with the paradigmatic chaotic Ising chain, see Fig.~\ref{fig:lanczos}. Remarkably, we observe that the resulting bond dimension of the operators $\widehat{O}_n$ grows only polynomially: 
\begin{equation}
    \chi[\widehat{O}_n] = \mathrm{O}(n^a) \,,\, a \approx 2 \,,  \label{eq:chi_lanczos}
\end{equation}
shown in Fig. \ref{fig:lanczos} (c), while one would naively expect exponential growth. This means that, in principle, one could reach $n=60-80$ with moderate hardware, far beyond $30-40$ by the exact method~\cite{parker2018universal}.

Practically, however, numerical precision becomes a limiting issue. Due to the iterative nature of the algorithm, any small compression error in $\widehat{O}_n$ is magnified on subsequent steps. One can see from Fig. \ref{fig:lanczos} (d) and (e) that the $\widehat{O}_n$'s singular value spectrum has a gap where the almost Schmidt values fall off by several orders of magnitude. A truncation targeted at the gap will be essentially lossless. However, the smallest singular value above the gap decreases rapidly with $n$, eventually reaching machine precision. Beyond that point, the singular value spectrum will look continuous with no apparent gap, and any further truncation will induce errors that grow quickly --- as shown in Fig. \ref{fig:lanczos} (b).  One could account for this by dynamically increasing the working precision along with $n$. Although this is harder to implement and slower,  the resource cost would still grow only polynomially with $n$, a qualitative improvement over the exact method, so long as \eqref{eq:chi_lanczos} continues to hold.  It will be very interesting to elucidate the reason of such an advantageous bond dimension scaling.

	\begin{figure}
	    \centering
	    \includegraphics{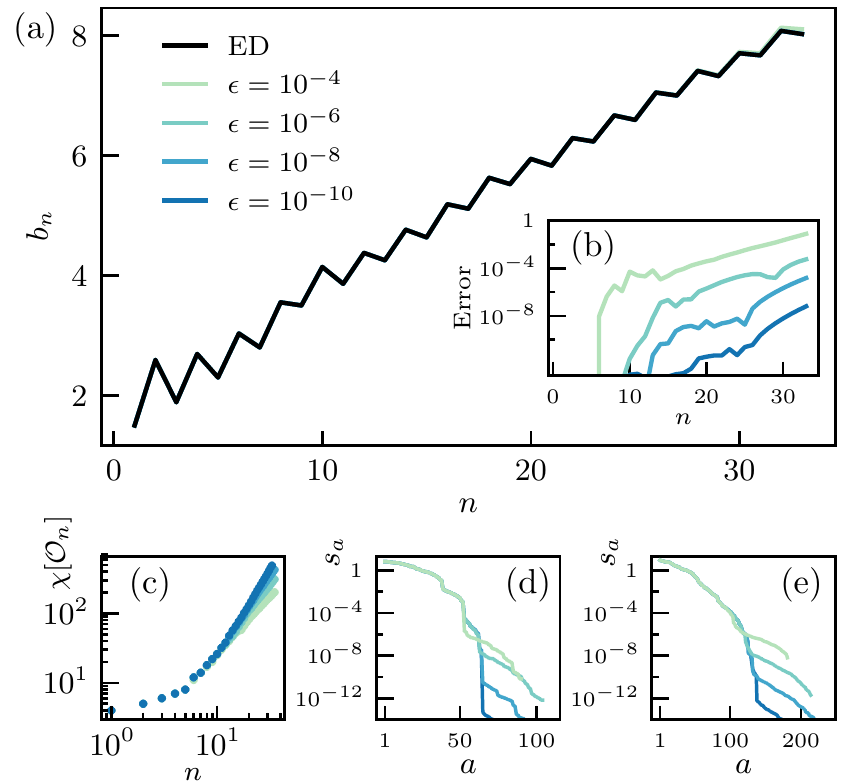}
		\caption{Results of an iMPO implementation of the Lanczos algorithm, applied to $\widehat{H} = \frac12 \sum_n Z_n Z_{n+1} -1.05 Z_n + 0.5 X_n$ and $\widehat{O} = \sum_n Z_n$. (a) The Lanczos coefficients $b_n$ computed by the iMPO implementation with SVD truncation threshold $\varepsilon$, compared to the exact method (``ED'') of Ref. \cite{parker2018universal}. (b) Error in the $b_n$'s at precision $\varepsilon$ (compared to ``ED'' values). (c) Bond dimension of the operators $\widehat{O}_n$. The growth rate is roughly $\mathrm{O}(n^2)$.  (d) The almost Schmidt spectra of $\widehat{O}_{10}$. A large gap is visible at $a \sim 60$ where $s_a$ drops by $\sim 10^{-6}$. (e) The almost Schmidt spectra of $\widehat{O}_{20}$. The gap is barely visible even with the smallest $\varepsilon$; the error starts to grow rapidly around the same $n$.} 
	    \label{fig:lanczos}
	\end{figure}

\section{Conclusions}
\label{sec:conclusions}

In this work we have endevoured to promote matrix-product operators to ``first-class citizens'' amoung computational techniques. Our primary focus was the physically relevant case of local operators, operators that tend to the identity at spatial infinity. Locality of an operator imposes a constraint upon its matrix-product representations, namely a certain upper-triangular block structure. We then adapted the standard tools and techniques of matrix-product states to this framework. In particular, we generalized the notion of left and right canonical forms to the MPO case in a way that respects the local structure, and gave efficient algorithms for computing them. These lead naturally to a novel compression scheme for MPOs that also respects locality and is almost as optimal as SVD truncation is in the MPS case. We treated both the finite and infinite cases and proved the correctness of our techniques wherever possible. To showcase the utility of these new techniques, we included two brief applications: computing the Lanczos coefficients of operator dynamics, and compressing long-range (i)MPOs. In summary, this work enables all standard operations of matrix-product states to be performed on explicitly local matrix-product operators.

On a practical level, these results are applicable both to simulating quantum dynamics in 1d and solving strongly correlated systems in 2d. In 1d, this compression scheme should enable hydrodynamic coefficients, such as diffusion or conductivity, to be calculated using Krylov space techniques. The idea is that the Green's function $G(\omega,k)$ may be well-approximated by information contained in the Lanczos coefficients~\cite{parker2018universal,viswanath2008recursion}. Above we computed these for an example model at $k=0$ (translation invariant sums), but one may work at arbitrary wavevector by slightly modifying the form of the MPO to
\begin{equation}
	\W(k) = \begin{pmatrix}
		e^{ik} \1 & \CC & \DD\\
		0 & \AA & \BB\\
		0 & 0 & \1
	\end{pmatrix}.
\end{equation}
This application will be the focus of future work. In 2d, DMRG studies on infinite strips can be limited by the large bond dimension of the Hamiltonian operator. However, since these Hamiltonians are constructed ``by-hand'', it is reasonable to expect that, in many cases, they can be highly compressed. Moreover, as they have an upper-triangular form, this compression can be carried out quite efficiently. Alternatively, one could use an ``over-compressed'' Hamiltonian as a pre-conditioning step to find an approximate ground state before carrying out the full DMRG algorithm. In any event, the operator-centric tools developed in this work should bring immediate practical benefits to a variety of applications.

We wish to close with a few speculative remarks on our theoretical results. Operators are more than merely states in a doubled Hilbert space in at least two ways: (I) they have an algebraic structure and can thus be multiplied, and (II) they can be local. One perspective on this work is that local operators, as we have defined them, are the analogue of area law states, with a bounded amount of information per site. The standard notions of quantum information theory, especially the entanglement spectrum, struggle to capture the non-trivial local structure of operators --- which is what led us to define the ``almost-Schmidt decomposition". It is unclear how general this notion is. For example, how do we treat ``second degree'' and ``multi-local'' operators that arise naturally as products such as $\widehat{H} \widehat{H}$ (used in computing energy fluctuations in DMRG~\cite{hubig2017generic})? Can it be extended beyond 1d?

Curiously, the algebraic nature of operators is almost completely absent from this work. After all, locality is a by-product of the operator algebra, namely the condition that spatially-separated operators tend to commute. It is natural to speculate that a deeper ``quantum information theory of operators'' would be intimately connected to the operator algebra structure and yield greater benefits for computation.

\begin{acknowledgments}
  We thank Ehud Altman, Nick Bultinck, Joel E. Moore, Johannes Motruk, Frank Verstraete, and Laurens Vanderstraeten for helpful discussions. We thank the second anonymous referee whose careful comments greatly improved this work.  We acknowledge support from the NSF Graduate Research Fellowship Program NSF DGE 1752814 (DP), ERC synergy Grant UQUAM and DOE grant DE-SC001938 (XC). MZ was supported by the DOE, office of Basic Energy Sciences under contract no. DEAC02-05-CH11231.
\end{acknowledgments}

\newpage

	\appendix

\section{Proofs for Local MPOs}
\label{app:proofs}
This appendix proves statements about local MPOs from Section~\ref{sec:local_iMPOs} of the main text. Our main goal is the proof of the forms of the dominant Jordan blocks, Prop~\ref{prop:quasi_local_Tv}, but we begin with a series of technical Lemmas.

\begin{lemma}
   Suppose $\W$ and $\W'$ are related by a gauge transform $L \W = \W' L$, and $\W$ is first degree. Then $\W'$ is also first degree. 
   \label{lem:quasilocal-gauge}
\end{lemma} 
\begin{proof}
	The block triangular form~\eqref{eq:gauge_form} of the gauge matrix $L$ implies the sub-matrices $\AA$ and $\AA'$ are related by $\mathsf{L} \AA = \AA' \mathsf{L}$. Then, by the definition of the transfer matrix, we have 
   \begin{align}
    &    [\mathsf{L}^\dagger X \mathsf{L}] T_A = 
    \sum_\alpha  A_{\alpha}^\dagger \mathsf{L}^\dagger X \mathsf{L}  A_{\alpha}  \nonumber \\
    = 
     &\sum_\alpha \mathsf{L}^\dagger (A'_{\alpha})^\dagger X A'_{\alpha}  \mathsf{L} 
     = \mathsf{L}^\dagger (X T_{A'}) \mathsf{L}  \,. \label{eq:JXJ}
   \end{align}
   Now, suppose $\W'$ is not first degree, then there is $X$ such that $X T_{A'} = \lambda X$ with $|\lambda| \ge 1$. By \eqref{eq:JXJ},
   $Y := \mathsf{L}^\dagger X \mathsf{L}$ is an eigenvector of $T_A$ with the same $\lambda$, which contradicts the first degree property of $\W$.  
\end{proof}
\begin{lemma}
	Suppose $\spec(T_A)$ is strictly inside the unit disk. Then so is $\spec(A_0)$.
	\label{lemma:A_0_spectrum_constraint}
\end{lemma}
\begin{proof}
	Suppose not. Then there is a (generalized) eigenvalue $\sigma \in \spec(A_0)$ with $\n{\sigma} \ge 1$. This eigenvalue must be in some Jordan block
	\begin{equation}
		J = \begin{pmatrix}
			\sigma & 1 & & &\\
			& \ddots & \ddots & &\\
			& & \sigma & 1\\
			& & & \sigma
		\end{pmatrix}
	\end{equation}
	with some (generalized) eigenvector $A_0 \v{v} = \sigma \v{v}$. Then $\v{w} := \v{v}^\dagger \otimes \v{v}$ is an eigenvector $T_{A_0} \v{w} = \n{\sigma}^2 \v{w}$. So
	\begin{equation}
		\IP{\v{w}}{T_{A^0}^N \v{w}} = \n{\sigma}^{2N} \geq 1 \,,\, \forall N \,.
	\end{equation}

	For each component $A_\alpha$, $0\le \alpha < d^2$, of $\AA$, let $T_\alpha[X] := \left( A_\alpha \right)^\dagger [X] A_\alpha$. Each of these is a positive map and $T_A = \sum_\alpha T_\alpha$, so
	\begin{equation}
		 T_A^N = T_{0}^N + \sum_{\substack{\alpha_1,\dots, \alpha_N\\\exists \alpha_i \neq 0}} T_{\alpha_1} \cdots T_{\alpha_N},
		 \label{eq:T_A_power}
	\end{equation}
	Since the composition of positive maps is positive, $\IP{\v{w}}{T_{\alpha_1} \cdots T_{\alpha_N} \v{w}} \ge 0$. So \eqref{eq:T_A_power} implies $\IP{\v{w}}{T_A^N\v{w}}~\ge~\n{\sigma}^{2N} \to \infty$.	But all the eigenvalues of $T_A$ are less than $1$, so $\IP{\v{w}}{T_A^N \v{w}} \to 0$, a contradiction.
\end{proof}

\begin{lemma}
    Let 
    \begin{equation}
         T = \begin{pmatrix}
          A & B \\ 0 & C 
         \end{pmatrix}
    \end{equation}
    be a block upper-triangular matrix such that $A$ and $C$ are square matrices. Let $\lambda \in \mathrm{spec}(A) \setminus \mathrm{spec}(C)$ and $(x \; y) T = \lambda (x \; y)$ be a left eigenvector. Then $x \neq 0$ and satisfies $x A = \lambda x$, so $x$ is a left eigenvector of $A$.
    \label{lemma:triangular}
\end{lemma}
\begin{proof}
    $(x \; y) T = \lambda (x \; y)$ means $xA = \lambda x$ and $x B + y = \lambda y$.
    Now suppose $x = 0$. Then $y \neq 0$, and $y C =  \lambda y$, so $\lambda \in \mathrm{spec}(C)$, a contradiction. 
\end{proof}
	
\begin{lemma}
\label{lem:traceform}
    Suppose $\W$ is an first degree iMPO. Then there exists a gauge transform 
	\footnote{When using this to compute the norm via Eq. \eqref{eq:norm_left_can}, one should only compute $\v{s}$ and set $\v{t} \equiv 0$ so that left-canonical form is preserved. When compressing iMPOs, one should instead set $\v{s} = 0$ and use only $\v{t}$.}
    \begin{equation}
    \W' = L \W L^{-1} \text{ where }
        L = \begin{pmatrix}
        1 & \v{t} & 0 \\ 
        0 &  \Id &  \v{s} \\
        0 & 0 & 1 
        \end{pmatrix} \label{eq:gauge_perp_form}
    \end{equation}
    such that
    \begin{equation}
        W_0' = \IP{\1}{\W'} = \begin{pmatrix}
        1 &  0 & d_0' \\
         0 & A_0 & 0 \\
          0 &  0 & 1 
        \end{pmatrix} \, \label{eq:W0_traceform}
    \end{equation}
    In fact, $\AA' = \AA$ is unchanged.
\end{lemma}
\begin{proof}
	A direct computation shows
    \begin{equation}
        \v{s} := (A_0- \Id)^{-1} \v{b}_0 \,,\,
        \v{t} := \v{c}_0  (A_0 - \Id)^{-1} \,
		\label{eq:last_col_Id_free}
    \end{equation}
    give the desired gauge transform. The inverse $(A_0- \Id)^{-1}$ exists by Lemma \ref{lemma:A_0_spectrum_constraint}.
\end{proof}

We now have all the tools needed to unravel the Jordan block structure of MPOs. We prove Prop \ref{prop:dominant_jordan_block} in the special case where the operator has subextensive trace, and subsequently sketch the more general case.

\begin{proof}[Proof of Prop. \ref{prop:dominant_jordan_block}.]

The idea of the proof is to explicitly find the dominant Jordan block (i.e. the Jordan block that gives the leading contribution to the norm) using the block structure of $T_W$. Unfortunately, just as in \eqref{eq:T_W_example}, there are two other ``spurious" eigenvectors whose eigenvalue is also $1$. Just as the dominant Jordan block is responsible for the extensive norm, they give rise to the extensive part of the trace. For a traceless operator, they form an invariant subspace that does not contribute to the extensive norm --- hence the name ``spurious".

We first impose the condition of tracelessness. Without loss of generality, we work in the gauge of Lemma~\ref{lem:traceform}, and note that $\AA$ is unchanged so the first degree property is maintained. On a finite system of $N$ sites, the trace is given by,
\begin{align}
    \mathrm{Tr}[\widehat{H}_N]& = \v{\ell} W_0^N \v{r}  = 
   \left(1 \; \v{\ell}' \; \ell_{\chi+1} \right) W_0^N  
    \left(r_0 \; \v{r}' \; 1\right)^T \nonumber \\
  & =  r_0 + \ell_{\chi+1} + N d_0 + \v{\ell}' A_0^N \v{r}' \nonumber \\
  & = N d_0 + O(1) \,,\, N \to \infty \,,
\end{align}
where we used the standard boundary conditions \eqref{eq:standard_boundary_conditions} and used Lemma~\ref{lemma:A_0_spectrum_constraint} for the last asymptotic. Therefore
\begin{equation}
    \lim_{N\to \infty} \frac{1}{N} \Tr[\widehat{H}_N] = 0 \Longleftrightarrow d_0 = 0 \text{ in gauge \eqref{eq:gauge_perp_form}}.
\end{equation}

We now exhibit all the generalized eigenvectors with eigenvalue $1$. For concision, we rewrite $\W$ as
\begin{equation}
    \W = \begin{pmatrix} \V &  \v{\widehat{f}} \\ 0 & \1 \end{pmatrix}
    \,,\,
   \widehat{ \v{f}} := \begin{pmatrix} \DD\\
    \BB \end{pmatrix} \label{eq:W_with_f} \,,
\end{equation}
with block sizes $1+\chi$ and $1$. Similarly to Eq. \eqref{eq:Tv_block}, we have
\begin{equation}
    T_{W} = 
	 \left(
	 \begin{array}{c|cc|cc|c}
	T_V & 0 & \overline{U} & 0 & U & F \\ \hline
	0 & 1 & 0 & 0 & 0 &  0 \\
	0 & 0 & A_0 & 0 & 0 & 0 \\ \hline
	0 & 0 & 0 &  1 & 0 & 0 \\
	0 & 0 & 0 &  0 & \overline{A}_0 & 0 \\ \hline
	0 & 0 & 0 &  0 & 0 & 1 \\
	\end{array} \right) 
    \label{eq:TW_block}
\end{equation}
for some $U$, where the block sizes are $(1+\chi)^2, 1 + \chi, 1+ \chi, 1$, and
\begin{equation}
    F := \sum_{\alpha} \overline{\v{f}}_\alpha \otimes \v{f}_\alpha \,.\label{eq:Fmatrix}
\end{equation}
We observe that $T_W$ is the sum of ``reduced" and ``spurious" parts
\begin{equation}
  T_W =   \begin{pmatrix}
	T_{V} & \overline{U} & {U}  & F \\
	0 & { A_0} & 0 &  0 \\
	0 & 0 & \overline{A}_0 & 0 \\
	0 & 0 & 0 &  1
	\end{pmatrix}
	\oplus 
    \begin{pmatrix}
    1 & 0 \\ 0 & 1 
    \end{pmatrix}
    =: T_{\text{red}} \oplus T_{\text{sp}}
\end{equation}
The spurious block $T_{\text{sp}}$ has eigenvectors $E$ and $E^T$ where $E_{ab} = \delta_{a0} \delta_{b, \chi+1}$ and, in particular, $E_{00} = 0$.

\begin{comment}

\begin{align}
 &   T_W = \begin{pmatrix}
    1 & 0 \\ 0 & 1 
    \end{pmatrix} \oplus T_{W,r}  \,.  \nonumber \\
  &  T_{W,r} := 
     \begin{pmatrix}
	T_{V} & \overline{U} & {U}  & F \\
	0 & { A_0} & 0 &  0 \\
	0 & 0 & \overline{A_0} & 0 \\
	0 & 0 & 0 &  1
	\end{pmatrix} \,.
\end{align}
The two spurious eigenvalues $1$ have no Jordan block structure. Their eigenvectors are $E$ and $E^T$, where
\begin{equation}
    E := \begin{pmatrix}
     0 & 1 \\
     0 & 0 \\
    \end{pmatrix} \label{eq:Eif}
\end{equation}
with the matrix notation (the block sizes are $\chi+1$ and $1$).
\end{comment}

The dominant Jordan block comes from $T_\text{red}$. Consider the truncated operator
\begin{equation}
    T_\text{red}^\text{truncated} = 
    \begin{pmatrix}
	T_{V} & \overline{U} & {U}  & 0 \\
	0 & { A_0} & 0 &  0 \\
	0 & 0 & \overline{A}_0 & 0 \\
	0 & 0 & 0 &  0
	\end{pmatrix} \,.
\end{equation}

By Proposition~\ref{prop:quasi_local_Tv} and Lemma~\ref{lemma:A_0_spectrum_constraint}, it has a unique eigenvalue $1$ (the rest have $|\lambda|<1$). By Lemma~\ref{lemma:triangular}, the corresponding left eigenvector of $T_\text{red}$ is (after rescaling)
\begin{equation} Z' = \begin{pmatrix}X & \v{z}  \\  \overline{\v{z}}  & 0 \end{pmatrix}  \,.  \label{eq:Z'}
\end{equation} for some $\v{z}$ and where $X$ is the unique largest eigenvector of $T_V$ from Eq. \eqref{eq:T_V_leading_eigenvector}. Then we have
\begin{align}
	&	\begin{pmatrix}
			Z' T_{W} &
			Z T_W \\
		\end{pmatrix}
		= \begin{pmatrix}
			Z' &
			Z \\
		\end{pmatrix} \begin{pmatrix}
			1 & \rho \\
			0 & 1
		\end{pmatrix} \,,\, Z = \begin{pmatrix} 0 & 0  \\ 0  & 1 \end{pmatrix} 
	\label{eq:dominant_jordan_block}
\end{align}
where $ \rho := X F = \sum_{a,b=1}^{\chi} X_{ab} \IP{\v{\widehat{f}}_a}{\v{\widehat{f}}_b}$. (In practice, one should compute $\rho$ using Eq. \eqref{eq:norm_left_can} which makes use of canonical form.) All the other eigenvalues of $T_{\text{red}}$, and indeed all other eigenvectors of $T_W$ are those of $A_0$ and $\overline{A}_0$, and satisfy $|\lambda| < 1$ by first degreeness. We have thus found the dominant Jordan block of $T_W$, as well as the ``spurious'' eigenvectors.

We are now ready to compute the norm $\dn{H}^2_N$ using the transfer matrix formula \eqref{eq:norm_transfer_matrix}. We expand $\v{\ell}\v{\ell}$ in the left generalized eigenbasis of $T_{W}$:
\begin{equation}
    \v{\ell}\v{\ell} = \underbrace{(a Z' + b Z)}_{\lambda = 1 \text{ Jordan block}} + 
    \underbrace{(c E + \overline{c} E^T )}_{\lambda = 1 \text{ `spurious'}}  + 
    \underbrace{S}_{|\lambda|<1} \label{eq:uu_expand}
\end{equation}  
where $S$ is a linear combination of generalized left eigenvectors with eigenvalues $|\lambda| < 1$. It follows that 
\begin{align}
     (\v{\ell}\v{\ell}) T_W^N  (\v{r}\v{r}) =  N  a \rho (Z \v{r}\v{r}) + \mathrm{O}(1) 
     = N a \rho  + \mathrm{O}(1) 
\end{align}
as $N \to \infty$, since $\v{r}_{\chi+1} = 1$ by the regular form. It remains to determine the coefficient $a$. For this we look at the $00$-component of \eqref{eq:uu_expand}. First,  $S T_W^N \longrightarrow 0$ by the definition of $S$. Meanwhile, \eqref{eq:Tv_block} and \eqref{eq:TW_block} imply $(S T_W)_{00} = S_{00}$. Therefore, $S_{00} = 0$. For the other terms of the RHS, we have $Z'_{00} = 1$ by \eqref{eq:Z'} and \eqref{eq:T_V_leading_eigenvector}, $Z_{00} =0$ by \eqref{eq:dominant_jordan_block}, and $E_{00} = 0$. On the LHS, the regular form~\eqref{eq:standard_boundary_conditions} requires $(\v{\ell}\v{\ell})_{00}=1$. Therefore we have $a = 1$ and 
\begin{equation}
      \dn{\widehat{H}_N}_F^2 =  \v{\ell}\v{\ell} T_W^N  \v{r}\v{r} = N \rho + \mathrm{O}(1)  \,.
\end{equation}

\end{proof}

As noted above, the condition that the trace is sub-extensive can be lifted.

	Suppose $\widehat{W}$ is an first degree iMPO for $\widehat{H}$. Then the transfer matrix $T_W$ has maximum eigenvalue unity with a generalized eigenspace $V_1$ of dimension four. This may be Jordan decomposed as follows:
	\begin{enumerate}
		\item[Case 1.] $V_1 = J_3 \oplus J_1$ if $\Tr[\widehat{H}_N] = O(N)$, i.e. the trace is extensive
		\item[Case 2.] $V_1 = J_2 \oplus J_1 \oplus J_1$ if $\Tr[\widehat{H}_N] = o(N)$, i.e. the trace is subextensive.
	\end{enumerate}
	Without loss of generality, we adopt the gauge from Lemma \ref{lem:traceform}.	Define block matrices of size $\chi+1 \times \chi+1$
	\begin{equation}
		Z_i = \begin{pmatrix} X & \v{z}\\ \v{z}^\dagger & 0 \end{pmatrix}, \quad	
		Z_{t} = \begin{pmatrix} 0 & \v{t}\\ 0 & 0 \end{pmatrix},\quad
		Z_f = \begin{pmatrix} 0 & 0\\ 0 & 1 \end{pmatrix} 	
	\end{equation}
	and $Z_{t'} = Z_t^\dagger$ where $X$ is the dominant eigenvalue of $T_A$, $\v{z}$ is the same as above, and $\v{t} = (1,0,\dots,0)$ is a vector of length $\chi$. These span the dominant generalized eigenspace:
	\begin{equation}
		\begin{pmatrix} Z_i\\ Z_t\\ Z_{t'}\\ Z_n \end{pmatrix} 
		T_W 
		= 
		\begin{pmatrix} Z_i\\ Z_t\\ Z_{t'}\\ Z_n \end{pmatrix} 
		\underbrace{
		\begin{pmatrix} 1 & d & d & \rho\\
						0 & 1 & 0 & d\\
						0 & 0 & 1 & d\\
						0 & 0 & 0 & 1
					\end{pmatrix}}_{M^\dagger}
	\end{equation}
	where $d$ is the extensive part of the trace: $\Tr[\widehat{H}_N] = N d$ and the dagger is because $T_W$ acts on the right. The Jordan decomposition $M = SJS^{-1}$ is then
		\begin{enumerate}
			\item[Case 1.]
				\begin{equation}
					J = 
				\begin{pmatrix}
				 1 & 0 & 0 & 0 \\
				 0 & 1 & 1 & 0 \\
				 0 & 0 & 1 & 1 \\
				 0 & 0 & 0 & 1 
				\end{pmatrix},
				\quad
				S = 
				\begin{pmatrix}
				 0 & 0 & 0 & 1 \\
				 -\frac{1}{2} & 0 & d & 0 \\
				 \frac{1}{2} & 0 & d & 0 \\
				 -\frac{\rho }{2 d} & 2 d^2 & \rho  & 0
				\end{pmatrix}
			\end{equation}
		\item[Case 2.]
				\begin{equation}
					J = 
				\begin{pmatrix}
				 1 & 1 & 0 & 0 \\
				 0 & 1 & 0 & 0 \\
				 0 & 0 & 1 & 0 \\
				 0 & 0 & 0 & 1 
				\end{pmatrix},
				\quad
				S = 
				\begin{pmatrix}
					0 & \frac{1}{\rho} & 0 & 0 \\
					0 & 0 & 0 & 1 \\
					0 & 0 & 1 & 0 \\
					1 & 0 & 0 & 0
				\end{pmatrix}
			\end{equation}
		\end{enumerate}
		Case 2 is, of course, the same as the above proof, where $Z_t$ and $Z_{t'}$ span the `spurious' dimensions and the Jordan block of size $2$ is responsible for the extensive norm. In Case 1, however, those two dimensions are now mixed together. One can compute
		\begin{equation}
			(\v{\ell}\v{\ell}) T_W^N  (\v{r}\v{r})  = N^2 d^2 + N (\rho - d) + O(1).
		\end{equation}
The Frobenius norm is then no longer extensive as is has been ``polluted" with the trace. Nevertheless, the largest eigenvalue is still unity and the matrix $Z_i$ overlaps with the dominant Jordan block.

The proof for these statements is directly analogous to the above Proof with the single modification of Eq. \eqref{eq:TW_block} to  
\begin{equation}
    T_{W} = 
	 \left(
	 \begin{array}{c|cc|cc|c}
		 T_V & d\v{t} & \overline{U} & d\v{t} & U & F \\ \hline
	0 & 1 & 0 & 0 & 0 &  0 \\
	0 & 0 & A_0 & 0 & 0 & 0 \\ \hline
	0 & 0 & 0 &  1 & 0 & 0 \\
	0 & 0 & 0 &  0 & \overline{A}_0 & 0 \\ \hline
	0 & 0 & 0 &  0 & 0 & 1 \\
	\end{array} \right). 
\end{equation}

\section{Proofs for Canonical forms}\label{app:convergence}
This appendix provides a sufficient condition for the convergence of the QR iteration in Algorithm~\ref{alg:left_can_alg_infinite_iterated_QR} for first degree MPOs, and proves the existence of left canonical forms. 

It is clear from the definition of canonical forms that only the upper-left sub-matrix $\V$ of an iMPO $\W$ will be actively involved. Indeed, any gauge transform of the sub-matrix ${L} \V = \V' {L}$ can be easily promoted the iMPO level:
\begin{align}
     \underbrace{\begin{pmatrix}
        {L} & \\ & 1
    \end{pmatrix}}_{L_W} \underbrace{ \begin{pmatrix}
        \V &  \widehat{\v{f}}\\
        & \1
    \end{pmatrix}}_{\W} = \underbrace{\begin{pmatrix}
        \V' & {L} \widehat{\v{f}}\\
        & \1
    \end{pmatrix}}_{\W'} \underbrace{ \begin{pmatrix}
        {L} & \\ & 1
    \end{pmatrix}}_{L_W}  \,.   \label{eq:promotion}
\end{align}
Hence we focus on $\V$ and its gauge transforms. \footnote{Accordingly, the notation in this appendix will differ form the main text in that gauge matrices acting on $\V$ will not have an overline.} From this point of view, the QR iteration Algorithm~\ref{alg:left_can_alg_infinite_iterated_QR} is defined by the following recursion:
\begin{subequations}
\begin{align}
   {R}_0 &:= \Id_{[0,\chi]},\\
  \label{eq:QR1_app}
   \V_{n-1} &:= {\Q}_{n} {R}_{n}, \quad \forall n \ge 1\\
   \V_{n} &:={R}_{n}  {\Q}_{n},\\ 
  	\label{eq:QR1Kn_app} 
   {L}_n &:= {R}_n \dots {R}_1,
\end{align}
\end{subequations}
where \eqref{eq:QR1_app} is a (normal) QR decomposition as defined in \eqref{eq:QRdef}.

We also point out a simple fact: two gauge transforms $ L \W = \W' L$ and $ L' \W'= \W'' L'$ can be composed to obtain a new one: $L'L \W = \W'' L' L$. 

\begin{lemma}
	QR iteration produces a sequence $\st{\V_n}$ that are each related to $\V$ be a gauge transform:
\begin{equation}
     {L}_n \V = \V_n {L}_n \,. \label{eq:QR1gauge_app} 
\end{equation}
\end{lemma}
\begin{proof}
Eq.~\eqref{eq:QR1_app} implies the gauge transform ${R}_m \V_{m-1} 
= \V_{m} {R}_m$ for any $m > 0$. Then \eqref{eq:QR1gauge_app} follows from Eq. \eqref{eq:QR1Kn_app} by composing the gauge transforms.
\end{proof}
\begin{comment}
by induction on $n$, we have:
\begin{align*}
    {L}_n \V =& {R}_n {L}_{n-1} \V =
{R}_n \V_{n-1} {L}_{n-1} \nonumber \\
= & \V_{n} {R}_n {L}_{n-1} = 
 \V_{n} {L}_{n}  \,,
\end{align*} 
the base case $n = 0$ being trivial.
\end{comment}

Algorithm~\ref{alg:left_can_alg_infinite_iterated_QR} enjoys also a close relation to the `small' transfer matrix:
\begin{lemma}
For any $n \ge 0$,  
\begin{align}
  &  \Id_{[0,\chi]} (T_{V})^n = {L}_n^\dagger {L}_n \,,\, \label{eq:QR_Tv_app}
 \end{align}
 where $\V$ has bond dimension $\chi$, that is, $(1+\chi)$ rows and columns.
 \label{lem:QR_Tv}
\end{lemma}
\begin{proof}
 We again proceed by induction on $n$. The base case $n=0$ is trivial. For $n>0$, we have
 \begin{align*}
    \Id_{[0,\chi]} (T_{V})^{n} 
    &= ({L}_{n-1}^\dagger {L}_{n-1}) T_V  \\
    &= \sum_\alpha V_\alpha^\dagger {L}_{n-1}^\dagger {L}_{n-1} V_\alpha\\
    &=  \sum_\alpha  {L}_{n-1}^\dagger V_{n-1,\alpha}^\dagger  V_{n-1,\alpha} {L}_{n-1}  & \nonumber \\
    &= \sum_\alpha {L}_{n-1}^\dagger {R}_{n}^\dagger {Q}_{n,\alpha}^\dagger {Q}_{n,\alpha}{R}_{n} {L}_{n-1} \nonumber \\
     &= \sum_\alpha {L}_{n}^\dagger {Q}_{n,\alpha}^\dagger {Q}_{n,\alpha} {L}_{n}\\
     &= {L}_{n}^\dagger  {L}_{n}  \nonumber
\end{align*}
where we used the induction hypothesis, \eqref{eq:TW_matrix_notation}, \eqref{eq:QR1gauge_app}, \eqref{eq:QR1_app}, \eqref{eq:QR1Kn_app}, and the definition of QR, respectively.
\end{proof}

We now address the sufficient condition for the convergence of QR iteration. First we must remove some arbitrariness in QR decomposition. For instance, $\W  = \Q R = (-\Q) (-R)$ are both valid, but such freedom can introduce unhelpful oscillations in $n$ preventing convergence. To this end, we require our QR sub-routine to be positive rank-revealing, in the following sense:
\begin{defn}
   Suppose $\V$ have $1+\chi$ columns and column rank $1+\chi'$, where $0 \le \chi' \le \chi$. The  QR decomposition routine $\Q, R \gets {QR}[\V]$ is called positive rank-revealing when the following are guaranteed:
   \begin{enumerate}
   \item[(I).] \textbf{Rank-revealing}: $\Q$ has $\chi'+1$ columns and $R$ has $\chi'+1$ rows. 
	   \item[(II).] \textbf{Positive}: if $\chi' = \chi$ (full column rank), $R$ has positive diagonal elements:
        \begin{equation}
            R_{aa} > 0 \,,\, a = 0, \dots, \chi \,. \label{eq:diagonal_positive}
        \end{equation}
    \end{enumerate}
	\label{defn:rank-revealing}
\end{defn}
These requirements can be fulfilled, for example, by the Gram-Schmidt procedure applied to the columns of $\V$.

\begin{prop}
  Let $\W$ is a first degree iMPO of bond dimension $\chi$, and let the sequence $(\W_n, L_n, R_n)_{n\ge 1}$ be generated by positive, rank-revealing QR starting from $\W$. 
  Suppose further that the leading eigenvector $X$ of $T_V$ is an invertible $(1+\chi)\times (1+\chi)$ matrix. 
  Then the iteration converges and brings $\W$ to left canonical form. \label{prop:QRconverge}
\end{prop}

The proof will follow a after a few lemmas.
\begin{lemma}
Let $m > 0$. Let $\mathbb{T}_m$ be the space of $m\times m$ upper-triangular matrices with positive  diagonal elements and let $\mathcal{P}_m$ be the space of $m \times m$ positive definite matrices. Then
 \begin{equation}
	 \mathbb{T}_m \ni L \mapsto L^\dagger L \in \mathcal{P}_m
 \end{equation}
 is a homeomorphism. \label{lem:Cholesky}
\end{lemma}
The continuous inverse is constructed explicitly in standard linear algebra textbooks.

In general, the QR iteration with rank revealing will produce a sequence of $\W_n$'s with reducing bond dimensions, $\chi_0 \ge \chi_1 \ge \dots $. However, when $X$ is non-singular, no strict bond dimension reduction can occur:
\begin{lemma}
  Under the same hypotheses of Prop.~\ref{prop:QRconverge}, all the $\W_n$'s have the same bond dimension as $\W$.  \label{lem:no_reduction}
\end{lemma}
\begin{proof}
By the gauge transform~\eqref{eq:QR1gauge_app} and Lemma~\ref{lem:quasilocal-gauge}, $\W_n$ is also first degree. So we can apply Prop.~\ref{prop:quasi_local_Tv} and let $X_n$ be the dominant eigenvector of $T_{V_n}$: $ X_n T_{V_n} = X_n.$  Then the gauge transform \eqref{eq:QR1gauge_app} implies \begin{equation}
    [{L}_n^\dagger X_n {L}_n] T_V = {L}_n^\dagger X_n {L}_n \,,
\end{equation}
similarly to \eqref{eq:JXJ}. This means that ${L}_n^\dagger X_n {L}_n = X$ by Prop.~\ref{prop:quasi_local_Tv}  (the constant is fixed by the $00$-th element). For $X$ to be non-singular, ${L_n}$ must be a square matrix, so the bond dimension does not change.
\end{proof}

We remark on a useful consequence of Lemma~\ref{lem:no_reduction}: since no rank reduction will happen, we only need the QR to be positive, not necessarily rank-revealing. This can be fulfilled by numerically stable implementations of QR based on Givens rotations or Householder reflections. 

\begin{proof}[Proof of Prop.~\ref{prop:QRconverge}]
By the definition of positive rank-revealing QR, and Lemma~\ref{lem:no_reduction}, for any $n\ge 1$, ${R}_n \in \mathbb{T}_{1+\chi}$, and thus ${L}_n \in \mathbb{T}_{1+\chi}$. Now, Lemma~\ref{lem:QR_Tv} and Prop.~\ref{prop:quasi_local_Tv} imply that 
\begin{equation}
  {L}_{n}^\dagger  {L}_{n} =  \Id_{[0,\chi]} (T_{V})^{n} \xrightarrow{n \to \infty}
  X = \begin{pmatrix}
  1 & \v{y} \\ 
  \v{y}^\dagger & Y 
  \end{pmatrix}. \,\, \label{eq:KntoX_app}
\end{equation}
Note that $({L}_n)_{00} = 1$ for all $n$.  Eq.~\eqref{eq:KntoX_app} implies that $X$ is positive semi-definite. Since we assume $X$ is non-singular, $X$ is positive {definite}. Then, Lemma~\ref{lem:Cholesky} implies that ${L}_n \to {L}$ for some invertible ${L}$, and the QR iteration converges as follows:
\begin{align*}
    & \V_n = {L}_n \V {L}_n^{-1} \to {L} \V {L}^{-1} := \V_L \\  
    & {R}_n = {L}_{n+1}^{-1} {L}_n \to \Id_{[0,\chi+1]}  \\
    &{\Q}_n = \V_{n-1} {R}_n^{-1} \to  \V_L \,,
\end{align*}
so that $\V_L$ is a left canonical MPS. Promoting to the iMPO level using~\eqref{eq:promotion} completes the proof.
\end{proof}

Prop.~\ref{prop:QRconverge} establishes the existence of left canonical for all ``generic'' first degree iMPOs, in the sense that $X$ is non-singular. We now treat the singular cases:
\begin{prop}
Let $\W$ be a first degree iMPO and such that the leading eigenvector $X$ of $T_V$ is positive semi-definite of rank $1+\chi'\le 1+\chi$. Then there is gauge transform $L\W = \W' L$ is such that $\W'$ has bond dimension $\chi'$ and such that $X'$ is positive definite. \label{lem:rank_reduce}
\end{prop}
\begin{proof}
We will construct the gauge transform by composing two gauge transforms, and still work on the level of $\V$.

First, we perform a Cholesky step followed by eigen-decomposition:
\begin{align}
    X = & \begin{pmatrix}
  1 & \v{x} \\ 
  \v{x}^\dagger & \mathsf{X} 
  \end{pmatrix} =  
  \begin{pmatrix}
  1 & 0 \\ 
  \v{x}^\dagger & \Id
  \end{pmatrix}  \begin{pmatrix}
  1 & 0 \\ 
  0 & \mathsf{X} - \v{x}^\dagger \otimes \v{x}
  \end{pmatrix}\begin{pmatrix}
  1 & \v{x} \\ 
  0  & \Id
  \end{pmatrix}   \nonumber 
   \nonumber\\ = &  
   \begin{pmatrix}
  1 & 0 \\ 
  \v{x}^\dagger & \mathsf{U}^\dagger 
  \end{pmatrix} 
  X_1
  \begin{pmatrix}
  1 & \v{x} \\ 
  0  & \mathsf{U}
  \end{pmatrix} =: {L}^\dagger X_1 {L} \label{eq:Ymatrix_app}
\end{align}
where $U$ is unitary and $X_1 = \diag (1,\sigma_1, \dots, \sigma_\chi)$ where
\begin{equation}
	\begin{cases}
		\sigma_a > 0 & \text{if } a \le \chi'\\
	\sigma_a = 0 & \text{if } a > \chi'.
	\end{cases}
\end{equation}
Since ${L}$ is invertible, we have the gauge transform
\begin{equation}
    \V_1 := {L} \V {L}^{-1}  \label{eq:rank_reduce_gauge1}
\end{equation}
so that the leading eigenvector of $T_{V_1}$ becomes the diagonal matrix $X_1$. Thus, the $aa$-th component of the equation $X_1 T_{V_1} = X_1$ becomes
\begin{equation}
    \sigma_a = \sum_{b=0}^\chi \sigma_b \left< (\V_1)_{ba}, (\V_1)_{ba}\right>  \,.
\end{equation}
When $a > \chi'$, $\sigma_a = 0$, so every term on the RHS must also vanish. Now for $b \le \chi'$, $\sigma_b > 0$, so $(\V_1)_{ba} = 0$. Namely, we showed that $\V_1$ has the block-diagonal form:
\begin{align}
    \V_1 &= \begin{pmatrix}
        \V'_{[0,\chi']} & 0 \\ 
        0 & * \\
    \end{pmatrix},
\end{align}
where $\V'$ has shape $(1+\chi') \times (1+\chi')$. This implies that $\V_1$ can be gauge transformed to $\V'$ by a projector:
\begin{equation}
 \begin{pmatrix}
 \Id_{[0,\chi']} & 0 
 \end{pmatrix}    \V_1=  \V' \begin{pmatrix}
 \Id_{[0,\chi']} & 0 
 \end{pmatrix}   \label{eq:rank_reduce_gauge2}
\end{equation}
It is easy to check that $T_{V'}$ has leading eigenvector $X_2 =\diag(1, \sigma_1, \dots, \sigma_{\chi'})$, which is non-singular. Composing the two gauge transforms~\eqref{eq:rank_reduce_gauge1} and \eqref{eq:rank_reduce_gauge2} and promoting them to the iMPO level completes the proof.
\end{proof}

Now we can finally prove the existence of left canonical form for all first degree iMPOs.
\begin{proof}[Proof of Prop.~\ref{prop:leftcan_exist}]
By Prop~\ref{lem:rank_reduce}, we find first a rank-reducing $L_0$ and $\W'$ so that $L_0 \W = \W' L_0$ and $\W'$ satisfies the assumptions of Prop.~\ref{prop:QRconverge}. Then the QR iteration must converge and bring $\W'$ to a left canonical $\W_L$ by some gauge transform $L_1 \W' = \W_L L_1$. Composing the gauge transforms gives $L \W = \W_L L$ with $L = L_1 L_0$.
\end{proof}

Note that the above proof and that of Lemma~\ref{lem:rank_reduce} provide a foolproof algorithm to compute the left canonical form: first precondition the MPO by reducing its rank, then use QR iteration. We provide an implementation in Algorithm~\ref{alg:general_left_canonical_form}. This algorithm is provably convergent for all first degree iMPOs, and has comparable numerical precision and stability to the QR iteration Algorithm~\ref{alg:left_can_alg_infinite_iterated_QR}. (Recall that any method of taking the square root of $X$ directly reduces the precision from $10^{-16}$ to $10^{-8}$ with standard floating point; QR iteration is required for high precision.)
\begin{figure}
\begin{algorithm}[H]
    \caption{iMPO Left Can. Form: General}
	\label{alg:general_left_canonical_form}
	\setstretch{1.35}
	\begin{algorithmic}[1]
        \Procedure{Precondition}{$\W, \eta$}  
	    \State $X \gets \textsc{EigMax}(T_V)$  \Comment{Find max. eigenvector}
		\State $\v{x}, \mathsf{U}, \Sigma \gets X$  \Comment{Eq.~\eqref{eq:Ymatrix_app}}
	    \State $\chi' \gets \max \{a: \sigma_a > \eta^2 \}$
		\State $\v{x}, \mathsf{U}  \gets [{x}_a]_{1\le a\le \chi'},  [\mathsf{U}_{ab}]_{1 \le a\le \chi', 1\le b\le \chi }$ \\
	    \State $L \gets \begin{pmatrix}
			1 & \v{x} & 0 \\ & \mathsf{U} & 0 \\ && 1
	    \end{pmatrix}, L' \gets  \begin{pmatrix}
			1 & -\v{x} & 0 \\ & \mathsf{U}^\dagger & 0 \\ && 1
	    \end{pmatrix}$
	    \State \textbf{return} $ L \W L', L	     $
	    \EndProcedure
	  \Procedure{LeftCan}{$\W, \eta$}  \Comment{$\eta$: tolerance}
	    \State $\W, L_0 \gets \textsc{PreCondition}(\W, \eta)$
	    \State $\W, L_1 \gets \textsc{QRIter}(\W, \eta)$ \Comment{Alg. ~\ref{alg:left_can_alg_infinite_iterated_QR}}
	    \State \textbf{return} $\W, L_1 L_0$ 
	  \EndProcedure
	\end{algorithmic}
\end{algorithm}
\end{figure}
The main drawback of Algorithm~\ref{alg:general_left_canonical_form} is its efficiency: the preconditioning routine involves two eigenvalue problems: finding the leading eigenvector $X$, and (almost) diagonalizing it. It is often more expensive than the QR iteration itself. This brings us to a natural question: why couldn't we prove the existence of left canonical form for all first degree iMPOs (Prop.~\ref{prop:leftcan_exist}) directly using QR iteration? After all, the rank-revealing QR can also reduce bond dimension and potentially serve the r\^ole of the preconditioning step. The answer, unfortunately, is that there are first degree iMPOs for which the QR iteration fails. 

\begin{eg} \label{eg:troll} Consider the spin-half iMPO
\begin{equation}
    \W := \begin{pmatrix}
    \1 & 0 & \widehat{Z} \\
       & \alpha \widehat{Z} & \widehat{X} \\
       &               & \1
    \end{pmatrix}  \,,
\end{equation}
where $|\alpha| < 1$ so that $\W$ is first degree. But applying Algorithm~\ref{alg:left_can_alg_infinite_iterated_QR} to it will yield 
\begin{equation}
    \W_n  = \begin{pmatrix}
    \1 & 0 & \widehat{Z} \\
       & \alpha \widehat{Z} & \alpha^n \widehat{X} \\
       &               & \1
    \end{pmatrix},
	\quad L_n = \begin{pmatrix}
     1 & 0 & 0 \\
       & \alpha^n  & 0 \\
       &  0                & 1 \\
    \end{pmatrix}  \,.
\end{equation}
Everything seems to converge, but $\lim_{n\to\infty}\W_n$ is not left canonical! In fact, $\lim_{n\to\infty} L_n$ is singular, which makes the argument in the proof of Prop.~\ref{prop:QRconverge} inapplicable. The origin of this failure is that, the middle state of the state machine is not reachable from the initial state, so the middle row and column can be removed altogether. (This is precisely what the \textsc{Precondition} routine in Algorithm~\ref{alg:general_left_canonical_form} does.) But the rank-revealing QR fails to detect this, because $\W$ has full column rank. 
\end{eg}

We close this appendix by noting that the above theory for the convergence of QR iteration can be improved. Indeed the assumption of Prop.~\ref{prop:QRconverge} can be certainly relaxed. It will be interesting to find a sufficient and necessary condition of convergence, and improve the efficiency of the preconditioning step.

\section{Exact estimates of Schmidt values}\label{app:schmidt}
We study the singular values of the matrix $M$ defined in \eqref{eq:Schmidt_for_MPO} (which form the entanglement spectrum of an MPO) by repeatedly applying a rank one perturbation. 

First, we consider the sub-matrix 
\begin{equation}
    M_0 := \begin{pmatrix}
		\mathcal{N}_R & \v{p}_R  \\
		0 & \mathsf{S} 
	\end{pmatrix} \,,
\end{equation}
where $\mathsf{S} = \diag(s_1 \ge \dots \ge s_\chi)$ so that 
\begin{equation}
  M_0^\dagger M_0    = 
     \begin{pmatrix}
		0 & 0 \\
		0 & \mathsf{S}^2   
	\end{pmatrix} + 
	 \begin{pmatrix} \mathcal{N}_R \\ \v{p}_R^\dagger \end{pmatrix} \begin{pmatrix} \mathcal{N}_R & \v{p}_R \end{pmatrix}
\end{equation}
is a rank one perturbation of $\mathrm{diag}(0, s_1^2, s_2^2, \dots )$.  A standard result then shows that the singular values of $M_0$, denoted $\mu_0 \ge \mu_1 \ge \mu_2 \ge \dots \mu_\chi$, are given by the positive roots of the equation 
\begin{equation}
    \frac{\mathcal{N}_R^2}{\mu^2} + 
    \sum_a \frac{|p_R^a|^2}{\mu^2 - s_a^2} = 1 \,. \label{eq:rank1}
\end{equation}
This implies the interlacing relation 
\begin{equation}
    \mu_0 \ge s_1 \ge \mu_1 \ge s_2 \ge  \dots \ge s_\chi \ge \mu_\chi.
    \label{eq:interlacing0}
\end{equation}
For the largest singular value, \eqref{eq:rank1} further implies 
$$ \frac{\mathcal{N}_R^2}{\mu_0^2} + 
    \sum_a \frac{|p_R^a|^2}{\mu_0^2} \le 1 \le   \frac{\mathcal{N}_R^2}{\mu_0^2-s_1^2} + 
    \sum_a \frac{|p_R^a|^2}{\mu_0^2-s_1^2} \,, $$
leading to the following estimates:
\begin{equation}
  \mathcal{N}_R^2 + \dn{\v{p}_R}^2 + s_1^2 \ge  \mu_0^2 \ge \mathcal{N}_R^2 + \dn{\v{p}_R}^2 \,. \label{eq:mu0_esti}
  \end{equation}
 In particular, the separation of scales \eqref{eq:scales} implies $\mu_0^2 = \mathrm{\Theta}(N)$ and $\mu_{a\ge 1}^2 = \mathrm{O}(1) $.

In a very similar fashion, we now go back to the full matrix and consider 
\begin{equation}
    M M^\dagger = 
    \begin{pmatrix}
    M_0 M_0^\dagger & \\ & 0
    \end{pmatrix} + \begin{pmatrix} 0 \\ \v{p}_L  \\ \mathcal{N}_L  \end{pmatrix} \begin{pmatrix} 0 & \v{p}_L^\dagger  & \mathcal{N}_L  \end{pmatrix}
\end{equation}
which is similar to 
\begin{equation}
  \begin{pmatrix}
   \mu_0^2 && \\  & D_{\mu} & \\ & & 0
    \end{pmatrix} + \begin{pmatrix} \v{q}_L  \\ \mathcal{N}_L  \end{pmatrix} \begin{pmatrix} \v{q}_L^\dagger & \mathcal{N}_L  \end{pmatrix} \,, \label{eq:M_similar}
\end{equation}
under conjugation where $D_\mu = \mathrm{diag}(\mu_1^2, \dots, \mu_\chi^2)$,  $\v{q}_L= U (0 \; \v{p}_L)^T $, $U$ being a unitary matrix such that $U M_0 M_0^\dagger U^{\dagger} = \mathrm{diag}(\mu_0^2, \mu_1^2, \dots, \mu_\chi^2) $. Applying rank one perturbation again to \eqref{eq:M_similar}, we obtain the following equation determining the singular values of $M$:
\begin{equation}
    \frac{\mathcal{N}_L^2}{\lambda^2} + 
     \frac{|q_L^0|^2}{\lambda^2 - \mu_0^2} + 
    \sum_{a=1}^{\chi} \frac{|q_L^a|^2}{\lambda^2 - \mu_a^2} = 1  \,. \label{eq:rank1_2}
\end{equation}
This implies the interlacing relation 
\begin{equation}
    \lambda_{-1} \ge \mu_0 \ge \lambda_0 \ge \mu_1 \ge  \dots \ge \mu_\chi \ge \lambda_\chi  \,,
\end{equation}
which, combined with \eqref{eq:interlacing0}, gives \eqref{eq:interlacing} in the main text. 

Similarly to \eqref{eq:mu0_esti}, we can bound $\lambda_{-1}$ as follows:
\begin{subequations}
\begin{align}
 \lambda_{-1}^2 &\ge 
  \mathcal{N}_L^2 + \dn{\v{q}_L}^2 =  \mathcal{N}_L^2 + \dn{\v{p}_L^2} \\
  \lambda_{-1}^2 &\le 
   \mathcal{N}_L^2 + \dn{\v{p}_L^2}  + \mu_0^2 \,.
\end{align}
   \label{eq:lambda-1}%
   \end{subequations}
Under the separation of scales~\eqref{eq:scales}, $\lambda_{-1} = \mathrm{\Theta}(N)$ is extensive. 

We also need a useful lower bound for largest singular value $\lambda_0$. For this, we note that \eqref{eq:rank1_2} implies
\begin{equation}
    \frac{\mathcal{N}_L^2}{\lambda_{0}^2} +  \sum_{a=1}^{\chi} 
    \frac{|q_L^a|^2}{\lambda_{0}^2} \le 1 + \frac{|q_L^0|^2}{\mu_0^2 - \lambda_0^2}
\end{equation}
which is a quadratic inequality (of $\lambda_0^2$). Its solution entails
\begin{align}
   2 \lambda_0^2  &\ge\,  \mu_0^2 + \mathcal{N}_L^2+ 
     \dn{\v{q}_L}^2 - \\ &\sqrt{(\mu_0^2 - \mathcal{N}_L^2 - 
     \dn{\v{q}_L}^2)^2 + 4 \mu_0^2  |q_L^0|^2}  \nonumber \\
     &\ge\, 2 \min ( \mu_0^2, \mathcal{N}_L^2 + \dn{\v{q}_L}^2) - 2 \mu_0 |q_L^0| \,.
\end{align}
Now, under~\eqref{eq:scales}, $\mu_0^2, \mathcal{N}_L^2 \in \mathrm{\Theta}(N)$ and $\v{q}_L \in \mathrm{O}(1)$, so we conclude that $\lambda_0^2 \in \mathrm{\Theta}(N)$ is also extensive. 

\section{Hamiltonian Error Bound}
\label{app:ham_error_bound}
This Appendix discusses the relation between the norm we have used throughout this work and the standard sup norm.

Recall that the sup norm is
\begin{equation}		
\dn{\widehat{H}}_s^2 := \sup_{\ket{\psi}} \frac{\braket{\psi|\widehat{H}\widehat{H}|\psi}}{\braket{\psi|\psi}}.
\end{equation}
For convenience, we work in this section with the non-scaled Frobenius norm, which we denote with a lowercase `$f$':
	\begin{equation}
		\dn{\widehat{H}}_f^2 := \Tr[\widehat{H}^\dagger \widehat{H}]  = \Tr[I] \cdot \dn{\widehat{H}}^2_F
	\end{equation}	
	With our default inner product, if $\hat{O}$ is an operator supported on $S\subset \Z$, a set of size $\n{S} = k$, then
\begin{equation}
	\braket{\hat{O}|\hat{O}} = \frac{\Tr[\hat{O}^\dagger \hat{O}]}{\Tr[I]} = \frac{\Tr[\widehat{O}_S^\dagger \widehat{O}_S]}{\Tr[\1^{\otimes k}]} = \frac{\dn{\widehat{O}}_f^2}{d^k}. 
\end{equation}
Quite generally, $\dn{\widehat{O}}_s \le \dn{\widehat{O}}_f$. So $\braket{\widehat{O}\big|\widehat{O}} =1$ implies
\begin{equation}
	\dn{\widehat{O}}_s^2 \le \dn{\widehat{O}}_f^2 = d^k
	\label{eq:locality_norm_bound}
\end{equation}
for an operator supported on $k$ sites.

We now prove Prop. \ref{prop:GS_error_bound}.  The idea is that each term in the Hamiltonian, being local, can only change the ground state energy slightly. The total change in the energy is then bounded above by the number of terms times the size of each term, which, we know to be small since they have small singular values. One could prove analagous bounds broader classes of Hamiltonians, such as long-range interactions, but this might require a significant amount of ``technology" to specify the class of operators under discussion. Nevertheless, we expect the essential point to remain unchanged: for Hamiltonian-class operators, the change in the ground state energy is $O(1)$ times the weight of the truncated singular values.

\begin{proof}[Proof of Pro.p \ref{prop:GS_error_bound}]

We will assume that $E_0 = \dn{\widehat{H}}_2$ is the ground state of the Hamiltonian, though in principle it could also be the ground state of $-\widehat{H}$. As each term is unique, the operators on the right and left sides are both orthonormal:
\begin{equation}
	\braket{\widehat{H}_S^a|\widehat{H}_S^b} = \delta^{ab}, S \in \st{L,R}.
\end{equation}
%As $U$ and $V$ are both unitary, it follows that
%\begin{equation}
%	\braket{\widehat{O}_S^a|\widehat{O}_S^b} = \delta^{ab}, S \in \st{L,R}.
%\end{equation}
As each term is supported on at most $k$ sites, it follows from \eqref{eq:locality_norm_bound} that
\begin{equation}
	\dn{\widehat{H}_L^a \widehat{H}_R^b}_f^2 = d^k.
\end{equation}

It is a standard fact about extremal eigenvalues that if $\widehat{H} = \widehat{H}' + \delta \widehat{H}$ and $E_0' := \dn{\widehat{H}'}_s$, then
\begin{equation}
	\delta E := \n{E_0' - E_0} \le \dn{\delta \widehat{H}}_s.
\end{equation}
By hypothesis
\begin{equation}
	\dn{\delta \widehat{H}}^2_s 
	= \dn{\sum_{a=\chi'}^\chi \widehat{O}_L^a s_a \widehat{O}_R^a}_s^2
	\le 
	\sum_{a=\chi'}^\chi s_a^2 \dn{\widehat{O}_L^a \widehat{O}_R^a}_s^2  
\end{equation}
We can now separately bound each term in the sum using locality:
\begin{align*}
		&\dn{\widehat{O}_L^c \widehat{O}_R^c}_s^2\\
		& \le \dn{\widehat{O}_L^c \widehat{O}_R^c}_f^2\\
		& = \dn{\sum_{a,b=1}^\chi U^{ac} V^{cb} \widehat{H}_L^a \widehat{H}_R^b}_f^2 \\	
		& = \sum_{abef} U^{ac} V^{cb} (U^{ec})^* (V^{cf})^* \Tr[\widehat{H}_L^{a\dagger} \widehat{H}_R^{b\dagger} \widehat{H}_L^e \widehat{H}_R^f]\\
		& = \sum_{abef} U^{ac} V^{cb} (U^{ec})^* (V^{cf})^* d^k \delta^{ae} \delta^{bf}\\ 
		& = d^k\sum_{a=1}^\chi \n{U^{ac}}^2 \sum_{b=1}^\chi \n{V^{cb}}^2\\
		& = d^k,
\end{align*}
where we have used \eqref{eq:locality_norm_bound} several times and the last two equalities follow from orthogonality of the $H$'s and orthogonality of the columns and rows of $U$ and $V$, respectively.

Combining our inequalities, we have
\begin{equation}
	\dn{\delta \widehat{H}}_s^2 \le \sum_{a=\chi'}^\chi s_a^2 \dn{\widehat{O}_L^a \widehat{O}_R^a}_s^2 \le d^k \sum_{a =\chi'}^{\chi} s_a^2.
\end{equation}

\end{proof}

\section{Elementary operations}
	\label{app:elementary_operations}
	%scalar mult
	%addition
	%commutator
	%multiplication -> leaves the algebra, but needed for <(H-E)^2>.
	% non-disjoint multiplication

	This Appendix discusses how to perform the standard algebraic operations --- scalar multiplication, addition, multiplication, and commutation --- for local MPOs. These are standard operations and are discussed in various places in the literature, but we review them here for completeness.

Suppose below that $\lambda \in \R$ is a scalar and operators $\O_1$ and $\O_2$ are represented by iMPOs
\begin{equation}
	\W[\O_1] = \begin{pmatrix}
		\1 & \CC_1 & \DD_1\\
		0 & \AA_1 & \BB_1\\
		0 & 0 & \1\\
	\end{pmatrix}
	,
	\W[\O_2] = \begin{pmatrix}
		\1 & \CC_2 & \DD_2\\
		0 & \AA_2 & \BB_2\\
		0 & 0 & \1
	\end{pmatrix}
	\label{eq:example_iMPOs}
\end{equation}
respectively with finite-automata as follows.

\begin{center}
	\includegraphics{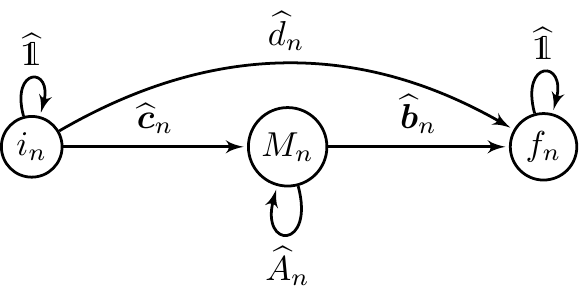}
\end{center}
Here $(i_n,M_n, f_n), n = 1,2$ stand for the initial state, the $\chi$ middle states, and the final state. 

The scalar product is straightforward: each term needs to be scaled exactly once as it moves through the automata. This can be done by scaling all the edges that are incident to the final (or initial) state.

\begin{center}
\includegraphics{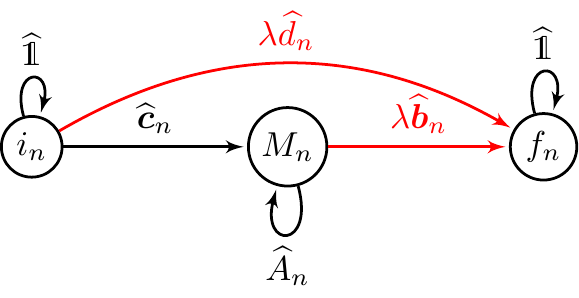}
\end{center}
At the matrix level:
\begin{equation}
	\W[\lambda \O_1] 
	= \begin{pmatrix}
		\1 & \CC_1 & \lambda \DD_1\\
		0 & \AA_1 & \lambda \BB_1\\
		0 & 0 & \1\\
	\end{pmatrix}
	= \begin{pmatrix}
		\1 & \lambda\CC_1 & \lambda \DD_1\\
		0 & \AA_1 &  \BB_1\\
		0 & 0 & \1\\
	\end{pmatrix}.
	\label{eq:scalar_mult_of_MPOs}
\end{equation}
These two choices preserve left and right canonical forms respectively.

Addition of iMPOs is essentially the direct sum of the matrices:
\begin{equation}
	\W[\O_1 + \O_2]
	= 
	\begin{pmatrix}
		\1 & \CC_1 & \CC_2 & \DD_1 + \DD_2\\
		0 & \AA_1 & 0 & \BB_1\\
		0 & 0 & \AA_2 & \BB_2\\
		0 & 0 & 0 & \1
	\end{pmatrix}.
	\label{eq:sum_of_MPOs}
\end{equation}

The operation of multiplication is more involved. The multiplication of two local operators, say $\O_1 = \sum_i \widehat{X}_i$ and $\O_2 = \sum_i \widehat{Y}_i$ is ``bi-local'' or ``second degree", with arbitrarily long strings of identities between sites with information: $\O_1 \O_2 = \sum_i \sum_{N=0}^\infty  \widehat{X}_i \1^N \widehat{Y}_{i+N} + \cdots$. This is represented as an iMPO as 
\begin{equation}
\W[O_1 O_2]
	= 
	\begin{pmatrix}
	\1 & \widehat{X} & \widehat{Y} & i \widehat{Z}\\
	0 & \1 & 0 & \widehat{Y}\\
	0 & 0 & \1 & \widehat{X}\\
	0 & 0 & 0 & \1
	\end{pmatrix}.
 \end{equation}
 The $\1$'s on the diagonal are an unavoidable consequence of being ``second degree'': $\W[O_1 O_2]$ norm $\propto N^2$ in a system of size $N$.
 
 It is insightful to look at the generic ``product automata''.
\begin{center}
	\includegraphics[scale=.9]{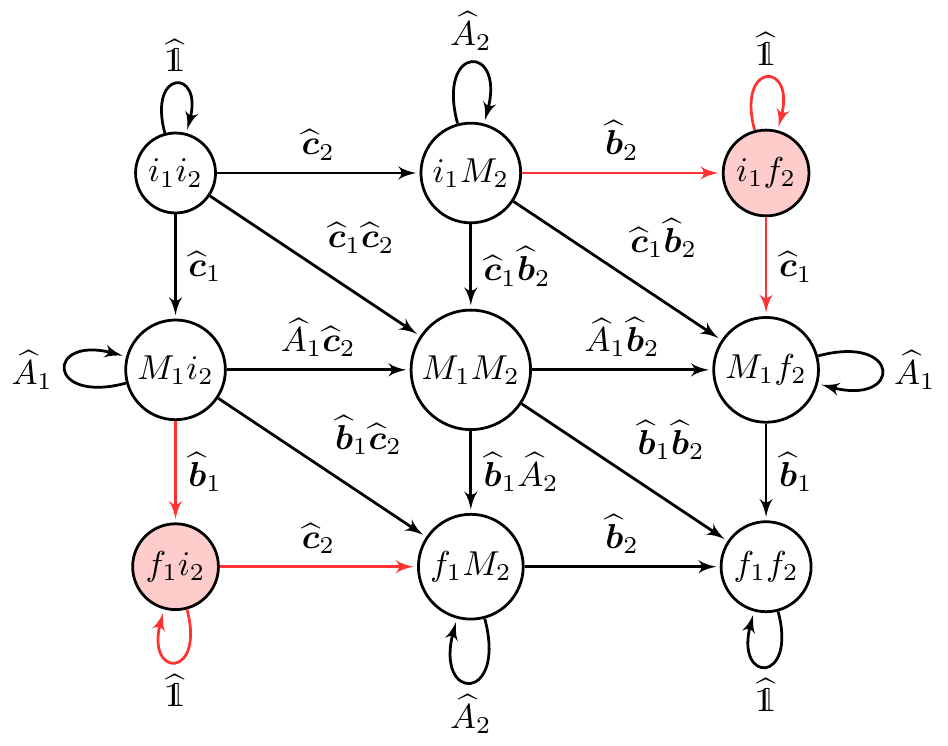}
\end{center}
(We have dropped the $\DD$ terms and also the self-loop on $M_1M_2$ for clarity.) One should interpret the products on edges as the tensor products of the ancilla space but products in the physical space. For example, ``$\BB_1 \AA_2$'' has components
\begin{equation}
	\left( \BB_1 \AA_2 \right)^\gamma_{(a_1 a_2), b_2} = \sum_{\alpha,\beta} f_{\alpha \beta}^\gamma	(B_1)_{a_1}^\alpha \left( A_2 \right)_{a_2,b_2}^\beta
\end{equation}
where $f_{\alpha\beta}^\gamma$ are the structure constants of the on-site algebra $\A$. 

The non-locality of the product comes only from the shaded parts of the automata. What if we were to simply remove the troublesome parts? This motivates a definition. 
	\begin{defn}
		Suppose $\O_1$ and $\O_2$ are two strings of single site operators (Pauli strings in the spin-$1/2$ case) with support on sites $[a_1, b_1]$ and $[a_2,b_2]$ respectively. The \textbf{non-disjoint product} is
		\begin{equation}
			\O_1 \odot \O_2 = \begin{cases}
				\O_1 \O_2 & \text{ if } [a_1,b_1] \cap [a_2,b_2] \neq \emptyset\\
				0 & \text{ otherwise.}
			\end{cases}
			\label{eq:non-disjoint_product}
		\end{equation}
The definition extends to any local operators by linearity. At the MPO level, this is just the non-shaded part of the above diagram.
	\end{defn}

	Terms with disjoint spatial support always commute, so the ``non-disjoint commutator'' is the same as the normal one:
	\begin{equation}
	[\O_1,\O_2] = \O_1 \odot \O_2 - \O_2 \odot \O_1.
	\label{eq:non-disjoint_commutator}
\end{equation}
This means that the commutator is local whenever $\O_1 \odot \O_2$ is. Therefore strictly local operators form a closed algebra under commutation.

First degree operators are not closed under commutation, as we now demonstrate by counterexample. This is a consequence of the fact that the class of first degree operators includes operators which do not make sense as local or physical Hamiltonians. These can have bizarre properties from a ground state perspective, such as superextensive ground state energy, which lead in turn to other strange issues such as the non-closure under commutation. The subset of first degree operators which \textit{are} physical Hamiltonians should be free of these issues. 

For the counterexample, suppose $\widehat{H}_l$ has an iMPO representation
\begin{equation}
	\widehat{W}_l = \begin{pmatrix}
		\1 & \widehat{X} & 0\\
		0 & {\O} & \widehat{Y}\\ 
		0 & 0 & \1
	\end{pmatrix}
\end{equation}
where $\O = \frac{c}{2} \left( \1 + \widehat{Z} \right) = \begin{pmatrix}
	c & 0\\
	0 & 0
\end{pmatrix}$ is an on-site projector matrix and take $c \in (2^{1/4}, 2^{1/2})$. The norm of $H_l$ is $\dn{H_l}^2 = \sum_{N=0}^\infty \dn{{\O}}^{2N} = \sum_{N=0}^\infty (c^2/2)^N < \infty$. However, the norm of the product diverges:
\begin{equation}
	\dn{H_l \odot H_l}^2 > \sum_{N=0}^\infty \dn{{\O} {\O}}^{2N} = \sum_{N=0}^\infty \left( c^4/2 \right)^N = \infty,
\end{equation}
since $c > 2^{1/4}$. The divergent terms here are not from the diagonal ones but from an eigenvalue $c^4/2 > 1$ of $T_A$. So not only can the product of two first degree iMPOs be strictly non-local, but the norm-per-unit-length is not even submultiplicative: there are cases where $\dn{{\O}_1 {\O}_2} \not\leq \dn{{\O}_1} \dn{{\O}_2}$. It would be interesting to find the largest closed subalgebra of the first degree operators.

Thankfully, the commutator of a first degree operator with a strictly-local operator is well-controlled, which is what enables us to perform the Lanczos algorithm within first degree operator, so long as the Hamiltonian is strictly local --- the most physically relevant case.

\begin{prop}
	If $\O_1$ is strictly local and $\O_2$ is first degree, then $[\O_1, \O_2]$ is first degree.
\end{prop}
\begin{proof}
It is sufficient to show $\O_1 \odot \O_2$ is first degree.

	Let the iMPOs for the operators be given by Eq. \eqref{eq:example_iMPOs}. In particular, $\AA_1$ is strictly upper triangular. From the product automata above, we can see that the $\AA$ block of $\O_1 \odot \O_2$ is given by
	\begin{equation}
		\AA = 
		\begin{pmatrix}
		\AA_2 & 0 & \CC_1 \AA_2 & \CC_1 \BB_2 & 0\\
		0 & \AA_1 & \AA_1 \CC_2 & 0 & \BB_1 \CC_2\\
		0 & 0 & \AA_1 \AA_2 & \AA_1 \BB_2 & \BB_1 \AA_2\\
		0 & 0 & 0 & \AA_1 & 0 \\
		0 & 0 & 0 & 0  & \AA_2
	\end{pmatrix},
	\label{eq:product_A_subblock}
	\end{equation}
	where ``multiplications'' such as $\AA_1 \AA_2$ again stands for the tensor product in ancilla indices and multiplication in the physical indices. This is block-upper triangular, so the transfer matrix $T_A$ is also block upper triangular, and it's spectrum is the union of the spectra of the transfer matrices of the diagonal blocks of $\AA$. Since $\AA_1$ and $\AA_1 \AA_2$ are upper triangular with zeros on the diagonal, the maximal eigenvalue of their transfer matrices is also zero. Since $\AA_2$ is first degree, the maximal eigenvalue of its transfer matrix is some $\lambda<1$, so the maximal eigenvalue of $T_A$ is also $\lambda$. This completes the proof.
\end{proof}

As a practical matter, then, one should compute the commutator of two MPOs via Eq. \eqref{eq:non-disjoint_commutator}. It is advisible to compress the operator after each product and again after the difference. In circumstances where $\O_1$ and $\O_2$ are Hermitian or anti-Hermitian, the two non-disjoint products are related by a Hermitian conjugate and a sign, and need to be computed only once.

	\bibliography{references}

	\end{document}